\documentclass[11pt, a4paper]{article}
\usepackage[utf8]{inputenc}
\usepackage{amsmath}
\usepackage{amsthm}
\usepackage{bm}
\usepackage{mathtools}
\usepackage{braket}
\usepackage{amssymb}
\usepackage{geometry}
\usepackage{csquotes}
\usepackage{slashed}
\usepackage{tensor}
\usepackage{fancyhdr}
\usepackage{comment}
\usepackage{wrapfig}
\usepackage[toc,page]{appendix}
\usepackage[english]{babel}
\usepackage{hyperref}
\usepackage{graphicx}
\usepackage[export]{adjustbox}
\usepackage[disable]{todonotes}
\usepackage{xcolor}
\usepackage{caption}
\usepackage{subcaption}
\usepackage{authblk}
\usepackage{tikz}
\usetikzlibrary{shapes,positioning,automata}
\usetikzlibrary{arrows.meta,arrows}

\newtheorem{lemma}{Lemma}

\numberwithin{equation}{section}

\usepackage[
backend=biber,sorting=none,doi=false,isbn=false,url=false
]{biblatex}

\allowdisplaybreaks 
\graphicspath{ {./wa_figures/} }
\def\cN{\mathcal{N}} 
\def\bea{\begin{eqnarray}}
\def\eea{\end{eqnarray}} 
\def\Dim{{\rm{Dim}}}
\def\tr{ {\rm{tr}} }
\def\mC{ \mathbb{C} } 
\def\cZ{ \mathcal{Z} } 
\def\cO{ \mathcal{O} } 
\def\cA{ \mathcal{A} } 
\def\cL{ \mathcal{L} }

\begin{document}


\begin{flushright}
QMUL-PH-22-33\\
\end{flushright}

\bigskip
\bigskip 
\bigskip
\bigskip
\bigskip 
\bigskip

\centerline{\huge \bf Matrix and tensor witnesses of}
\vspace{.2cm}
\centerline{\huge \bf hidden symmetry algebras  }

\bigskip
\bigskip

\centerline{\bf    Sanjaye Ramgoolam ${}^{a,b, \dagger}$  and  Lewis Sword ${}^{a,*}$  }

\bigskip

\begin{center}

\small{
{ \it ${}^{a}$ Centre for Theoretical Physics, Department of Physics and Astronomy \\
Queen Mary University of London, 327 Mile End Road, London E1 4NS, UK}}

    \vspace{.2cm}
\centerline{{\it ${}^b$ National Institute for Theoretical Physics, School of Physics and Centre for Theoretical Physics, }}
\centerline{{\it University of the Witwatersrand, Wits, 2050, South Africa } }

\medskip 
{\it Email:} ${}^{\dagger}$ \texttt{s.ramgoolam@qmul.ac.uk} ,  ${}^{*}$ \texttt{l.sword@qmul.ac.uk} \quad 

\end{center}

\bigskip 
\bigskip
\centerline{ \it \today}

\begin{abstract}
Permutation group algebras, and their generalizations called permutation centralizer algebras (PCAs),  play a central role as hidden symmetries in the combinatorics of large $N$ gauge theories and matrix models  with manifest continuous gauge symmetries. Polynomial functions invariant under the manifest symmetries are the observables of interest and have applications in AdS/CFT. 
We compute such correlators in the presence of matrix or tensor witnesses, which by definition, can include  a matrix or tensor field appearing as a coupling in the action (i.e a spurion)  or as a classical (un-integrated)  field in the observables, appearing alongside quantum (integrated) fields. In both matrix and tensor cases we find that two-point correlators of general gauge-invariant observables can be written in terms of gauge invariant functions of the witness fields, with coefficients given by structure constants of the associated PCAs. Fourier transformation on the relevant PCAs, relates combinatorial bases to representation  theoretic bases. The representation theory basis elements  obey orthogonality results for the two-point correlators which generalise known orthogonality relations to the case with witness fields. The new orthogonality equations involve two representation basis elements for observables as input and a representation basis observable constructed purely from witness fields as the output. These equations extend known equations in the  super-integrability programme initiated by Mironov  and Morozov, and are a direct physical realization of the Wedderburn-Artin decompositions of the hidden permutation  centralizer algebras of matrix/tensor models. 
\end{abstract}

 
\newpage

\tableofcontents

\newpage


\section{Introduction}
\label{s:introduction}
Matrix models  have played a prominent role in gauge-string duality, starting from the duality between matrix models with double scaling limits and low-dimensional non-critical strings  \cite{BrezKaz1990,DougShenk1990,GrossMig1990}. 
They also arise in the description of BPS states and their correlators in $\cN=4$ super-Yang Mills theories. They are therefore important for the AdS/CFT duality between $\cN=4$ super-Yang Mills (SYM)  theories and string theory on $AdS_5 \times S^5$ \cite{Maldacena:1997re,Witten:1998qj,Gubser:1998bc}. In particular, complex matrix models with one complex variable matrix $Z$, and the $U(N)$ invariant polynomial functions of $Z$,  are relevant for the half-BPS sector of $\cN=4$ super-Yang Mills (SYM) with $U(N)$ gauge group. AdS/CFT brings into focus the classification of the invariant polynomial functions and the computation of their correlators. It was observed in  \cite{BBNS} that the breakdown of the standard map between multi-traces  in SYM and multi-particle states in AdS, for large dimension operators,  is related to a failure of orthogonality of multi-traces in the inner product constructed from 2-point functions of holomorphic and anti-holomorphic gauge invariant operators. Sub-determinant operators were also proposed as SYM duals of giant gravitons extended in the  $S^5$.  Orthogonal bases were thus recognised as important for the identification of bulk space-time duals for the quantum states in the CFT corresponding to large dimension local operators. This led to the classification of half-BPS operators in $U(N)$ $\cN=4$ SYM theory with dimension $n$  in terms of operators $\cO_R ( Z ) $ labelled by Young diagrams $R$ having $n$ boxes \cite{Corley:2001zk}.  These operators were shown to be orthogonal in the free-field inner product, and this informed a proposal  for a general map between half-BPS gauge invariant operators in SYM to giant gravitons in the string theory, distinguishing giant gravitons extended in the $S^5$ from those extended in the $ AdS_5$ giants, and single-giant states from multi-giant states in terms of the Young diagrams. This proposal has passed non-trivial checks based on the calculation of 3-point functions of large dimension operators in SYM  and comparison with calculations in the $AdS_5 \times S^5$ (see \cite{Bissi:2011dc,Caputa:2012yj,Lin:2012ey,Kristjansen:2015gpa,Jiang:2019xdz,Yang:2021kot,Chen:2019gsb,Holguin:2022zii}).  The  free-field inner product is equivalently expressible in terms of a Gaussian matrix model correlator. The precise form of $\cO_R ( Z ) \equiv \chi_R ( Z )  $ will be recalled in \eqref{eqn:spb_complex_character_schur} of section \ref{ss:single_matrix_schur_basis}: they are linear combinations of multi-trace operators weighted by characters of the symmetric group $S_n$.  
 The inner product was  calculated as 
\bea\label{schur-orthog} 
\langle \cO_R ( Z ) \cO_S ( Z^{ \dagger} ) \rangle  = \delta_{ RS } { n! \Dim_N ( R ) \over  d_R } 
\eea
where $\Dim_N R $ is the dimension of the $U(N)$ irreducible representation associated with Young diagram $R$. This object is referred to throughout as the two-point correlator. The reason for this nomenclature stems from the fact that the matrices $Z$ and $Z^{\dagger}$ are inserted at a given point in spacetime, and the convention followed in this paper chooses to suppress this spacetime dependence. In section \ref{ss:single_matrix_schur_basis} we consider a modification of the standard Gaussian matrix model integral by introducing a matrix coupling $A$. The partition function is modified as: 
\bea 
\int [dZ] e^{-\text{Tr}(ZZ^{\dagger})} \rightarrow \int [dZ] e^{-\text{Tr}(Z A Z^{\dagger})} 
\eea 
Similar modifications of hermitian matrix models have been considered in \cite{Mironov:2022fsr}. Using a subscript $A$ to denote the correlation functions calculated with this modified action, we show that 
\bea\label{schur-orthog-refined}  
\boxed{ 
\langle \cO_R ( Z ) \cO_S ( Z^{ \dagger} ) \rangle_{A} =  \delta_{ RS } { n! \cO_R ( B )  \over  d_R } 
}
\eea
where $ B = A^{-1}$.  The RHS of \eqref{schur-orthog} is recovered from the RHS of \eqref{schur-orthog-refined} by setting $B$ to be equal to a unit matrix. The matrix-coupling dependent result is obtained by generalising the combinatorial and diagrammatic techniques in \cite{Corley:2001zk} and \cite{Corley:2002mj} to the case of the background matrix coupling.

For operators of dimension $n$ obeying  $ n < N$, a familiar basis of operators is given by trace structures. For operators of dimension $n$, the different trace structures can be parameterised by integers $p = \{ p_1 , p_2 , \cdots , p_n \}$ giving the powers of different traces 
\bea 
\cO_{ p } ( Z )  = ( \tr Z )^{ p_1} ( \tr Z^2 )^{ p_2} \cdots ( \tr Z^n )^{ p_n }  
\eea
These integers obey 
\bea 
n  = p_1 + 2p_2 + \cdots + n p_n = \sum_{ i=1 }^n  i p_i 
\eea
i.e, they define a partition of $n$. The linear change of basis  between the $\cO_{p} $ and $\cO_{R}$ is given by 
\bea 
\cO_{ R } ( Z )  = {1 \over n!} \sum_{\sigma \in S_{n}} \chi^{R}(\sigma) ~ \mathcal{O}_{\sigma} (Z)  = { 1 \over n! } \sum_{ p \vdash n }   \chi^R_p ~  \cO_{ p } ( Z ) 
\eea
where $\chi^R_p $ is the character for any permutation $\sigma^{(p)}  $  in the conjugacy class $p$ in $S_n$, $|T_p|$ is the number of group elements in the conjugacy class $p$ and 
\bea 
\cO_{ p }  (Z )  =  |T_{p}|  \mathcal{O}_{\sigma^{(p)}} ( Z )  \, . 
\eea
 For $ n>N$, the operators $\cO_p (Z ) $ form an overcomplete basis, due to finite $N$ trace relations. The description of the finite $N$ state space is simple in terms of the Young diagram basis elements $\cO_R (Z)$. We simply drop all Young diagrams with $l(R) > N$, where $l(R)$ is the number of rows in the Young diagram $R$.  The two-point function \eqref{schur-orthog} reflects this cutoff since the RHS, viewed as a polynomial in $N$, vanishes for $l(R) > N$. 
The two-point function in the trace basis is 
\bea\label{trace-basis-corr}  
\langle \cO_{ p_1} ( Z ) ( \cO_{ p_2} ( Z ))^{ \dagger} \rangle =n! \sum_{ p_3 \vdash n }  C_{ p_1 , p_2 }^{ p_3 } N^{ C_{ p_3} }    
\eea
where $C_{ p_1 , p_2}^{ p_3} $ are the structure constants of the commutative algebra formed by the centre of 
$\cZ ( \mC ( S_n ) ) $. These  are described more explicitly in section  \ref{s:alg_one_mat_bmf}. 

A partition $p$ of $n$ determines a conjugacy class in $S_n$. The formal sum of group elements in the conjugacy class, denoted $T_p$, is an element in the group algebra $\mC ( S_n)$ and commutes with all elements of $\mC ( S_n)$. As $p$ ranges over all partitions of $n$, these $T_p$ form a basis for the centre of $\mC ( S_n)$. The multiplication of two of these central elements can be expanded in terms of the same basis
\bea\label{permbasprod}  
T_{ p_1} T_{ p_2} = \sum_{ p_3} C_{ p_1 , p_2 }^{ p_3} T_{ p_3} 
\eea
where $C_{ p_1 , p_2 }^{p_3} $ are the structure constants of the multiplication in the algebra $ \cZ ( \mC ( S_n) )$. We will show that, with the background matrix coupling $A$, the equation \eqref{trace-basis-corr} 
 admits a simple modification 
\bea\label{permbasCorr2}  
\boxed{ 
\langle \cO_{ p_1} ( Z ) ( \cO_{ p_2} ( Z ))^{ \dagger} \rangle_A  = n! \sum_{ p_3 \vdash n }  C_{ p_1 , p_2 }^{ p_3 } 
\cO_{ p_3} ( B ) 
} 
\eea
This shows that the class algebra of the symmetric group has a direct realization in terms of matrix model correlators in the presence of a background matrix coupling. 
In \eqref{permbasprod} we  have a direct realisation of the same structure constants, purely within the algebra, with no connection to matrix elements. To go from \eqref{permbasprod}  to \eqref{permbasCorr2} we are simply decorating the algebra elements with matrix model quantities : 
\bea\label{algmapmat}  
&& T_{ p_1} \rightarrow \cO_{ p_1} ( Z ) \cr 
&& T_{ p_2} \rightarrow  ( \cO_{ p_2} ( Z ) )^{ \dagger}  \cr 
&& T_{p_3 } \rightarrow \cO_{ p_3} ( B ) 
\eea

The equality of the number of conjugacy classes and irreducible representations of a finite group is a familiar mathematical fact.  A related fact is that there is a change of basis in the centre of the group algebra between a basis labelled by the conjugacy classes and a basis labelled by irreducible representations. In the case of 
$S_n$, we have a basis of (un-normalized) projectors 
\bea 
P_{ R } = { 1 \over n! } \sum_{ p } \chi^R_p T_p 
\eea
The multiplication in this basis is diagonal, i.e. zero unless that two projectors are identical:  
\bea
 P_{R_1}  P_{R_2}  = \delta_{ R_1 R_2 } {  P_{R_1}   \over d_{ R_1} } 
\eea
To obtain \eqref{schur-orthog-refined} we apply the same map \eqref{algmapmat}
\bea 
&& P_{R_1}  \rightarrow \cO_{ P_{R_1} } ( Z ) \equiv \cO_{ R_1} ( Z )  \cr 
&& P_{R_2}  \rightarrow  ( \cO_{ P_{R_2}  } ( Z ) )^{ \dagger} \equiv \cO_{ R_2} ( Z^{\dagger}  )  \cr 
&& P_{ R_3} \rightarrow \cO_{ P_{ R_3} } ( B ) \equiv \cO_{ R_3} ( B )  
\eea
Our derivation of the refined 2-point correlator in the Schur basis will in fact be obtained by first deriving 
\eqref{permbasCorr2}  and then Fourier transforming.  The result \eqref{schur-orthog-refined} may be expected by analogy to similar results in the super-integrability literature. Our derivation shows that this result follows using Fourier transformation on a hidden symmetry algebra underlying the combinatorics of invariant operators, which becomes visible (as in \eqref{permbasCorr2})  in the presence of a background matrix coupling. 

The super-integrability results in \cite{Mironov:2022fsr} also include examples where there is a classical (un-integrated) matrix in the observable, alongside the quantum (integrated) matrix. Emulating these results, 
we also show that an  equation similar to \eqref{schur-orthog-refined} can be obtained by considering gauge invariant operators constructed from classical as well as quantum fields.  The result takes the form 
\bea 
\label{eqn:intro_spbback_correlator_with_background}
\boxed{ \langle \mathcal{O}_{R}(ZA) \left(\mathcal{O}_{S}(ZA) \right)^{\dagger} \rangle = \frac{\delta_{RS} n!}{d_{R}} \mathcal{O}_{S}(B)  \,, }   
\eea
where we now  have $B = A^{ \dagger} A$. This is also related by a Fourier transform on the
 algebra $\cZ ( \mC ( S_n ) ) $ to the fact that the structure constants of this algebra 
 arise in 2-point functions for operators labelled by permutations. In the context of AdS/CFT with $\cN=4$ SYM, where the theory has three complex matrices $X, Y, Z $, the matrix $A$ could be $ Y, X$ which transform in the same way as $Z $ with an adjoint $U(N)$ action, so that the observables $ \cO_{ R } ( Z A ) $ are gauge-invariant under the simultaneous gauge transformation of $ Z $ and $A$.    

Classical matrices which occur as couplings in the action are matrix  spurions. Spurions are coupling constants in quantum field theories  which are promoted to fields (see for example \cite{Penco} for a recent effective field theory discussion). Here the matrix $A$ can be viewed as a matrix spurion in a zero-dimensional quantum field theory.  Since there are results of the same algebraic nature whether we consider classical (un-integrated) matrices as coupling constants in the action or classical matrices appearing inside observables, it is natural to have a single name for both uses of a classical field in a matrix model. We propose to use ``matrix witness'' as a unifying terminology which includes a matrix being used in either role. There is some resonance with the fact that classical objects appear as measuring apparatuses in classical/quantum interactions. We leave the exploration of witnesses, as defined here, for applications in quantum information theory as an intriguing question for the future.

Motivated by the CFT description of open strings attached to giant gravitons, the orthogonality relation \eqref{schur-orthog}  has been generalised to  multi-matrix systems   \cite{Kimura:2007wy,Brown:2007xh,Bhattacharyya:2008rb,Bhattacharyya:2008xy,Kimura:2008ac,Brown:2008ij,Pasukonis:2013ts}. The study of tensor models as combinatorial models of quantum gravity \cite{Ambjorn:1990ge,Sasakura:1990fs,Gross:1991hx,Oriti:2006se, 
Gurau:2010ba} and as models of gauge-string duality \cite{Witten:2016iux}  has also led to tensor generalisations of the orthogonality relation \cite{BenGeloun:2013lim,Diaz:2017kub,deMelloKoch:2017bvv,BenGeloun:2017vwn}. It has been recognised that these orthogonality equations for representation theoretic bases   are related to permutation centralizer algebras \cite{Mattioli:2016eyp}. These are in general non-commutative algebras.   The starting point is that the observables can be constructed by using a set of permutations to parametrize the contractions of upper and lower indices of matrix or tensor fields.  We will refer to these as the parameterising-permutations. There are redundancies in the parameterising permutations, which are themselves described by a smaller permutation   group. It is useful to think of the smaller permutation group as giving a discrete gauge symmetry which acts on the discrete set of parameterising-permutations. Parameterising permutations which are gauge-equivalent give rise to the same gauge invariant polynomial.    In the case of the one-matrix problem with $U(N)$ invariants of degree $n$, 
 \bea 
&\hbox{Parameterising permutations}:   &G =    S_n   \cr 
& \hbox{Gauge permutations } : &H = S_n 
\label{eqn:intro_param_gauge_one_matrix}
 \eea
 For the two-matrix problem with $U(N)$ invariants of degree $m$ in one matrix and degree $n$ in the other 
 \bea 
 & \hbox{ Parameterising permutations  } : & G = S_{ m+n }  \cr 
 & \hbox{ Gauge permutations  } : & H = S_m \times S_n \subset S_{ m+n } 
 \label{eqn:intro_param_gauge_two_matrix}
 \eea
 In both of the above cases, the gauge permutations act by conjugation on $G$.  For the case of a complex 3-index tensor $\Phi$ transforming in 3-fold tensor product $V_N \otimes V_N \otimes V_N$ 
 of the fundamental representation $V_N$ of $U(N)$, the construction of invariants of degree $n$ in $\Phi $ and $\bar \Phi $ can be done by using   
 \bea 
 & \hbox{ Parameterising permutations  } : & G  = S_{n} \times S_n \times S_{n}  \cr 
 & \hbox{ Gauge permutations  } : &   H  =  \text{Diag}(S_n) \times \text{Diag} ( S_n ) 
  \label{eqn:intro_param_gauge_tensor}
 \eea
 $\text{Diag}  (S_n)$ is  the diagonal sub-group of $S_n\times S_n \times S_n $ consisting of permutations  $\gamma \in S_n$ embedded diagonally in $S_n \times S_n \times S_n  $ as  $ ( \gamma , \gamma , \gamma )$. Writing the general parameterising $S_n \times S_n \times S_n $ elements as ordered triples $ ( \sigma_1  , \sigma_2, \sigma_3 ) $ with $ \sigma_1 , \sigma_2, \sigma_3 \in S_n$, the gauge permutations $\gamma_L , \gamma_R \in S_n  $ act as
\bea 
(\gamma_L , \gamma_R ) : ( \sigma_1 , \sigma_2 , \sigma_3) \rightarrow  ( \gamma_L \sigma_1 \gamma_R , \gamma_L \sigma_2 \gamma_R , \gamma_L \sigma_3 \gamma_R ) 
\eea  
 with the $\gamma_L$ acting diagonally on the left and $\gamma_R $ acting diagonally on the right. 
  In all the above cases,  the group action by $H$ organises the set of permutations in $G$ into gauge orbits. Sums over $G$-permutations within an orbit of the $H$-action commute with $H$, and form a {\it combinatorial  basis} of a $(G,H)$ permutation centralizer algebra (PCA). This is defined as the subspace of the group algebra $\mC ( G ) $ which is stabilised by  elements of $ H $. A more detailed and more general account of permutation centralizer algebras is given in \cite{Mattioli:2016eyp}.  The PCA for the 2-matrix models is closely related to Littlewood-Richardson coefficients (reduction multiplicities for irreps of $S_{m+n}$ to the subgroup $S_m \times S_n$), while that for the tensor model the PCA (denoted ${\mathcal{K}}(n)$) is related to Kronecker coefficients (Clebsch-Gordan multiplicities) for tensor products of $S_n$ irreps. These latter algebras have been used to realise the squared Kronecker coefficients as the dimensions of  null spaces of  combinatorially defined integer matrices: the null vectors can be obtained using standard combinatorial  integer matrix algorithms \cite{BenGeloun:2020yau,Geloun:2022kma}. Using an involution on ${\mathcal{K}}(n)$ these remarks are generalised to the Kronecker coefficients themselves. 
  
An important aspect of the PCAs is that there is a change of basis (Fourier transformation)  from the combinatorial basis to a Wedderburn-Artin basis labelled by representation theory data (a collection of Young diagrams and associated representation theoretic multiplicity labels). The Wedderburn-Artin basis shows that the associative algebras are isomorphic to a direct sum of matrix algebras. A general discussion of the Wedderburn-Artin theorem for associative algebras in a form we find accessible for physicists is in \cite{Ram-Chapter1}. 
Explicit equations describing the bases which exhibit the isomorphism to a direct sum of matrix algebras, drawing on the background physics literature,   will be reviewed at the start of sections  \ref{s:alg_two_mat_bmf}, \ref{s:multi_matrix_model}, \ref{s:tensor_model}. A general synopsis of this paper's results is captured in Figure \ref{fig:intro_boxed_maps}.

The paper is organized as follows. Section \ref{s:alg_one_mat_bmf} develops the results outlined in the earlier part of this introduction for the complex 1-matrix model. Section  \ref{s:alg_two_mat_bmf} develops all the analogous results for the complex 2-matrix model. Section \ref{s:multi_matrix_model} generalises this discussion to the complex multi-matrix models. Section \ref{s:tensor_model} describes how the structure constants of the Kronecker PCA arise by using tensor witness fields appearing inside gauge invariant composites of quantum (integrated) and classical (unintegrated) tensor fields.

\begin{figure}
\tikzstyle{block} = [draw, rectangle, minimum height=5em, minimum width=5em]
\tikzstyle{virtual} = [coordinate]
\begin{center}
\begin{tikzpicture}[auto, node distance=8cm]
    \node [block] (upleft) {\begin{tabular}{cc}
							Correlators with \\
							witness fields in \\
							combinatorial basis
							\end{tabular}  };
    \node [block, below of=upleft] (botleft) {\begin{tabular}{cc}
										      Correlators with \\
										      witness fields in\\
											  W.A. basis \\
											 \end{tabular}  };
    \node [block, right of=upleft] (upright) {\begin{tabular}{cc}
										      Structure constants \\
											  of permutation \\
											  centraliser algebras
											 \end{tabular}  };
    \node [block, below of=upright] (botright)    {\begin{tabular}{cc}
										      Representation theoretic \\
										      orthogonality relations \\
											  with witness fields \\
											 \end{tabular}};
    \draw [-{Stealth[length=3mm, width=2mm]}] (upleft) -- node  [name=wickcombo,label={[yshift=0.3mm]\bf{Wick contraction}}] {} (upright);
    \draw [{Stealth[length=3mm, width=2mm]}-{Stealth[length=3mm, width=2mm]}] (upright) -- node [name=FTright] {\bf{Fourier Transform}} (botright);
    \draw [{Stealth[length=3mm, width=2mm]}-{Stealth[length=3mm, width=2mm]}] (botleft) -- node [name=FTleft] {\bf{Fourier Transform}} (upleft);
    \draw [-{Stealth[length=3mm, width=2mm]}] (botleft) -- node [name=wickWA,label={[yshift=0.3mm]\bf{Wick contraction}}] {} (botright);
\end{tikzpicture}
\end{center}
\caption{The diagram demonstrates the core results of the paper. The action of Fourier transforming exchanges combinatorial and Wedderburn-Artin (W.A.) bases, while Wick contractions are used to produce the orthogonality relations of correlators as well as the structure constants of the associated PCAs.}
\label{fig:intro_boxed_maps}
\end{figure}
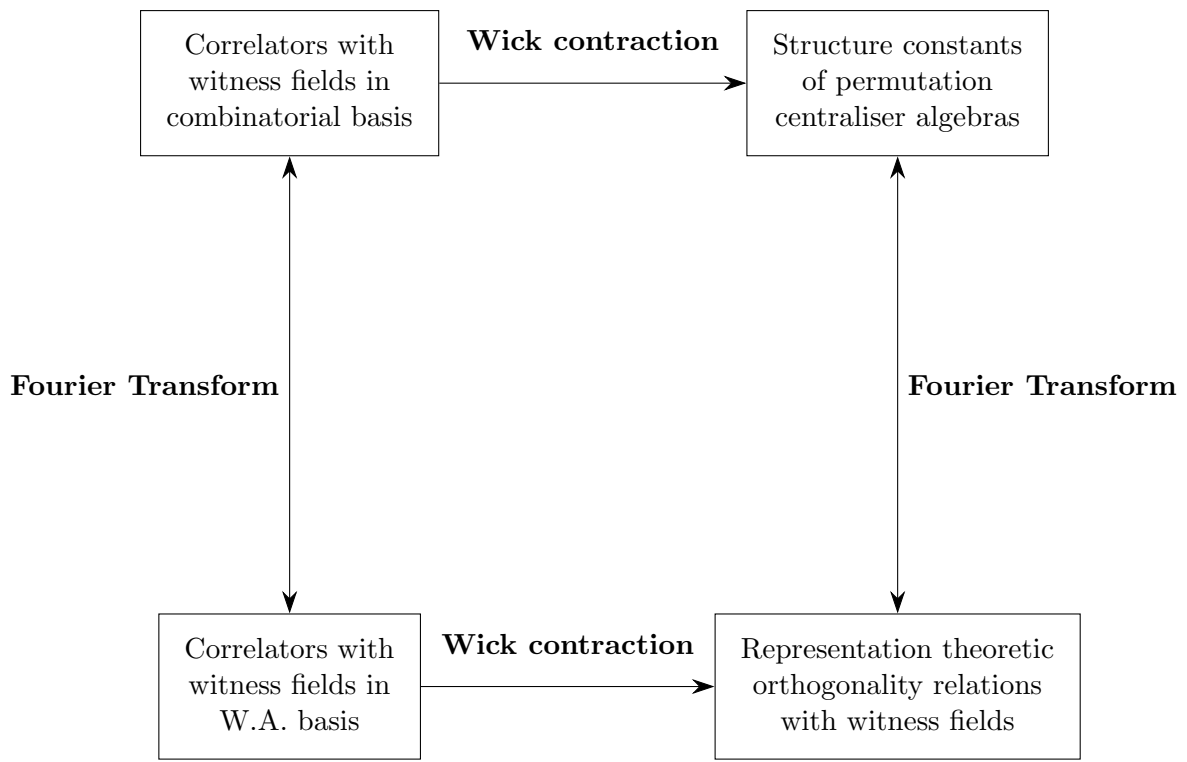

\clearpage
\section{Algebras and single matrix correlators with matrix witnesses }
\label{s:alg_one_mat_bmf}

In this section we consider  the  complex-matrix model of a single complex matrix $Z$ of size $N$, transforming in the adjoint of $U(N)$, which is relevant to the half-BPS sector of $\cN=4$ SYM theory with $U(N)$ gauge group. We first set up the notation following previous work \cite{Corley:2001zk,Corley:2002mj}. The holomorphic gauge invariant functions of degree $n$ are BPS operators with dimension $n$. These are products of traces of powers of $Z$. 
  For the $N \times N$, complex matrix $Z$ and permutation $\sigma \in S_{n}$, we can define  a gauge invariant operator (GIO) as
\begin{equation}
\label{eqn:aomb_gio_definition}
\mathcal{O}_{\sigma}(Z) = Z_{i_{\sigma(1)}}^{i_{1}} \dots Z_{i_{\sigma(n)}}^{i_{n}} = \text{Tr}_{V_{N}^{\otimes n}}(Z^{\otimes n} \cL_{\sigma}) \,.
\end{equation}
The indices $i_1, \cdots , i_n$ are summed over the range $\{ 1 , \cdots , N \}$. The last expression indicates that $\cO_{ \sigma } ( Z ) $ is the trace in the $n$-fold tensor product of the $N$-dimensional fundamental representation $V_N$ of $U(N)$ of the product of two linear operators: $Z^{ \otimes n } \mathcal{L}_{\sigma} $. Here, $\mathcal{L}_{\sigma}$ has the following action on the tensor product of basis vectors
\begin{equation}
\label{eqn:bmf_lin_op_action}
\mathcal{L}_{\sigma} \ket{e_{i_{1}} \otimes \dots \otimes e_{i_{n}}} = \ket{e_{i_{\sigma(1)}} \otimes \dots \otimes e_{i_{\sigma(n)}}} \,.
\end{equation}
It follows straightforwardly that these GIO satisfy
\begin{equation}
\label{eqn:aomb_obs_equiv_rel}
\mathcal{O}_{\sigma}(Z) = \mathcal{O}_{\gamma \sigma \gamma^{-1}}(Z)
\end{equation}
where $\gamma \in S_{n}$. This invariance means that the holomorphic gauge invariants can be described in terms of  equivalence classes of permutations $ \sigma \in S_n$, obtained by conjugation with permutations $\gamma \in S_n $:   
\begin{equation}
\label{eqn:aomb_equiv_rel}
\sigma \sim \gamma \sigma \gamma^{-1}
\end{equation}
Sums of permutations $\sigma $ within an equivalence class belong to the centre $\mathcal{Z}[\mathbb{C}[S_{n}]]$ 
of the group algebra $\mathbb{C}[S_{n}]$. These equivalence classes are the conjugacy classes in $S_n$, associated with partitions $ p $ of $n$, which describe  cycle structures of permutations. We denote as $\mathbf{C}_{p} $, the conjugacy class of permutations associated with partition  $p$, and let $\sigma^{(p)} $ be any chosen permutation in the class  $\mathbf{C}_{p}$. The automorphism group $\text{Aut}_{\mathcal{Z}}\left( \sigma^{ (p) } \right)$, is the subgroup of $S_{n}$ that leaves $\sigma^{(p)}$ invariant under conjugation i.e. the stabiliser subgroup composed of the permutations $\gamma \in S_{n}$, that satisfy $\gamma \sigma^{(p)} \gamma^{-1} = \sigma^{(p)}$. The order of this stabiliser group depends only on $ p$, and not the choice of $ \sigma^{(p)} $, and we denote  this order as 
$| \text{Aut}_{\mathcal{Z}}\left( p \right ) | $. We define central elements $T_p$ labelled by partitions $p$ as follows:  
\begin{equation}
\label{eqn:bmf_zcs_algebra_ele}
T_{p} =  \frac{1}{|\text{Aut}_{\mathcal{Z}}\left(p\right)|} \sum_{\gamma \in S_{n}} \gamma \sigma^{(p)} \gamma^{-1} = \sum_{\alpha \in \mathbf{C}_{p}} \alpha \quad \in \mathcal{Z}[\mathbb{C}[S_{n}]] \,.
\end{equation}  
  As central elements $T_p$ satisfy  $T_{p} = \gamma T_{p} \gamma^{-1}$ for all $ \gamma \in S_n$. The size of the class is denoted by $|\mathbf{C}_{p}| = |T_{p}|$ and it is useful to note  that $|T_{p}| = {  n! \over | \text{Aut}_{\mathcal{Z}}\left( p \right ) | }$. 

The algebra  $\mathcal{Z}[\mathbb{C}[S_{n}]]$  is an example of a permutation centraliser algebra (PCA) as defined in \cite{Mattioli:2016eyp}. As discussed in the introduction, the definition of the PCAs used in this paper involves a set of parameterising permutations and a set of gauge permutations. In this case the parametrising permutations $ \sigma $ and the gauge permutations $\gamma $ are both general permutations in $S_n$ (as indicated in \eqref{eqn:intro_param_gauge_one_matrix}). In  section  \ref{ss:background_matrix_field} we show that the two-point function involving a holomorphic and an anti-holomorphic gauge invariant operator, in the presence of an invertible matrix coupling $A = B^{-1}$,  can be written in terms of the structure constants of the  algebra $\mathcal{Z}[\mathbb{C}[S_{n}]]$. In section \ref{ss:single_matrix_schur_basis} we perform the Fourier transform from the basis in  $\mathcal{Z}[\mathbb{C}[S_{n}]]$  labelled by partitions $p$ to a basis labelled by Young diagrams, to show that the Young-diagram-labelled operators form an orthogonal basis for the two-point functions. In section \ref{ss:back_fields_in_observables} we consider gauge-invariant operators with the matrix $Z$  replaced by the matrix product $ZY$, where $Y$ is another complex matrix (such as exists in $\cN=4$ SYM),  transforming in the adjoint of $U(N)$. We define two-point functions where $Z$ is integrated while $Y$ is left as an unintegrated classical field. 
We show that the results  are given in terms of the structure constants of $\mathcal{Z}[\mathbb{C}[S_{n}]]$. In 
section \ref{ss:schur_basis_back_fields_in_observables} we use Fourier transformation to obtain the orthogonal two-point functions in the Young diagram basis. Thus we have the same structure of results whether we have matrix couplings or matrix classical fields, both of which are referred to as examples of witness fields. The results of this section provide the template which is emulated in subsequent sections, with  $\mathcal{Z}[\mathbb{C}[S_{n}]]$ replaced by appropriate PCAs, which are in general non-commutative.



\subsection{Two-point function of general operators with  matrix coupling }
\label{ss:background_matrix_field}

The partition function of the single matrix model with coupling witness matrix  $A$ is 
\begin{equation}
\label{eqn:bmf_part_func}
\Sigma[0] = \int [dZ] e^{-\text{Tr}(Z A Z^{\dagger})} \,,
\end{equation}
and throughout we take $A$ to be a positive definite Hermitian matrix. In Appendix \ref{app:background_matrix_field_correlator}, we derive the basic correlator\footnote{Note that $\langle Z^{i}_{j} Z^{k}_{l} \rangle = \langle (Z^{\dagger})^{i}_{j} (Z^{\dagger})^{k}_{l} \rangle = 0$, also.} 
\begin{equation}
\label{eqn:bmf_new_correlator_part_func}
\langle Z^{i}_{j} (Z^{\dagger})^{k}_{l} \rangle = \frac{1}{\Sigma[0]}  \int [dZ] Z^{i}_{j} (Z^{\dagger})^{k}_{l} e^{-\text{Tr}(Z A Z^{\dagger})}  =\delta^{i}_{l} (A^{-1})^{k}_{j} = \delta^{i}_{l} B^{k}_{j} \,,
\end{equation}
where we set $B = A^{-1}$ and henceforth refer to $B$ as the coupling witness field. Recalling the permutation parameterisation of the GIO from \eqref{eqn:aomb_gio_definition}
\begin{equation}
\label{eqn:aomb_gio_definition2}
\mathcal{O}_{\sigma}(Z) = \text{Tr}_{V^{\otimes n}_{N}}(Z^{\otimes n} \cL_{\sigma}) = Z^{i_{1}}_{i_{\sigma(1)}} \dots Z^{i_{n}}_{i_{\sigma(n)}} \,,
\end{equation}
the Hermitian conjugate is 
\begin{equation}
\label{eqn:aomb_gio_herm_definition}
(\mathcal{O}_{\sigma}(Z))^{\dagger} = \left(\text{Tr}_{V^{\otimes n}_{N}}(Z^{\otimes n} \cL_{\sigma})\right)^{\dagger} = (Z^{\dagger})^{i_{1}}_{i_{\sigma^{-1}(1)}} \dots (Z^{\dagger})^{i_{n}}_{i_{\sigma^{-1}(n)}} = \mathcal{O}_{\sigma^{-1}}(Z^{\dagger}) \,.
\end{equation}
By taking the expectation value of the product of \eqref{eqn:aomb_gio_definition2} and \eqref{eqn:aomb_gio_herm_definition}, the two-point function/correlator of GIOs with coupling matrix field is
\begin{align}
\langle \mathcal{O}_{\sigma_{1}}(Z) & (\mathcal{O}_{\sigma_{2}}(Z) )^{\dagger} \rangle \nonumber 
\\
&= \frac{1}{\Sigma[0]}  \int [dZ] \mathcal{O}_{\sigma_{1}}(Z)  (\mathcal{O}_{\sigma_{2}}(Z) )^{\dagger}  e^{-\text{Tr}(Z A Z^{\dagger})}
\\
&= \langle \text{Tr}_{V^{\otimes n}_{N}}(Z^{\otimes n} \cL_{\sigma_{1}})  \left(\text{Tr}_{V^{\otimes n}_{N}}(Z^{\otimes n} \cL_{\sigma_{2}}) \right)^{\dagger} \rangle  \label{eqn:bmf_exp_val_gen1} 
\\
&= \sum_{ \substack{ i_{1}, i_{2}, \dots , i_{n} \\ j_{1}, j_{2}, \dots , j_{n}}} \left\langle Z^{i_{1}}_{i_{\sigma_{1}(1)}} \dots Z^{i_{n}}_{i_{\sigma_{1}(n)}} (Z^{\dagger})^{j_{1}}_{j_{\sigma_{2}^{-1}(1)}} \dots (Z^{\dagger})^{j_{n}}_{j_{\sigma_{2}^{-1}(n)}} \right\rangle  \label{eqn:bmf_exp_val_gen2} 
\\
&= \sum_{ \substack{ i_{1}, i_{2}, \dots , i_{n} \\ j_{1}, j_{2}, \dots , j_{n}}} \sum_{\gamma \in S_{n}} \delta^{i_{1}}_{j_{ \sigma_{2}^{-1}\left(\gamma(1)\right)}} \dots   \delta^{i_{n}}_{j_{\sigma_{2}^{-1}\left(\gamma (n)\right)}} B^{j_{\gamma(1)}}_{i_{\sigma_{1}(1)}}\dots B^{j_{\gamma(n)}}_{i_{\sigma_{1}(n)}} \label{eqn:bmf_exp_val_gen3_point_5} 
\\
&= \sum_{ \substack{ i_{1}, i_{2}, \dots , i_{n} \\ j_{1}, j_{2}, \dots , j_{n}}} \sum_{\gamma \in S_{n}} \delta^{i_{1}}_{j_{\gamma \sigma_{2}^{-1}(1)}} \dots   \delta^{i_{n}}_{j_{\gamma \sigma_{2}^{-1}(n)}} B^{j_{\gamma(1)}}_{i_{\sigma_{1}(1)}}\dots B^{j_{\gamma(n)}}_{i_{\sigma_{1}(n)}} \label{eqn:bmf_exp_val_gen3} 
\\  
&= \sum_{ \substack{ i_{1}, i_{2}, \dots , i_{n} \\ j_{1}, j_{2}, \dots , j_{n}}} \sum_{\gamma \in S_{n}} \delta^{i_{\gamma^{-1} \sigma_{1}(1)}}_{j_{\gamma^{-1} \sigma_{1} \gamma \sigma_{2}^{-1}(1)}} \dots  \delta^{i_{\gamma^{-1} \sigma_{1}(n)}}_{j_{\gamma^{-1} \sigma_{1} \gamma \sigma_{2}^{-1}(n)}} B^{j_{1}}_{i_{\gamma^{-1} \sigma_{1}(1)}} \dots B^{j_{n}}_{i_{\gamma^{-1} \sigma_{1}(n)}} \label{eqn:bmf_exp_val_gen4}
\\ 
&= \sum_{j_{1}, j_{2}, \dots , j_{n}} \sum_{\gamma \in S_{n}} B_{j_{\gamma^{-1} \sigma_{1} \gamma \sigma_{2}^{-1}(1)}}^{j_{1}} \cdots B_{j_{\gamma^{-1} \sigma_{1} \gamma \sigma_{2}^{-1}(n)}}^{j_{n}} \label{eqn:bmf_exp_val_gen5}
\\
&= \sum_{\gamma \in S_{n}} \text{Tr}_{V_{N}^{\otimes n}}(B^{\otimes n} \cL_{\gamma^{-1}\sigma_{1}\gamma\sigma_{2}^{-1}}) \label{eqn:bmf_exp_val_gen6}\,,
\end{align}
In \eqref{eqn:bmf_exp_val_gen3_point_5}, Wick's theorem has been used to write the correlator as a sum over permutations $\gamma \in S_{n}$ in conjunction with the basic correlator of \eqref{eqn:bmf_new_correlator_part_func}. The step between equations \eqref{eqn:bmf_exp_val_gen3} and \eqref{eqn:bmf_exp_val_gen4} uses the following 
identity 
\begin{equation}
\label{eqn:bmf_kron_equiv}
\delta^{i_{\tau(1)}}_{j_{\sigma(1)}} \dots \delta^{i_{\tau(n)}}_{j_{\sigma(n)}} = \delta^{i_{1}}_{j_{\tau^{-1}\sigma(1)}} \dots \delta^{i_{n}}_{j_{\tau^{-1}\sigma(n)}}
\end{equation}
where the action of permutations on a number $a$ is $\tau^{-1} \sigma (a) = \sigma(\tau^{-1}(a))$, i.e. the chosen convention is to act with the leftmost permutation first (left-to-right).  We will refer to this identity for products of Kronecker delta functions  as \textit{Kronecker equivariance}. It follows by re-ordering the Kronecker delta's. In equation \eqref{eqn:bmf_exp_val_gen4}, we also used the identity 
\bea 
B^{j_{\gamma(1)}}_{i_{\sigma_{1}(1)}}\dots B^{j_{\gamma(n)}}_{i_{\sigma_{1}(n)}} 
= B^{j_{1}}_{i_{\gamma^{-1} \sigma_{1}(1)}} \dots B^{j_{n}}_{i_{\gamma^{-1} \sigma_{1}(n)}} \, . 
\eea
This follows from a  similar re-ordering of the $B$-matrix elements. The final equation \eqref{eqn:bmf_exp_val_gen6} writes the correlator as a  trace  over the $n$-fold tensor product of $V_N$. 
 
A neat way to understand the structure of the  derivation and to anticipate the outcome is to use the diagrammatic representation of linear operators in tensor spaces \cite{Corley:2002mj,Pasukonis:2013ts}. Figure \ref{fig:two_perm_correlator_bmf_diagram} demonstrates the original two-point function on the left, while the right hand side replaces the matrix variables with permutation operators and the tensor product of coupling matrix $B$. To obtain the result from the diagram, one follows a branch downward and multiplies the $B^{\otimes n}$ and $\mathcal{L}$ operator boxes encountered. The horizontal bars on each branch imply a trace by identifying the bottom and the top of the diagram. The figure therefore expresses the correlator equation 
\begin{align} 
\label{eqn:bmf_diagram_interpretation_step1}
\langle \text{Tr}_{V^{\otimes n}_{N}}(Z^{\otimes n}\mathcal{L}_{\sigma_{1}})  \text{Tr}_{V^{\otimes n}_{N}}((Z^{\dagger})^{\otimes n}\mathcal{L}_{\sigma_{2}^{-1}})  \rangle 
{}&= \sum_{\gamma \in S_{n}} \text{Tr}_{V^{\otimes n}_{N}}(B^{\otimes n} \mathcal{L}_{\gamma^{-1}} \mathcal{L}_{\sigma_{1}} \mathcal{L}_{\gamma} \mathcal{L}_{\sigma_{2}^{-1}})
\,,
\end{align} 
and we infer the equation on the right, starting at $B^{\otimes n}$ in the right side diagram. From definition \eqref{eqn:bmf_lin_op_action}, it is possible to show that multiple permutation operators acting on tensor space follow the rule $\mathcal{L}_{\sigma}\mathcal{L}_{\tau} = \mathcal{L}_{\sigma \tau}$, where again, $\sigma \tau (i) = \tau(\sigma(i))$ is the composition of elements, applying left-to-right convention. This action is represented diagrammatically in Figure \ref{fig:lin_op_combo} and using this fact, the expression of \eqref{eqn:bmf_diagram_interpretation_step1} reduces to that of the derived result in \eqref{eqn:bmf_exp_val_gen5}. For a more detailed discussion of the diagram interpretation and linear operators, see Appendix \ref{app:corr_diagram_interpret}.  
\begin{figure}[htb!]
\begin{center}
\includegraphics[width=15cm, height=7cm]{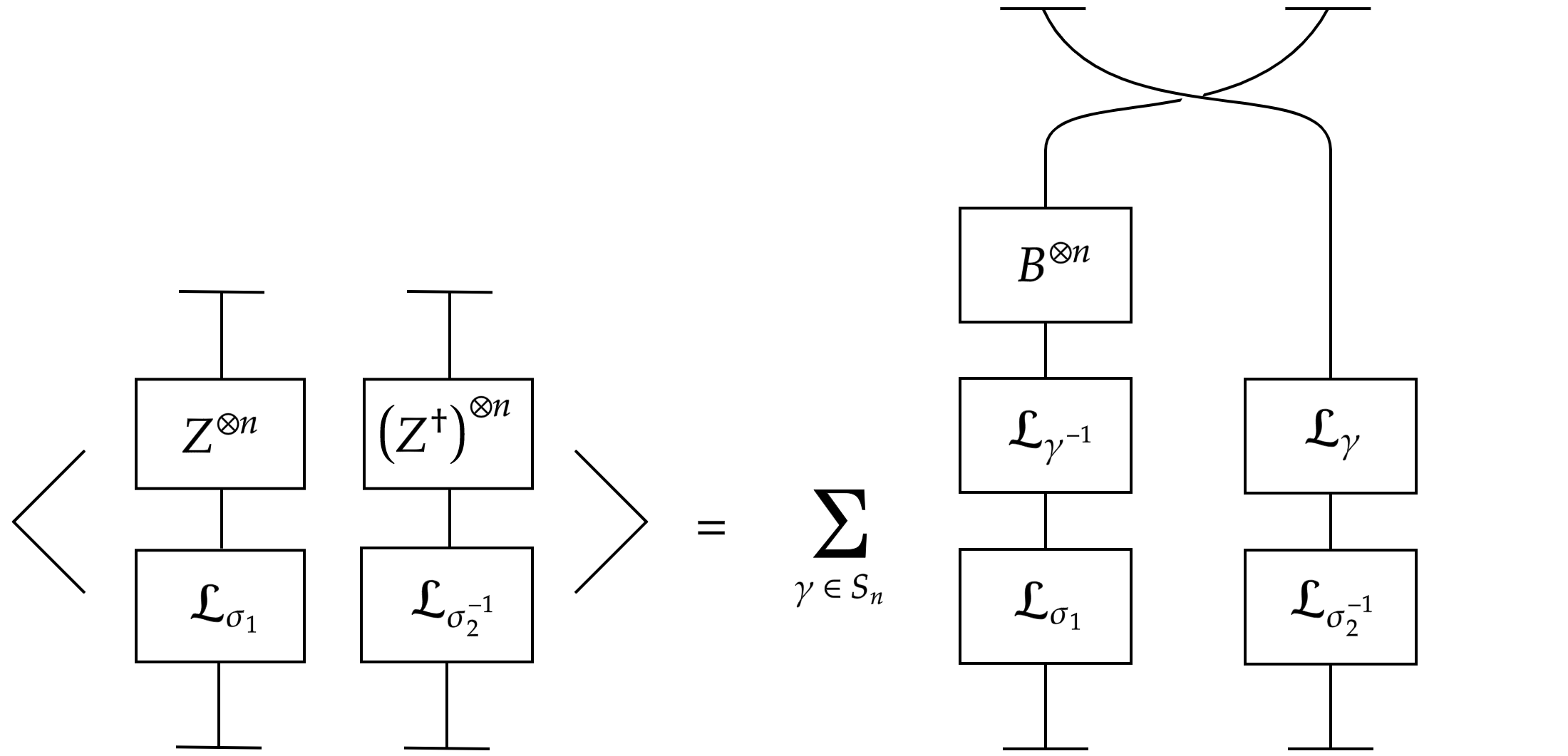}
\end{center}
\caption{The diagrammatic representation of the one-matrix correlator with coupling witness field. The expectation value of these gauge invariant operators can be written as a sum over permutations by exchanging $Z^{\otimes n}$ and $(Z^{\dagger})^{\otimes n}$ for permutation operators $\mathcal{L}_{\gamma^{-1}}$ and  $\mathcal{L}_{\gamma}$, inserting $B^{\otimes n}$ and then swapping the indices. This tensor product of the coupling witness matrix $B$ is incorporated when the $(\gamma^{-1},\gamma)$ permutations are inserted via Wick contraction, using equation \eqref{eqn:bmf_new_correlator_part_func}.}
\label{fig:two_perm_correlator_bmf_diagram}
\end{figure}
\begin{figure}[htb!]
\begin{center}
\includegraphics[width=5.5cm, height=5.5cm]{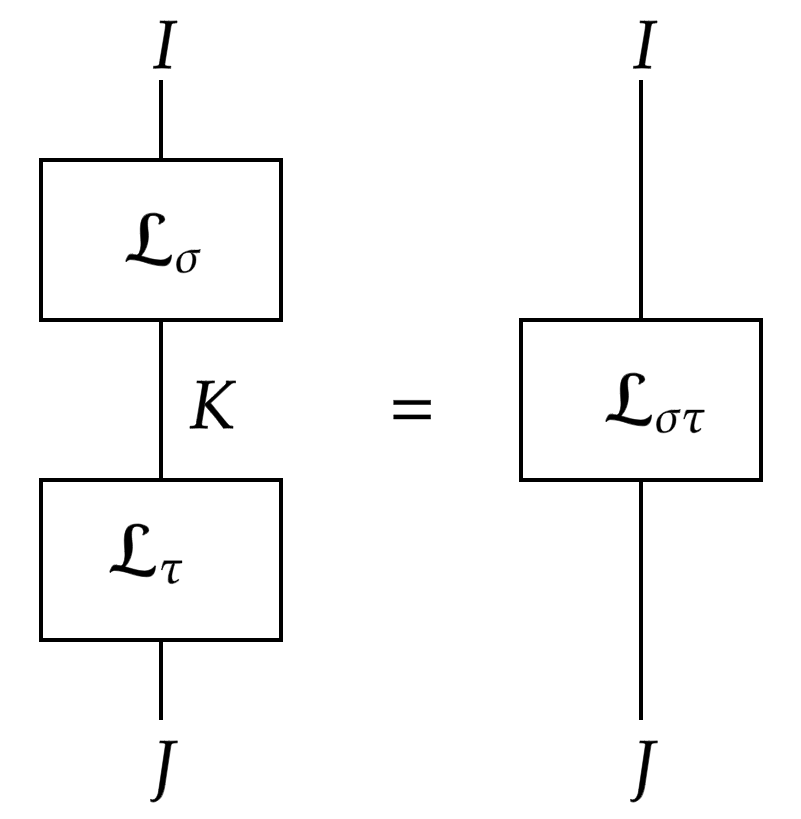}
\end{center}
\caption{The diagrammatic interpretation of the multiplication of linear operators $\mathcal{L}_{\sigma}$ and $\mathcal{L}_{\tau}$. In terms of the tensor components of these operators, $(\mathcal{L}_{\sigma})^{i_{1}, \dots, i_{n}}_{k_{1}, \dots, k_{n}} (\mathcal{L}_{\tau})^{k_{1}, \dots, k_{n}}_{j_{1}, \dots, j_{n}} = (\mathcal{L}_{\sigma \tau})^{i_{1}, \dots, i_{n}}_{j_{1}, \dots, j_{n}}$, where the sum over repeated $k$ indices is implicit.}
\label{fig:lin_op_combo}
\end{figure}

Like the previous work without witness matrix fields, this two-point function is naturally related to the cycle structure of the product of these permutations
\begin{align}
\label{eqn:bmf_cycle_struc_B_mat}
\langle \mathcal{O}_{\sigma_{1}}(Z) (\mathcal{O}_{\sigma_{2}}(Z) )^{\dagger} \rangle 
&= \sum_{\gamma \in S_{n}}  \sum_{j_{1}, j_{2}, \dots , j_{n}} B_{j_{\gamma^{-1} \sigma_{1} \gamma \sigma_{2}^{-1}(1)}}^{j_{1}} \cdots B_{j_{\gamma^{-1} \sigma_{1} \gamma \sigma_{2}^{-1}(n)}}^{j_{n}}
\\
&= \sum_{\gamma \in S_{n}} \prod_{i=1}^{n} \left[ \text{Tr}(B^{i}) \right]^{C_{i}(\gamma^{-1} \sigma_{1} \gamma \sigma_{2}^{-1})} \label{eqn:bmf_cycle_struc_B_mat3}
\\
&=\sum_{\sigma_{3} \in S_{n}} \sum_{\gamma \in S_{n}} \prod_{i=1}^{n} \left[\text{Tr}(B^{i}) \right]^{C_{i}(\sigma_{3})} \delta(\sigma_{3}^{-1} \gamma^{-1} \sigma_{1} \gamma \sigma_{2}^{-1}) \label{eqn:bmf_cycle_struc_B_mat4}\,.
\end{align}
Here ``$\text{Tr}$" is the matrix trace and $C_{i}(\gamma^{-1} \sigma_{1} \gamma \sigma_{2}^{-1})$ is the number of $i$-length cycles in the permutation $\gamma^{-1} \sigma_{1} \gamma \sigma_{2}^{-1}$. In \eqref{eqn:bmf_cycle_struc_B_mat4}, a delta function on permutations is introduced using $\sigma_{3} \in S_{n}$. This is defined by $\delta(\sigma) = 1$ for $\sigma = e$ (the identity element) and $\delta(\sigma) = 0$ for all other elements.  From here, the appearance of the underlying PCA in the correlator can also be made manifest. Since the $\sigma_{3}$ sum runs over the entire group, it can be decomposed in terms of conjugacy classes/partitions
\begin{align}
\sum_{\sigma_{3} \in S_{n}} \sum_{\gamma \in S_{n}} \prod_{i=1}^{n}& \left[\text{Tr}(B^{i}) \right]^{C_{i}(\sigma_{3})}\delta(\sigma_{3}^{-1} \gamma^{-1} \sigma_{1} \gamma \sigma_{2}^{-1}) \nonumber
\\
&= \sum_{\gamma \in S_{n}} \sum_{p_{3}  \vdash n } \left( \sum_{\alpha \in \mathbf{C}_{p_{3}}} \prod_{i=1}^{n} \left[\text{Tr}(B^{i}) \right]^{C_{i}(\alpha)} \delta(\alpha^{-1} \gamma^{-1} \sigma_{1} \gamma \sigma_{2}^{-1})  \right) 
\\
& = \sum_{\gamma \in S_{n}} \sum_{p_{3}  \vdash n}  \prod_{i=1}^{n} \left[\text{Tr}(B^{i}) \right]^{C_{i}\left(\sigma^{\left(p_{3}\right)}\right)} \delta \left(\sum_{\alpha \in \mathbf{C}_{p_{3}}} \alpha^{-1} \gamma^{-1} \sigma_{1} \gamma \sigma_{2}^{-1} \right) 
\\
& = \sum_{p_{3} \vdash n } \mathcal{O}_{\sigma^{\left(p_{3}\right)}}(B) \sum_{\gamma \in S_{n}}\delta \left(T_{p_{3}} \gamma^{-1} \sigma_{1} \gamma \sigma_{2}^{-1} \right)  \label{eqn:bmf_obs_with_p_label_and_delta} \,,
\end{align}
where $p_{3}$ labels the conjugacy classes in the decomposition. Note that $C_{i}(\alpha)$ is replaced with $C_{i}(\sigma^{\left(p_{3} \right)})$ here, as both $\alpha$ and $\sigma^{(p_{3})}$ represent permutations from conjugacy class $\mathbf{C}_{p_{3}}$ and therefore have the same cycle structure. Consequently, the sum over $\alpha$ can move inside the delta function and since $\alpha$ and $\alpha^{-1}$ belong to the same conjugacy class $\mathbf{C}_{p_{3}}$, the PCA element $T_{p_{3}} = \sum_{\alpha \in \mathbf{C}_{p_{3}}} \alpha =\sum_{\alpha \in \mathbf{C}_{p_{3}}} \alpha^{-1}$, is used. Equation \eqref{eqn:bmf_obs_with_p_label_and_delta} then sets 
\begin{equation}
\prod_{i=1}^{n} \left[\text{Tr}(B^{i}) \right]^{C_{i}\left(\sigma^{\left(p_{3}\right)}\right)}= \text{Tr}_{V^{\otimes n}_{N}}\left(B^{\otimes n} \cL_{\sigma^{\left(p_{3} \right)}} \right) = \mathcal{O}_{\sigma^{\left(p_{3} \right)}}(B) \,,
\end{equation}
following the general definition in equation \eqref{eqn:aomb_gio_definition}, and noting again that $\sigma^{\left(p_{3}\right)}$ is any representative permutation from class $\mathbf{C}_{p_{3}}$. At this point, it is useful to introduce the following lemma.
\begin{lemma} 
\label{lem:single-matrix_lemma}
For $\gamma, \sigma_{1},\sigma_{2} \in S_{n}$, $T_{p_{i}} \in \mathcal{Z}[\mathbb{C}[S_{n}]]$ the following equality holds 
\begin{equation}
\sum_{\gamma \in S_{n}} \delta(T_{p_{3}} \gamma^{-1} \sigma_{1} \gamma \sigma_{2}^{-1}) =  \frac{n!|T_{p_{3}} |}{|T_{p_{1}} | |T_{p_{2}}|} C^{p_{3}}_{p_{1} p_{2}} 
\end{equation}
where $C^{p_{3}}_{p_{1} p_{2}}$ is a $\mathcal{Z}[\mathbb{C}[S_{n}]]$ PCA structure constant and $\sigma_{1}, \sigma_{2}^{-1}$ belong to conjugacy classes $\mathbf{C}_{p_{1}}, \mathbf{C}_{p_{2}}$ respectively.
\end{lemma} 
\begin{proof}
\begin{align}
\sum_{\gamma \in S_{n}} \delta(T_{p_{3}} \gamma^{-1} \sigma_{1} \gamma \sigma_{2}^{-1}) 
&= \sum_{\mu_{1} \in S_{n}} \delta\left(T_{p_{3}} (\mu_{1} \mu_{2})^{-1} \sigma_{1} (\mu_{1} \mu_{2}) \sigma_{2}^{-1}\right) \label{eqn:bmf_lem_relabel_gam}
\\
&= \sum_{\mu_{1} \in S_{n}} \delta\left((\mu_{2}^{-1} \mu_{2})T_{p_{3}} \mu_{2}^{-1} \mu_{1}^{-1} \sigma_{1} \mu_{1} \mu_{2} \sigma_{2}^{-1}\right) \label{eqn:bmf_lem_insert_id}
\\
&=  \sum_{\mu_{1} \in S_{n}} \delta \left( \underbrace{\mu_{2}T_{p_{3}} \mu_{2}^{-1}}_{=T_{p_{3}}} \mu_{1}^{-1} \sigma_{1} \mu_{1} \mu_{2} \sigma_{2}^{-1} \mu_{2}^{-1}\right)\label{eqn:bmf_lem_cycles_mu2inv}
\\
&=  \frac{1}{n!} \sum_{\mu_{1},\mu_{2} \in S_{n}} \delta \left( T_{p_{3}} (\mu_{1}^{-1} \sigma_{1} \mu_{1}) (\mu_{2} \sigma_{2}^{-1} \mu_{2}^{-1}) \right)\label{eqn:bmf_lem_sum_over_mu2_added}
\\
&=  \frac{1}{n!}  \delta \left( T_{p_{3}} \left(\sum_{\mu_{1} \in S_{n}}\mu_{1}^{-1} \sigma_{1} \mu_{1}\right) \left(\sum_{\mu_{2} \in S_{n}} \mu_{2} \sigma_{2}^{-1} \mu_{2}^{-1}\right) \right) \label{eqn:bmf_lem_collecting_sums_in_delta}
\\
&= \frac{| \text{Aut}_{\mathcal{Z}}\left(p_{1}\right)|| \text{Aut}_{\mathcal{Z}}\left(p'_{2}\right)|}{n!}  \delta \left( T_{p_{3}} T_{p_{1}} T_{p'_{2}} \right) \label{eqn:bmf_lem_eles_and_autos_included}
\\
&= \frac{n!}{|T_{p_{1}}||T_{p_{2}}|}  \delta \left( T_{p_{3}} T_{p_{1}} T_{p_{2}} \right)
\label{eqn:bmf_lem_all_pca_eles_prime_removed} \,.
\end{align}
In \eqref{eqn:bmf_lem_relabel_gam}, the sum over  $\gamma$ has been replaced by a sum over 
$\mu_1$ with $\gamma \rightarrow \mu_1 \mu_2 $. Equation \eqref{eqn:bmf_lem_insert_id} inserts the identity element $e = \mu_{2}^{-1} \mu_{2}$, while \eqref{eqn:bmf_lem_cycles_mu2inv} cycles the $\mu_{2}^{-1}$ permutation to the right hand side of the delta and utilises $T_{p_{3}} = \mu_{2} T_{p_{3}} \mu_{2}^{-1}$. An additional sum  over $\mu_{2}$ is introduced in \eqref{eqn:bmf_lem_sum_over_mu2_added}, and the appropriate normalisation factor $1/n!$ introduced (see Appendix \ref{app:delta_function_sums} for similar arguments). Equation \eqref{eqn:bmf_lem_collecting_sums_in_delta} collects the sums inside the delta functions, leading to the introduction of the PCA elements in \eqref{eqn:bmf_lem_eles_and_autos_included}, along with the automorphism group sizes as per the PCA definition in \eqref{eqn:bmf_zcs_algebra_ele}. We take $\sigma_{1}$ and $\sigma_{2}^{-1}$ to belong to conjugacy classes $\mathbf{C}_{p_{1}}$ and $\mathbf{C}_{p_{2}'}$ respectively: the primes on the partition/class labels indicating that they are related to an inverse permutation. The final step makes use of the orbit-stabiliser theorem, whereby the size of the automorphism group is related to the conjugacy/equivalence class size and gauge permutation group order by $|\text{Aut}_{\mathcal{Z}}(p)| = |S_{n}|/|T_{p}|$. As was the case with $T_{p_{3}}$, since a permutation and its inverse share the same conjugacy class, the prime labels have subsequently been removed in \eqref{eqn:bmf_lem_all_pca_eles_prime_removed}.

The next step makes use of PCA multiplication and identities of the delta function 
\begin{equation}
\delta(T_{p_{3}} T_{p_{1}} T_{p_{2}}) = \sum_{p_{k}} C^{p_{k}}_{p_{1} p_{2}} \delta(T_{p_{3}}T_{p_{k}}) = \sum_{p_{k}} \delta_{p_{3} p_{k}} C^{p_{k}}_{p_{1} p_{2}} |T_{p_{k}}| = |T_{p_{3}}|C^{p_{3}}_{p_{1} p_{2}}\,.
\end{equation}
Therefore
\begin{align}
\sum_{\gamma \in S_{n}} \delta(T_{p_{3}} \gamma^{-1} \sigma_{1} \gamma \sigma_{2}^{-1})  
&=  \frac{n!|T_{p_{3}} |}{|T_{p_{1}} | |T_{p_{2}}|} C^{p_{3}}_{p_{1} p_{2}}  \label{eqn:bmf_exp_val_with_bmf2}\,.
\end{align}
\end{proof}
The result of this lemma combined with equation \eqref{eqn:bmf_obs_with_p_label_and_delta}, directly shows that upon inclusion of a coupling matrix field, the correlator may be written as 
\begin{equation}
\label{eqn:aomb_correlator_result}
\langle \mathcal{O}_{\sigma_{1}}(Z)  (\mathcal{O}_{\sigma_{2}}(Z) )^{\dagger} \rangle = \sum_{p_{3}  \vdash n} \mathcal{O}_{\sigma^{\left(p_{3}\right)}}(B) \frac{n!|T_{p_{3}} |}{|T_{p_{1}} | |T_{p_{2}}|} C^{p_{3}}_{p_{1} p_{2}} \,.
\end{equation}
As the GIOs in the correlator are functions of conjugacy class, we may write $\mathcal{O}_{\sigma_{1}} \equiv \mathcal{O}_{\sigma^{\left(p_{1}\right)}}$ and $\mathcal{O}_{\sigma_{2}} \equiv \mathcal{O}_{\sigma^{\left(p_{2}\right)}}$.  Additionally, since a PCA element is the sum of permutations in a given equivalence class, the GIO may also be written in combinatorial basis form, using PCA elements divided by the size of the class 
\begin{equation}
\label{eqn:aomb_op_oerm_to_pca_relation}
\mathcal{O}_{\sigma^{\left(p_{i}\right)}} = \frac{1}{|T_{p_{i}}|}\mathcal{O}_{T_{p_{i}}} \equiv \frac{1}{|T_{p_{i}}|}\mathcal{O}_{p_{i}} \,.
\end{equation}
Therefore,  by rearranging the class size factors, the final result can be written as
\begin{equation}
\label{eqn:aomb_correlator_result_boxed}
\boxed{\langle \mathcal{O}_{p_{1}}(Z)  (\mathcal{O}_{p_{2}}(Z))^{\dagger} \rangle = n!\sum_{p_{3} \vdash n} C^{p_{3}}_{p_{1} p_{2}}  \mathcal{O}_{p_{3}}(B)}
\end{equation}
showing that the insertion of gauge invariant functions in  combinatorial basis,  $\cO_{ p_1} ( Z) , (\mathcal{O}_{p_{2}}(Z))^{\dagger}$ of the fluctuating/quantum field $Z$,  is equal to a linear combination of gauge invariant functions of the $B$ witness fields. The structure constants $C_{ p_1 p_2 }^{p_3} $  can be exactly reconstructed 
by choosing the appropriate basis labels $( p_1, p_2 , p_3 ) $ for the gauge invariant functions of the quantum and witness fields. Contact with the previous result of \cite{Corley:2001zk} can be made by setting the coupling matrix in equation \eqref{eqn:bmf_cycle_struc_B_mat3} equal to the $N \times N$ identity matrix, i.e. $B = \mathbb{I}$,
\begin{align}
\begin{split}
\label{eqn:aomb_recover_prev_res}
\langle \mathcal{O}_{\sigma_{1}}(Z) (\mathcal{O}_{\sigma_{2}}(Z))^{\dagger} \rangle 
&= 
\sum_{\gamma \in S_{n}} \prod_{i=1}^{n} \left[ \text{Tr}(\mathbb{I}^{i}) \right]^{C_{i}(\gamma^{-1} \sigma_{1} \gamma \sigma_{2}^{-1})} 
\\
&=  \sum_{\gamma \in S_{n}} \prod_{i=1}^{n}  N^{C_{i}(\gamma^{-1} \sigma_{1} \gamma \sigma_{2}^{-1})} 
\\
&=  \sum_{\gamma \in S_{n}} N^{C_{1}(\gamma^{-1} \sigma_{1} \gamma \sigma_{2}^{-1}) + \dots + C_{n}(\gamma^{-1} \sigma_{1} \gamma \sigma_{2}^{-1})} 
\\
&= \sum_{\gamma \in S_{n}} N^{C_{\gamma^{-1} \sigma_{1} \gamma \sigma_{2}^{-1}}}\,.
\end{split}
\end{align}
where $C_{i}(\sigma)$ is the number of $i$-length cycles in $\sigma$ and $C_{\sigma}$ is the total number of cycles in $\sigma$.

\subsection{Fourier basis for two-point function with matrix coupling}
\label{ss:single_matrix_schur_basis}

We have established above a direct connection between the structure constants of the algebra 
$\mathcal{Z}[\mathbb{C}[S_{n}]]$ and the two-point correlator of matrix observables in the presence of a matrix coupling. In this section we will show that a Fourier transform on $\mathcal{Z}[\mathbb{C}[S_{n}]]$ which  maps the combinatorial basis to a Young-diagram basis results in on orthogonal 2-point correlator in the presence of the matrix coupling. This generalizes the orthogonality result for Young-diagram basis operators, also called Schur polynomial operators, which captures finite $N$ effects in terms of a simple cut-off on the Young diagram and informs  the map between half-BPS operators 
in $ \cN =4$ SYM and giant gravitons   \cite{Corley:2001zk,Ramgoolam:2016ciq}. The role of the Fourier transform on algebras in the construction of orthogonal bases was emphasized in \cite{Brown:2008ij,Brown:2007xh,Pasukonis:2013ts,deMelloKoch:2012ck,Mattioli:2016eyp}. 


Applying the Fourier transform to the single matrix operator, as seen in \S \ref{ss:background_matrix_field}, we define a general gauge invariant operator as
\begin{equation}
\label{eqn:spb_complex_character_schur}
\mathcal{O}_{R}(Z) = \frac{1}{n!}\sum_{\sigma \in S_{n}} \chi_{R}(\sigma)\text{Tr}_{V^{\otimes n}}( Z^{\otimes n} \cL_{\sigma}) = \frac{1}{n!}\sum_{\sigma \in S_{n}} \chi_{R}(\sigma)\mathcal{O}_{\sigma}(Z)\,.
\end{equation}
along with conjugate operator
\begin{equation}
\left(\mathcal{O}_{R}(Z)\right)^{\dagger} = \frac{1}{n!}\sum_{\sigma \in S_{n}} \left(\chi_{R}(\sigma)\right)^{*} \left(\text{Tr}_{V^{\otimes n}}( Z^{\otimes n} \cL_{\sigma})\right)^{\dagger} = \frac{1}{n!}\sum_{\sigma \in S_{n}} \chi_{R}(\sigma) \left(\mathcal{O}_{\sigma}(Z)\right)^{\dagger} \,.
\end{equation}
Here $R$ denotes the irreducible representation (irrep) of $S_{n}$, $Z$ is a complex matrix and $\chi_{R}(\sigma)$ is the character of the representation $R$ of element $\sigma \in S_{n}$ \cite{Corley:2001zk}. The fact $(\chi_{R}(\sigma))^{*} = \chi_{R}(\sigma^{-1}) = \chi_{R}(\sigma)$ has been used since characters of $S_{n}$ representations can be chosen to be real and the operator $\mathcal{O}_{\sigma}(Z)$ is as previously defined in \eqref{eqn:aomb_gio_definition}. These $\mathcal{O}_{R}(Z)$, with $R$ having $n$ boxes,  form a basis in the space of $U(N)$ invariant polynomial functions of the matrix $Z$ of degree $n$, where $U(N)$ acts on $Z$ by conjugation. The calculation of the general GIO two-point functions starts with
\begin{align}
\label{eqn:spb_schur_gen_gauge_inv}
\langle \mathcal{O}_{R}(Z) \left(\mathcal{O}_{S}(Z)\right)^{\dagger} \rangle 
&= \frac{1}{(n!)^2} \sum_{\sigma \in S_{n}} \sum_{\tau \in S_{n}} \chi_{R}(\sigma) \chi_{S}(\tau) \langle \mathcal{O}_{\sigma}(Z) \left(\mathcal{O}_{\tau}(Z)\right)^{\dagger} \rangle \,.
\end{align}
The expectation value appearing in the right hand side is just the  correlator found in the previous section using the permutation parameterisation of the observables. Using its form given in equation \eqref{eqn:bmf_cycle_struc_B_mat4}, stated again here for convenience 
\begin{equation}
\langle \mathcal{O}_{\sigma}(Z) \left(\mathcal{O}_{\tau}(Z)\right)^{\dagger} \rangle = \sum_{\rho \in S_{n}} \sum_{\gamma \in S_{n}} \prod_{i=1}^{n} \delta(\rho^{-1} \gamma^{-1} \sigma \gamma \tau^{-1}) \left[\text{Tr}(B^{i}) \right]^{C_{i}(\rho)}\,,
\end{equation}
this can be substituted in to \eqref{eqn:spb_schur_gen_gauge_inv} to achieve the following
\begin{equation}
\label{eqn:spb_corr_as_delta_func}
\langle \mathcal{O}_{R}(Z) \left(\mathcal{O}_{S}(Z) \right)^{\dagger} \rangle =  
\frac{1}{(n!)^2} \sum_{\substack{\sigma, \tau, \rho, \gamma \\  \in S_{n}}}\prod_{i=1}^{n} \chi_{R}(\sigma) \chi_{S}(\tau)  \delta(\rho^{-1} \gamma^{-1} \sigma \gamma \tau^{-1}) \left[\text{Tr}(B^{i}) \right]^{C_{i}(\rho)} \,.
\end{equation}
To simplify the correlator further, the following lemma is introduced.
\begin{lemma}
\label{lem:single-matrix_schur_basis_lemma}
For $\sigma, \tau, \rho, \gamma \in S_{n}$ and where $\chi_{R}(\sigma)$ are characters of $\sigma$ in representation $R$, the following equality holds
\begin{equation}
\frac{1}{(n!)^2} \sum_{\substack{\sigma, \tau, \rho, \gamma \\  \in S_{n}} }\prod_{i=1}^{n} \chi_{R}(\sigma) \chi_{S}(\tau)  \delta(\rho^{-1} \gamma^{-1} \sigma \gamma \tau^{-1}) \left[\text{Tr}(B^{i}) \right]^{C_{i}(\rho)} = \frac{n!\delta_{RS} }{d_{R}} \mathcal{O}_{S}(B) \,.
\end{equation}
where $d_{R}$ is the dimension of the symmetry group's $R$ representation, $\text{Dim}(V_{R}^{S_{n}})$.
\end{lemma}
\begin{proof}
\begin{align}
\frac{1}{(n!)^2} \sum_{\substack{\sigma, \tau, \rho, \gamma \\  \in S_{n}}}\prod_{i=1}^{n}\chi_{R}(\sigma) \chi_{S}(\tau)&  \delta(\rho^{-1} \gamma^{-1} \sigma \gamma \tau^{-1}) \left[\text{Tr}(B^{i}) \right]^{C_{i}(\rho)} \nonumber\\
&=  \frac{1}{(n!)^2} \sum_{\tau, \rho, \gamma  \in S_{n}}\prod_{i=1}^{n} \chi_{R}(\gamma \rho \tau \gamma^{-1}) \chi_{S}(\tau)  \left[\text{Tr}(B^{i}) \right]^{C_{i}(\rho)} \label{eqn:spb_correlator_sigma_summed}
\\
&=  \sum_{\rho \in S_{n} } \prod_{i=1}^{n} \left(  \frac{1}{n!} \sum_{\tau\in S_{n}}\chi_{S}(\tau) \chi_{R}(\rho\tau) \right)  \left[\text{Tr}(B^{i}) \right]^{C_{i}(\rho)} \label{eqn:spb_correlator_char_invariance}
\\
& = \frac{\delta_{RS}}{d_{S}}  \sum_{\rho \in S_{n}}  \chi_{S}(\rho) \prod_{i=1}^{n} \left[ \text{Tr}(B^{i}) \right]^{\text{Cyc}_{i}(\rho)} \label{eqn:spb_correlator_with_background_char}
\\
& = \frac{\delta_{RS}}{d_{S}} \sum_{\rho \in S_{n}} \chi_{S}(\rho) \mathcal{O}_{\rho}(B) \label{eqn:spb_correlator_with_background_char2}
\\
& = \frac{ n!\delta_{RS}}{d_{R}} \mathcal{O}_{S}(B) \,.
\label{eqn:spb_correlator_with_background_1}
\end{align}
In equation \eqref{eqn:spb_correlator_sigma_summed}, the sum over $\sigma$ was computed to remove the delta function and the cyclic invariance of the character under conjugation of elements was used to remove the $\gamma$ dependence in \eqref{eqn:spb_correlator_char_invariance}, such that the $\gamma$ sum produces a factor of $n!$. Identity  
\begin{equation}
\label{eqn:spb_char_identity}
\frac{1}{n!}\sum_{\tau \in S_{n}} \chi_{S}(\tau) \chi_{R}(\rho \tau) = \frac{\delta_{RS}}{d_{S}} \chi_{S}(\rho) 
\end{equation}
was used to obtain \eqref{eqn:spb_correlator_with_background_char}  and the definition of the permutation-parameterised GIO  was applied in \eqref{eqn:spb_correlator_with_background_char2}. Finally, definition \eqref{eqn:spb_complex_character_schur} was used to write the final expression in terms of an operator in the representation basis of the witness field $B$. 
\end{proof}
Combining equation \eqref{eqn:spb_corr_as_delta_func} with lemma \ref{lem:single-matrix_schur_basis_lemma} above, the correlator in Schur basis simply reduces to
\begin{equation}
\label{eqn:spb_corr_result_in_schur_basis}
\boxed{ 
\langle \mathcal{O}_{R}(Z) \left(\mathcal{O}_{S}(Z) \right)^{\dagger} \rangle = \frac{n!\delta_{RS} }{d_{R}} \mathcal{O}_{S}(B) \,.
} 
\end{equation}
This result states that the correlator of two Schur polynomials/operators is orthogonal and proportional to a Schur polynomial/operator constructed purely from the witness fields, in agreement with recent ideas on the super-integrability of matrix models \cite{Mironov:2017och,Mironov:2019pij,Mironov:2022fsr}.
Additionally, by taking $B = \mathbb{I}$ (the $N \times N$ identity matrix) in equation \eqref{eqn:spb_corr_result_in_schur_basis}, the original result of \cite{Corley:2001zk} is obtained
\begin{equation}
\label{eqn:spb_corr_result_witness_set_to_identity}
\langle \mathcal{O}_{R}(Z) \left(\mathcal{O}_{S}(Z) \right)^{\dagger} \rangle = \frac{ n!\delta_{RS}}{d_{R}} \mathcal{O}_{S}(\mathbb{I}) = \frac{n!\delta_{RS} }{d_{R}}\frac{1}{n!}\sum_{\sigma \in S_{n}}\chi_{S}(\sigma)N^{C_{\sigma}} = \text{Dim}_{N}(S)\frac{n!\delta_{RS} }{d_{R}}\,,
\end{equation}
using the identity $\text{Dim}_{N}(S) = \frac{1}{n!}\sum_{\sigma \in S_{n}} \chi_{S}(\sigma)N^{C_{\sigma}}$ where $\text{Dim}_{N}(S)$ is the dimension of representation $S$ of the unitary group and $C_{\sigma}$ is the total number of cycles in $\sigma$.


\subsection{Gauge invariant functions (observables)  of quantum and classical fields }
\label{ss:back_fields_in_observables}

The previous sections derived the single matrix correlator result by establishing the witness matrix as a coupling in the action. An alternative approach is to define the operators themselves as containing some classical field. Start with a matrix field defined as 
\begin{equation}
\label{eqn:bmf_in_ob_mat_definition}
(ZA)^{i}_{j} = \sum_{k} Z^{i}_{k} A^{k}_{j} \,.
\end{equation}
where $A$ is the classical witness matrix, and $Z$ is the (quantum) matrix integration variable. Then define the gauge invariant operator 
\begin{equation}
\label{eqn:bmf_in_ob_obs_definition}
\mathcal{O}_{\sigma}(ZA) = \text{Tr}_{V^{\otimes n}_{N}} \left( (ZA)^{\otimes n}\cL_{\sigma} \right) = (ZA)^{i_{1}}_{i_{\sigma(1)}} \dots(ZA)^{i_{n}}_{i_{\sigma(n)}}  
\end{equation}
and its Hermitian conjugate
\begin{equation}
\label{eqn:bmf_in_ob_obs_definition_herm}
\left( \mathcal{O}_{\sigma}(ZA) \right)^{\dagger} = \left( \text{Tr}_{V^{\otimes n}_{N}} \left( (ZA)^{\otimes n}\cL_{\sigma}\right)  \right)^{\dagger} = (A^{\dagger} Z^{\dagger})^{j_{1}}_{j_{\sigma^{-1}(1)}} \dots(A^{\dagger} Z^{\dagger})^{j_{n}}_{j_{\sigma^{-1}(n)}} = \mathcal{O}_{\sigma^{-1}}(A^{\dagger} Z^{\dagger}) \,.
\end{equation}
The path integral, now without a coupling matrix field in the action, is
\begin{equation}
\label{eqn:amob_no_bmf_partition_func}
\Sigma[0] = \int  [dZ] e^{ - \text{Tr}\left(ZZ^{\dagger} \right)}
\end{equation}
which produces basic correlator
\begin{equation}
 \langle Z^{i}_{j} (Z^{\dagger})^{k}_{l} \rangle = \frac{1}{\Sigma[0]}  \int [dZ] Z^{i}_{j} (Z^{\dagger})^{k}_{l} e^{-\text{Tr}(Z Z^{\dagger})}  =\delta^{i}_{l} \delta^{k}_{j} \,.
\end{equation}
Therefore the correlator of these GIOs is
\begin{multline}
\langle  \mathcal{O}_{\sigma_{1}}(ZA)
(\mathcal{O}_{\sigma_{2}} (ZA))^{\dagger} \rangle 
= \sum_{I,J,K,L}  A^{k_{1}}_{i_{\sigma_{1}(1)}} \dots A^{k_{n}}_{i_{\sigma_{1}(n)}}(A^{\dagger})^{j_{1}}_{l_{1}} \dots (A^{\dagger})^{j_{n}}_{l_{n}} 
\\
\times \frac{1}{\Sigma[0]} \int [dZ] Z^{i_{1}}_{k_{1}} \dots Z^{i_{n}}_{k_{n}} (Z^{\dagger})^{l_{1}}_{j_{\sigma^{-1}_{2}(1)}} \dots (Z^{\dagger})^{l_{n}}_{j_{\sigma^{-1}_{2}(n)}}e^{-\text{Tr}(Z Z^\dagger)} \,,
\end{multline}
where the sum over $I,J,K,L$ represents a sum over all matrix indices. For notational convenience we define a function in the witness fields
\begin{equation}
\label{eqn:bmf_f_func_in_back_mat}
f(A, A^{\dagger}; \vec{i}, \vec{j}, \vec{k},\vec{l}; \sigma_{1} ) = A^{k_{1}}_{i_{\sigma_{1}(1)}} \dots A^{k_{n}}_{i_{\sigma_{1}(n)}}(A^{\dagger})^{j_{1}}_{l_{1}} \dots (A^{\dagger})^{j_{n}}_{l_{n}}\,,
\end{equation}
where $\vec{i}$, for example, is a vector denoting all possible $i$ indices for a witness matrix. Using this definition, the correlator is written as
\begin{align}
\langle  \mathcal{O}_{\sigma_{1}}(ZA){}&
(\mathcal{O}_{\sigma_{2}} (ZA))^{\dagger} \rangle \nonumber
\\
\begin{split}
{}&= \sum_{I,J,K,L} f(A, A^{\dagger}; \vec{i}, \vec{j}, \vec{k},\vec{l}; \sigma_{1} )  
\\
&\qquad \qquad \times \frac{1}{\Sigma[0]} \int [dZ] Z^{i_{1}}_{k_{1}} \dots Z^{i_{n}}_{k_{n}} (Z^{\dagger})^{l_{1}}_{j_{\sigma^{-1}_{2}(1)}} \dots (Z^{\dagger})^{l_{n}}_{j_{\sigma^{-1}_{2}(n)}}e^{-\text{Tr}(Z Z^\dagger)} \label{eqn:bmf_in_ob_corr}
\end{split}
\\
&=  \sum_{I,J,K,L} f(A, A^{\dagger}; \vec{i}, \vec{j}, \vec{k},\vec{l}; \sigma_{1} ) \langle Z^{i_{1}}_{k_{1}} \dots Z^{i_{n}}_{k_{n}} (Z^{\dagger})^{l_{1}}_{j_{\sigma_{2}^{-1}(1)}} \dots (Z^{\dagger})^{l_{n}}_{j_{\sigma_{2}^{-1}(n)}} \rangle \label{eqn:bmf_isolated_z_var_op}
\\
&=  \sum_{\gamma \in S_{n}} \sum_{I,J,K,L} f(A, A^{\dagger}; \vec{i}, \vec{j}, \vec{k},\vec{l}; \sigma_{1} ) \delta^{i_{1}}_{j_{\gamma \sigma_{2}^{-1}(1)}} \dots \delta^{i_{n}}_{j_{\gamma \sigma_{2}^{-1}(n)}}  \delta^{l_{\gamma(1)}}_{k_{1}} \dots \delta^{l_{\gamma(n)}}_{k_{n}}
\\
&=  \sum_{\gamma \in S_{n}} \sum_{J,K}   (A^{\dagger})^{j_{\gamma(1)}}_{k_{1}} \dots (A^{\dagger})^{j_{\gamma(n)}}_{k_{n}}   A^{k_{1}}_{j_{\sigma_{1}\gamma \sigma_{2}^{-1}(1)}} \dots A^{k_{n}}_{i_{\sigma_{1}\gamma \sigma_{2}^{-1}(n)}}
\\
&= \sum_{\gamma \in S_{n}} \sum_{j_{1}, \dots, j_{n}} \left(A^{\dagger} A\right)_{j_{\gamma^{-1} \sigma_{1} \gamma \sigma_{2}^{-1}(1)}}^{j_{1}} \dots \left(A^{\dagger} A\right)_{j_{\gamma^{-1} \sigma_{1} \gamma \sigma_{2}^{-1}(n)}}^{j_{n}} \label{eqn:bmf_in_ob_corr_explicit}
\\
&= \sum_{\gamma \in S_{n}} \sum_{j_{1}, \dots , j_{n}} B_{j_{\gamma^{-1} \sigma_{1} \gamma \sigma_{2}^{-1}(1)}}^{j_{1}} \cdots B_{j_{\gamma^{-1} \sigma_{1} \gamma \sigma_{2}^{-1}(n)}}^{j_{n}} \label{eqn:bmf_in_ob_corr_b_matrices} 
\\
&= \sum_{\gamma \in S_{n}} \sum_{\sigma_{3} \in S_{n}} \prod_{i=1}^{n} \left[\text{Tr}(B^{i}) \right]^{C_{i}(\sigma_{3})} \delta(\sigma_{3}^{-1} \gamma^{-1} \sigma_{1} \gamma \sigma_{2}^{-1}) \label{eqn:bmf_in_ob_corr_trace_res} 
\\
&= \sum_{\gamma \in S_{n}}\sum_{p_{3} \vdash n} \left[\text{Tr}(B^{i}) \right]^{C_{i}\left(\sigma^{(p_{3})}\right)} \delta \left(\sum_{\alpha \in \mathbf{C}_{p_{3}}} \alpha^{-1} \gamma^{-1} \sigma_{1} \gamma \sigma_{2}^{-1} \right) \label{eqn:bmf_in_ob_corr_sums} 
\\
&=\sum_{p_{3} \vdash n} \mathcal{O}_{\sigma^{(p_{3})}}(B) \sum_{\gamma \in S_{n}}\delta \left(T_{p_{3}} \gamma^{-1} \sigma_{1} \gamma \sigma_{2}^{-1} \right) 
\label{eqn:bmf_in_ob_corr_conj_class_obs_label_p} 
\\
&= \sum_{p_{3} \vdash n} \frac{n!|T_{p_{3}} |}{|T_{p_{1}} | |T_{p_{2}}|} C^{p_{3}}_{p_{1} p_{2}} \mathcal{O}_{\sigma^{(p_{3})}}(B) \,, \label{eqn:bmf_in_ob_corr_struct_const} 
\end{align}
Equation \eqref{eqn:bmf_in_ob_corr} introduces the $f$ function of \eqref{eqn:bmf_f_func_in_back_mat}, then Wick contraction, followed by the redistribution of permutations using Kronecker equivariance, occurs between \eqref{eqn:bmf_isolated_z_var_op}-\eqref{eqn:bmf_in_ob_corr_explicit}. Equation \eqref{eqn:bmf_in_ob_corr_b_matrices} sets $B = A^{\dagger}A$ to match the notation of equation \eqref{eqn:bmf_exp_val_gen5}. Equation \eqref{eqn:bmf_in_ob_corr_sums} splits the $\sigma_{3}$ sum into a sum over conjugacy classes/partitions (labelled $p_{3}$) and a sum over the elements in each class (labelled $\alpha$), while \eqref{eqn:bmf_in_ob_corr_conj_class_obs_label_p}  defines the GIO in witness field $B$ and inserts PCA element $T_{p_{3}} = \sum_{\alpha \in \mathbf{C}_{p_{3}}} \alpha^{-1} \equiv \sum_{\alpha \in \mathbf{C}_{p_{3}}} \alpha$. In \eqref{eqn:bmf_in_ob_corr_struct_const}, lemma \ref{lem:single-matrix_lemma} has been used to write the correlator in terms of structure constants.
In summary, the correlator is 
\begin{equation}
\langle  \mathcal{O}_{\sigma_{1}}(ZA)
(\mathcal{O}_{\sigma_{2}} (ZA))^{\dagger} \rangle   = \sum_{p_{3} \vdash n} \frac{n!|T_{p_{3}} |}{|T_{p_{1}} | |T_{p_{2}}|} C^{p_{3}}_{p_{1} p_{2}} \mathcal{O}_{\sigma^{(p_{3})}}(B) \,.
\end{equation} 
Using $\mathcal{O}_{\sigma_{1}} \equiv \mathcal{O}_{\sigma^{\left(p_{1}\right)}}$, $\mathcal{O}_{\sigma_{2}} \equiv \mathcal{O}_{\sigma^{\left(p_{2}\right)}}$ and equation \eqref{eqn:aomb_op_oerm_to_pca_relation}, we can rewrite this result in combinatorial basis, i.e. in terms of partitions/classes on which the operators depend
\begin{equation}\label{BoxedEquation3} 
\boxed{
\langle  \mathcal{O}_{p_{1}}(ZA)
(\mathcal{O}_{p_{2}} (ZA))^{\dagger} \rangle   = n! \sum_{p_{3} \vdash n} C^{p_{3}}_{p_{1} p_{2}} \mathcal{O}_{p_{3}}(B)
}
\end{equation} 
This shows that using either classical matrices or coupling matrices as the witness fields, the same correlator result is obtained. A box operator diagram is given in Figure \ref{fig:two_perm_correlator_bmf_within_obs_diagram} that portrays this classical matrix construction of the two-point function.

\begin{figure}[htb!]
\begin{center}
\includegraphics[width=14.5cm, height=7cm]{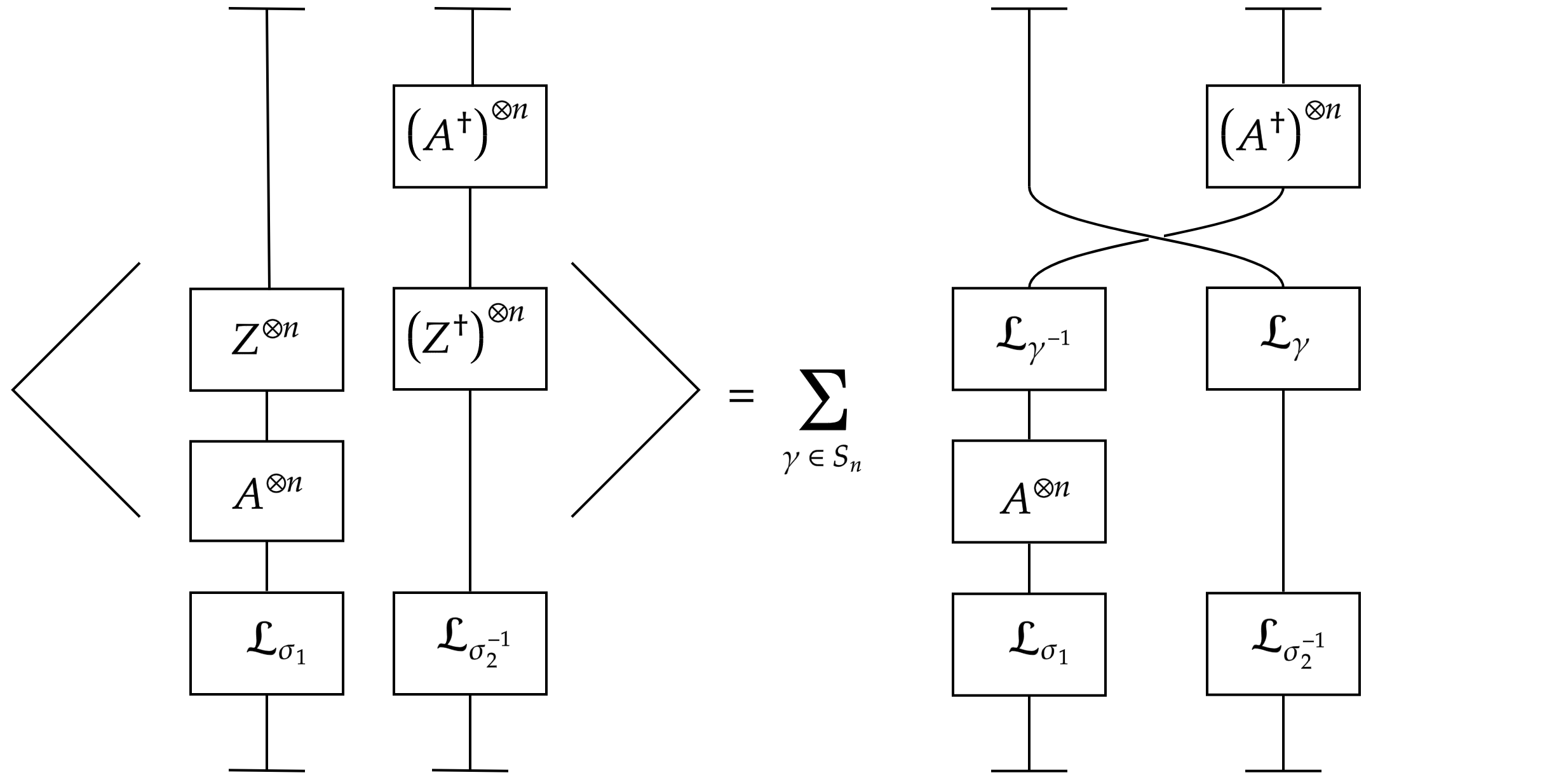}
\end{center}
\caption{The diagrammatic representation of the one-matrix, two-GIO correlator where the witness field is a classical matrix. This matrix is introduced as $A$ (along with its conjugate $A^{\dagger}$) and subsequently used to define $B = A^{\dagger}A$. The right hand diagram shows how upon Wick contracting, the $\mathcal{L}_{\gamma}, \mathcal{L}_{\gamma^{-1}}$, operators are incorporated and the indices are swapped.}
\label{fig:two_perm_correlator_bmf_within_obs_diagram}
\end{figure}

\newpage

\subsection{Fourier basis for two-point function with classical fields }
\label{ss:schur_basis_back_fields_in_observables}

Just as the Schur basis correlator was calculated with the coupling witness field in \S \ref{ss:single_matrix_schur_basis}, here the analog is derived for the classical witness field. Given the description of Schur operators in terms of permutation parameterised operators, as well as the equivalence of the combinatorial basis correlator using either coupling or classical witness fields (as shown by the results of \S\ref{ss:background_matrix_field} and \S\ref{ss:back_fields_in_observables}), this extension is straightforward. Define the classical field GIO in Schur basis as 
\begin{equation}
\label{eqn:spbback_complex_character_schur}
\mathcal{O}_{R}(ZA) = \frac{1}{n!}\sum_{\sigma \in S_{n}} \chi_{R}(\sigma)\text{Tr}_{V^{\otimes n}_{N}}\left( (ZA)^{\otimes n} \cL_{\sigma} \right) = \frac{1}{n!}\sum_{\sigma \in S_{n}} \chi_{R}(\sigma)\mathcal{O}_{\sigma}(ZA)\,,
\end{equation}
and the conjugate operator
\begin{equation}
\label{eqn:spbback_complex_character_schur_conj}
\left(\mathcal{O}_{R}(ZA)\right)^{\dagger} = \frac{1}{n!}\sum_{\sigma \in S_{n}} \left(\chi_{R}(\sigma)\right)^{*}\left(\text{Tr}_{V^{\otimes n}_{N}}\left( (ZA)^{\otimes n} \cL_{\sigma} \right) \right)^{\dagger} = \frac{1}{n!} \sum_{\sigma \in S_{n}} \chi_{R}(\sigma)\left(\mathcal{O}_{\sigma}(ZA)\right)^{\dagger}\,,
\end{equation}
where again, $(\chi_{R}(\sigma))^{*} = \chi_{R}(\sigma^{-1}) = \chi_{R}(\sigma)$ has been used and $R$ corresponds to irrep $R$. The correlator is therefore
\begin{align}
\label{eqn:spbback_schur_corr}
\langle \mathcal{O}_{R}(ZA) \left(\mathcal{O}_{S}(ZA)\right)^{\dagger} \rangle 
&= \frac{1}{(n!)^2} \sum_{\sigma \in S_{n}} \sum_{\tau \in S_{n}} \chi_{R}(\sigma) (\chi_{S}(\tau)) \langle \mathcal{O}_{\sigma}(ZA) \left(\mathcal{O}_{\tau}(ZA)\right)^{\dagger} \rangle 
\end{align}
The right hand side features the permutation parameterised correlator for classical witness fields of equation \eqref{eqn:bmf_in_ob_corr_trace_res}, restated here for convenience 
\begin{equation}
\label{spbback_comb_basis_corr}
\langle \mathcal{O}_{\sigma}(ZA)
(\mathcal{O}_{\tau} (ZA))^{\dagger} \rangle =\sum_{\rho \in S_{n}} \sum_{\gamma \in S_{n}} \prod_{i=1}^{n} \delta(\rho^{-1} \gamma^{-1} \sigma \gamma \tau^{-1}) \left[\text{Tr}(B^{i}) \right]^{C_{i}(\rho)} \,,
\end{equation}
where $B = A^{\dagger}A$. This can be substituted into \eqref{eqn:spbback_schur_corr} to simplify the Schur basis two-point function 
\begin{align}
\label{eqn:spbback_correlator_with_background}
\langle \mathcal{O}_{R}(ZA) \left(\mathcal{O}_{S}(ZA) \right)^{\dagger} \rangle &=  
\frac{1}{(n!)^2} \sum_{\substack{\sigma, \tau, \rho, \gamma \\  \in S_{n}}}\prod_{i=1}^{n} \chi_{R}(\sigma) \chi_{S}(\tau)  \delta(\rho^{-1} \gamma^{-1} \sigma \gamma \tau^{-1}) \left[\text{Tr}(B^{i}) \right]^{C_{i}(\rho)} \,.
\end{align}
We now use the Lemma \ref{lem:single-matrix_schur_basis_lemma} to conclude that 
\begin{equation} 
\label{BoxedResult1}
\boxed{ 
\langle \mathcal{O}_{ R }(ZA)
(\mathcal{O}_{ S } (ZA))^{\dagger} \rangle = \frac{\delta_{RS} n!}{d_{R}} \mathcal{O}_{S}(A^{ \dagger} A ) = \frac{\delta_{RS} n!}{d_{R}} \mathcal{O}_{S}(B ) 
} 
\end{equation} 
which reproduces the outcome of equation \eqref{eqn:spb_corr_result_in_schur_basis}.

\clearpage

\section{Algebras and two-matrix correlators with two-matrix witnesses }
\label{s:alg_two_mat_bmf}
Partition functions of complex  two-matrix models involve integration over complex matrices $ X , Y $. Correlators of holomorphic and anti-holomorphic polynomial functions of $ X , Y $ are of interest in connection with  the quarter-BPS sector of $\cN=4$ SYM theory (see \cite{Lewis-Brown:2020nmg} and references therein).  Invariants of degree $m $ in $X$ and degree $n$ in $Y$ can be constructed using a permutation parameterisation generalizing \eqref{eqn:aomb_gio_definition}. We define gauge invariant observables parametrised by a permutation $\sigma $ in the symmetric group $S_{ m+n} $ of all permutations of $\{ 1,2, \cdots , m+n \}$:  
\begin{equation}
\label{eqn:atmb_gio_def}
\mathcal{O}_{\sigma}(X,Y) = \text{Tr}_{V_{N}^{\otimes m+n}}(X^{\otimes m} \otimes Y^{\otimes n} \cL_{\sigma}) = X^{i_{1}}_{i_{\sigma(1)}} \dots X^{i_{m}}_{i_{\sigma(m)}}Y^{i_{m+1}}_{i_{\sigma(m+1)}} \dots Y^{i_{m+n}}_{i_{\sigma(m+n)}} \,.
\end{equation}
The indices $i_1, \cdots , i_{ m+n} $ are summed from $1$ to $N$, and operator $\mathcal{L}_{\sigma}$ has action
\begin{equation}
\label{eqn:atmb_lin_op_action}
\mathcal{L}_{\sigma} \ket{e_{i_{1}} \otimes \dots \otimes e_{i_{m+n}}} = \ket{e_{i_{\sigma(1)}} \otimes \dots \otimes e_{i_{\sigma(m+n)}}} \,.
\end{equation}
When $\gamma $ is any element of the subgroup $ ( S_m \times S_n )  \subset S_{ m+n}$ consisting of permutations which map  the subset  $ \{ 1, \cdots , m \} $ to itself and the subset  $ \{ m+1 , \cdots , m + n \} $ to itself, then by re-ordering the $X$ among each other and the $Y$ among each other, it can be shown that 
\begin{equation}
\mathcal{O}_{\gamma \sigma \gamma^{-1}}(X,Y) = \mathcal{O}_{\sigma}(X,Y)
\end{equation}
Thus the parameterising permutations are in $S_{ m+n } $ while the gauge permutations are in $S_m \times S_n \subset S_{ m+n } $ \cite{Bhattacharyya:2008rb}. The permutation centralizer algebra (PCA) $\cA ( m , n ) $ \cite{Mattioli:2016eyp} relevant to this 2-matrix problem is thus based on the equivalence classes : 
\begin{equation}
\label{eqn:atmb_equiv_rel}
\sigma \sim \gamma \sigma \gamma^{-1} \,,
\end{equation}
with $\sigma \in S_{m+n}$, $\gamma \in S_{m} \times S_{n}$. We will refer to ${\mathcal{A}} ( m , n ) $ as the  ``Necklace PCA''. It is defined as the sub-algebra of $\mathbb{C}[S_{m+n}]$ that commutes with $\mathbb{C}[S_{m} \times S_{n}]$.  There is a basis of the sub-algebra labelled by the equivalence classes (or orbits)  in $S_{m+n}$ generated by  the conjugation action by $\gamma \in S_m \times S_n$. Taking a label $p$ to run over the orbits we denote the orbits as $\text{Orb}_{\mathcal{A}}(p) $. Choosing a representative $\sigma^{(p)} $, the automorphism group $\text{Aut}_{S_{m} \times S_{n}} \left(\sigma^{\left(p\right)}\right)$ is  the subgroup of $S_{m} \times S_{n}$ which leaves  $\sigma^{\left(p\right)} \in S_{m+n}$ invariant under the action of conjugation by $\gamma \in S_{m} \times S_{n}$. The order of the automorphism group is  independent of the choice of representative $\sigma^{(p)}$ in the orbit $p$. We refer to this order as  $|\text{Aut}_{\mathcal{A}}(p)|$. By the orbit stabiliser theorem, the size of the orbit (or equivalence class)  labelled by $p$ is 
\bea 
|\text{Orb}_{\mathcal{A}}(p)| = {    m! n! \over  |\text{Aut}_{\mathcal{A}}(p)| } \, . 
\eea
Applying this notation\footnote{These $T^{\mathcal{A}}_{p}$ are referred to as ``Necklaces" in their original description. See \cite{Mattioli:2016eyp} for further explanation on the analogy.}, the  PCA basis elements for each orbit take the form
\begin{equation}
\label{eqn:atmb_algebra_element_necklace}
T^{\mathcal{A}}_{p} = \frac{1}{|\text{Aut}_{\mathcal{A}}(p)|} \sum_{\gamma \in S_{m} \times S_{n} } \gamma \sigma^{\left(p\right)} \gamma^{-1} = \sum_{\tau \in \text{Orb}_{\mathcal{A}}(p)} \tau \quad \in \mathcal{A}(m,n) \,,
\end{equation}
where the $\mathcal{A}$ label on $T^{\mathcal{A}}_{p}$, $\text{Aut}_{\mathcal{A}}(p)$ and $\text{Orb}_{\mathcal{A}}(p)$, indicate that these objects are associated with the $\mathcal{A}(m,n)$ PCA. Thus, $T^{\mathcal{A}}_{p}$ are sums over permutations in the same equivalence class/orbit, governed by relation \eqref{eqn:atmb_equiv_rel}. A notable feature of this PCA is that, unlike $\mathcal{Z}[\mathbb{C}[S_{n}]]$ in the single matrix case, it is non-commutative. We refer to these basis elements associated with orbits, and the combinatorics of group multiplications in $S_{ m+n}$, as combinatorial basis elements for $\cA ( m , n ) $.  

The two-matrix correlators using classical and coupling witness fields are now derived in both combinatorial and representation basis, following the same order as section \S\ref{s:alg_one_mat_bmf}.


\subsection{Two-point function of general operators with matrix couplings for two-matrix case}
\label{ss:gen_gauge_obs_two_mat_with_bmf}

The two-matrix analog of the one-matrix, two-point function with coupling fields, features initial coupling matrices $A_{x}$ for $X$ and $A_{y}$ for $Y$
\begin{equation}
\label{eqn:tm_bmf_path_integral_two_mat}
\Sigma[0] =  \int [dX][dY] e^{-\text{Tr}(X A_{x} X^{\dagger}) -\text{Tr}(Y A_{y} Y^{\dagger})} \,.
\end{equation}
Following a similar procedure to that of Appendix \ref{app:background_matrix_field_correlator}, the partition function can be written with vector variables and source fields, so that the basic correlators can be derived by taking derivatives with respect to said sources.
These correlators of the field variables can then be evaluated as
\begin{equation}
\label{eqn:tm_bmf_correlators_part_funcs_x}
\langle X^{i}_{j} (X^{\dagger})^{k}_{l} \rangle 
= \frac{1}{\Sigma[0]} \int [dX][dY] X^{i}_{j} (X^{\dagger})^{k}_{l} e^{-\text{Tr}(X A_{x} X^{\dagger}) -\text{Tr}(Y A_{y} Y^{\dagger})} = \delta^{i}_{l} (A_{x}^{-1})^{k}_{j} =  \delta^{i}_{l} (B_{x})^{k}_{j}
\end{equation}
\begin{equation}
\label{eqn:tm_bmf_correlators_part_funcs_y}
\langle Y^{i}_{j} (Y^{\dagger})^{k}_{l} \rangle 
= \frac{1}{\Sigma[0]} \int [dX][dY] Y^{i}_{j} (Y^{\dagger})^{k}_{l} e^{-\text{Tr}(X A_{x} X^{\dagger}) -\text{Tr}(Y A_{y} Y^{\dagger})} = \delta^{i}_{l} (A_{y}^{-1})^{k}_{j}= \delta^{i}_{l} (B_{y})^{k}_{j} 
\end{equation}
where again, $B_{x}$ and $B_{y}$ will be referred to as the coupling matrices in the subsequent discussion. Using the general permutation parameterised gauge-invariant operator \eqref{eqn:atmb_gio_def}
and its Hermitian conjugate
\begin{align}
\begin{split}
(\mathcal{O}_{\sigma}(X,Y))^{\dagger} 
&= \text{Tr}_{V_{N}^{\otimes m+n}}((X^{\dagger})^{\otimes m} \otimes (Y^{\dagger})^{\otimes n} \cL_{\sigma^{-1}})
\\
&= (X^{\dagger})^{i_{1}}_{i_{\sigma^{-1}(1)}} \dots (X^{\dagger})^{i_{m}}_{i_{\sigma^{-1}(m)}} (Y^{\dagger})^{i_{m+1}}_{i_{\sigma^{-1}(m+1)}} \dots (Y^{\dagger})^{i_{m+n}}_{i_{\sigma^{-1}(m+n)}} 
\\
& = \mathcal{O}_{\sigma^{-1}}(X^{\dagger}, Y^{\dagger}) \,,
\end{split}
\end{align}
the two-point function for the two matrix model can be constructed as follows
\begin{align}
\label{eqn:tm_gio_with_bmf1}
\langle {}&  \mathcal{O}_{\sigma_{1}}(X, Y) (\mathcal{O}_{\sigma_{2}}(X, Y))^{\dagger} \rangle \nonumber
\\
{}&=  \left\langle  \text{Tr}_{V_{N}^{\otimes m+n}} \left( X^{\otimes m} \otimes Y^{\otimes n} \cL_{\sigma_{1}} \right)  \left(\text{Tr}_{V_{N}^{\otimes m+n}} \left( X^{\otimes m} \otimes Y^{\otimes n} \cL_{\sigma_{2}}\right)\right)^{\dagger} \right\rangle  
\\
\begin{split}
{}&= \sum_{\substack{i_{1}, \dots, i_{m+n} \\ j_{1}, \dots, j_{m+n}}} \langle X^{i_{1}}_{i_{\sigma_{1}(1)}} \cdots X^{i_{m}}_{i_{\sigma_{1}(m)}} Y^{i_{m+1}}_{i_{\sigma_{1}(m+1)}} \cdots Y^{i_{m+n}}_{i_{\sigma_{1}(m+n)}}  
\\
{}& \qquad \qquad \qquad \times (X^{\dagger})^{j_{1}}_{j_{\sigma_{2}^{-1}(1)}} \cdots (X^{\dagger})^{j_{m}}_{j_{\sigma_{2}^{-1}(m)}} (Y^{\dagger})^{j_{m+1}}_{j_{\sigma_{2}^{-1}(m+1)}} \cdots (Y^{\dagger})^{j_{m+n}}_{j_{\sigma_{2}^{-1}(m+n)}}\rangle
\end{split}
\\
\begin{split}
{}&= \sum_{\substack{i_{1}, \dots, i_{m+n} \\ j_{1}, \dots, j_{m+n}}} \sum_{\gamma \in S_{m} \times S_{n}}  \delta^{i_{1}}_{j_{\gamma \sigma_{2}^{-1}(1)}} \cdots  \delta^{i_{m}}_{j_{\gamma \sigma_{2}^{-1}(m)}} \delta^{i_{m+1}}_{j_{\gamma \sigma_{2}^{-1}(m+1)}} \cdots \delta^{i_{m+n}}_{j_{\gamma \sigma_{2}^{-1}(m+n)}}
\\ 
{}& \qquad \qquad \qquad \times  (B_{x})^{j_{1}}_{i_{\gamma^{-1} \sigma_{1}(1)}} \cdots  (B_{x})^{j_{m}}_{i_{\gamma^{-1} \sigma_{1}(m)}} (B_{y})^{j_{m+1}}_{i_{\gamma^{-1} \sigma_{1}(m+1)}} \cdots (B_{y})^{j_{m+n}}_{i_{\gamma^{-1} \sigma_{1}(m+n)}} \,. \label{eqn:tm_gamma_exp}
\end{split}
\end{align}
Here a sum over $\gamma \in S_{m} \times S_{n}$ is introduced and the permutations applied to the indices, as required by Wick's theorem. $\gamma$ maps the set $\{1, \dots, m \}$ to itself, as well as $\{m+1, \dots, m+n\}$ to itself. Applying Kronecker equivariance
\begin{align}
\langle {}&  \mathcal{O}_{\sigma_{1}}(X, Y) (\mathcal{O}_{\sigma_{2}}(X, Y))^{\dagger} \rangle \nonumber
\\
\begin{split}
{}&= \sum_{\substack{i_{1}, \dots, i_{m+n} \\ j_{1}, \dots, j_{m+n}}} \sum_{\gamma \in S_{m} \times S_{n}}  \delta^{i_{\gamma^{-1} \sigma_{1}(1)}}_{j_{\gamma^{-1} \sigma_{1}\gamma \sigma_{2}^{-1}(1)}} \cdots  \delta^{i_{\gamma^{-1} \sigma_{1}(m)}}_{j_{\gamma^{-1} \sigma_{1} \gamma \sigma_{2}^{-1}(m)}} \delta^{i_{\gamma^{-1} \sigma_{1}(m+1)}}_{j_{\gamma^{-1} \sigma_{1}\gamma \sigma_{2}^{-1}(m+1)}} \cdots \delta^{i_{\gamma^{-1} \sigma_{1}(m+n)}}_{j_{\gamma^{-1} \sigma_{1}\gamma \sigma_{2}^{-1}(m+n)}} 
\\ 
{}& \qquad \qquad \qquad \times (B_{x})^{j_{1}}_{i_{\gamma^{-1} \sigma_{1}(1)}} \cdots  (B_{x})^{j_{m}}_{i_{\gamma^{-1} \sigma_{1}(m)}} (B_{y})^{j_{m+1}}_{i_{\gamma^{-1} \sigma_{1}(m+1)}} \cdots (B_{y})^{j_{m+n}}_{i_{\gamma^{-1} \sigma_{1}(m+n)}}  \label{eqn:tm_moved_sigma_2}
\end{split}
\\
\begin{split}
{}&=  \sum_{j_{1},\dots, j_{m+n}} \sum_{\gamma \in S_{m} \times S_{n}}  (B_{x})^{j_{1}}_{j_{\gamma^{-1} \sigma_{1}\gamma \sigma_{2}^{-1}(1)}} \cdots  (B_{x})^{j_{m}}_{j_{\gamma^{-1} \sigma_{1}\gamma \sigma_{2}^{-1}(m)}}  
\\
{}&  \qquad \qquad \qquad \qquad \qquad \qquad \times  (B_{y})^{j_{m+1}}_{j_{\gamma^{-1} \sigma_{1}\gamma \sigma_{2}^{-1}(m+1)}} \cdots (B_{y})^{j_{m+n}}_{j_{\gamma^{-1} \sigma_{1}\gamma \sigma_{2}^{-1}(m+n)}}\,. \label{eqn:tm_bmf_same_indices_single_line_perms} 
\end{split}
\end{align}
Therefore 
\begin{equation}
\label{eqn:tm_bmf_trace_single_line_perms}
\langle \mathcal{O}_{\sigma_{1}}(X, Y) (\mathcal{O}_{\sigma_{2}}(X, Y))^{\dagger} \rangle = \sum_{\gamma \in S_{m} \times S_{n}} \text{Tr}_{V_{N}^{\otimes m+n}}\left(B_{x}^{\otimes m} \otimes B_{y}^{\otimes n} \cL_{\gamma^{-1} \sigma_{1}\gamma \sigma_{2}^{-1}} \right) \,.
\end{equation}
The $i$ indices are contracted in \eqref{eqn:tm_bmf_same_indices_single_line_perms}. Following the trace notation of \eqref{eqn:tm_bmf_trace_single_line_perms}, the result of the calculation can be represented using the diagram of Figure \ref{fig:tm_two_perm_bmf_correlator_diagram}, recalling the composition property of the linear operators: $\cL_{\sigma}\cL_{\tau} = \cL_{\sigma \tau}$. 
\begin{figure}[htb!]
\begin{center}
\includegraphics[width=10.5cm, height=12cm]{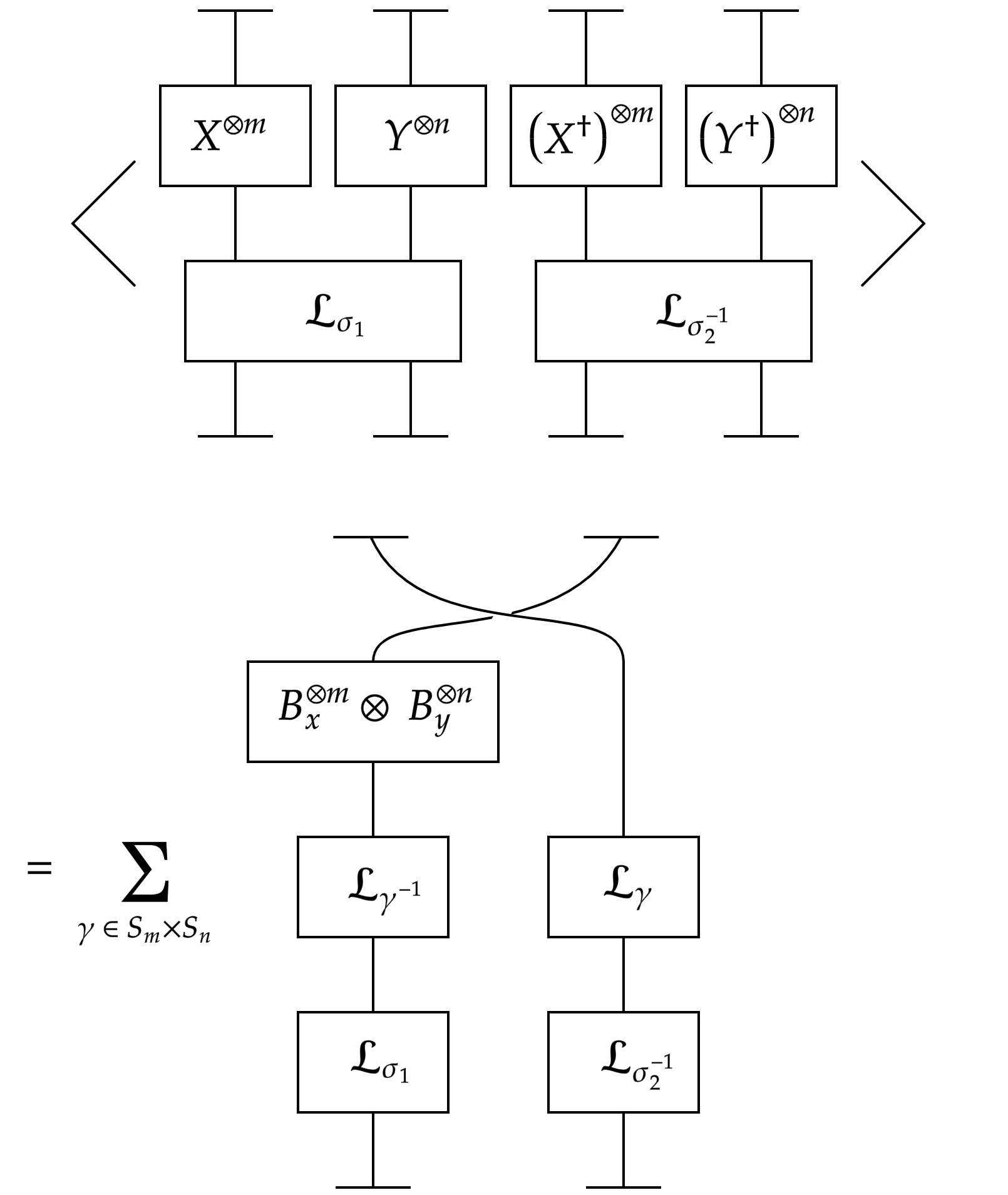}
\end{center}
\caption{The diagrammatic representation of the general two-matrix correlator with coupling matrix fields. As in the one-matrix case, this shows that the correlator can be written as a sum over permutations with a swap in the indices. Note that here $\gamma$ is an element of $S_{m} \times S_{n}$ (the direct product of symmetric groups of $m$ and $n$ objects), as required by the appropriate equivalence relation.}
\label{fig:tm_two_perm_bmf_correlator_diagram}
\end{figure}
A delta function can now be inserted and the $\sigma_{3} \in S_{m+n}$ sum decomposed into sums over equivalence classes as well as the elements therein
\begin{align}
\langle \mathcal{O}_{\sigma_{1}}(X, Y)& (\mathcal{O}_{\sigma_{2}}(X, Y))^{\dagger} \rangle \nonumber 
\\
&= \sum_{\gamma \in S_{m} \times S_{n}} \sum_{\sigma_{3} \in S_{m+n}} \text{Tr}_{V_{N}^{\otimes m+n}}\left(B_{x}^{\otimes m} \otimes B_{y}^{\otimes n} \cL_{\sigma_{3}} \right) \delta(\sigma_{3}^{-1} \gamma^{-1} \sigma_{1} \gamma \sigma_{2}^{-1}) \label{eqn:tm_bmf_trace_and_delta}  
\\
&= \sum_{\gamma \in S_{m} \times S_{n}} \sum_{\sigma_{3} \in S_{m+n}} \mathcal{O}_{\sigma_{3}}(B_{x}, B_{y}) \delta( \sigma_{3}^{-1} \gamma^{-1} \sigma_{1} \gamma \sigma_{2}^{-1}) \label{eqn:tm_bmf_trace_and_delta2_5}   
\\
&= \sum_{\gamma \in S_{m} \times S_{n}} \sum_{p_{3} } \left( \sum_{\alpha \in \text{Orb}_{\mathcal{A}}(p_{3})} \mathcal{O}_{\alpha}(B_{x}, B_{y}) \delta(\alpha^{-1}\gamma^{-1} \sigma_{1} \gamma \sigma_{2}^{-1} ) \right) \label{eqn:tm_bmf_trace_and_delta2}  
 \\
&=  \sum_{p_{3} }  \mathcal{O}_{\sigma^{\left(p_{3}\right)}}(B_{x}, B_{y})\sum_{\gamma \in S_{m} \times S_{n}} \delta\left( \sum_{\alpha \in \text{Orb}_{\mathcal{A}}(p_{3})} \alpha^{-1} \gamma^{-1} \sigma_{1} \gamma \sigma_{2}^{-1}  \right)\label{eqn:tm_bmf_trace_and_delta3}  
\\
&=  \sum_{p_{3} }  \mathcal{O}_{\sigma^{\left(p_{3}\right)}}(B_{x}, B_{y})\sum_{\gamma \in S_{m} \times S_{n}} \delta \left( T^{\mathcal{A}}_{p'_{3}}  \gamma^{-1} \sigma_{1} \gamma \sigma_{2}^{-1}  \right) \label{eqn:tm_bmf_trace_and_delta4} \,,
\end{align}
where $p_{3}$ labels the equivalence classes/orbits of the $\sigma_{3}$ sum decomposition. Here we set $  \text{Tr}_{V_{N}^{\otimes m+n}}\left(B_{x}^{\otimes m} \otimes B_{y}^{\otimes n} \cL_{\sigma_{3}} \right)=\mathcal{O}_{\sigma_{3}}(B_{x}, B_{y})$ as per equation \eqref{eqn:atmb_gio_def} and the definition of $T^{\mathcal{A}}_{p'_{3}} \in \mathcal{A}(m,n)$ was used in \eqref{eqn:tm_bmf_trace_and_delta4}. Note that given the permutation is $\alpha^{-1}$ and the sum is over $\alpha \in \text{Orb}_{\mathcal{A}}(p_{3})$, we attach a prime to the partition/orbit/class label to indicate that $\text{Orb}(p'_{3})$ is the orbit constructed from the inverse permutations of $\text{Orb}(p_{3})$.\footnote{While in the one-matrix case of \S\ref{ss:background_matrix_field}, permutations and their inverses share the same conjugacy class, meaning the prime label can be dropped, for the two-matrix and multi-matrix models, this is not generally the case.} In equation \eqref{eqn:tm_bmf_trace_and_delta3}, the $\alpha$ label on the GIO has been replaced by $\sigma^{\left(p_{3}\right)}$ to help indicate that it is a function of the $p_{3}$ partition/class. To identify how the PCA structure constants emerge as part of the correlator, the following lemma is provided.
\begin{lemma} 
\label{lem:two-matrix_lemma}
For $\gamma \in S_{m} \times S_{n}$, $\sigma_{i} \in S_{m+n}$ and $T^{\mathcal{A}}_{p_{i}} \in \mathcal{A}(m,n)$ the following equality holds 
\begin{equation}
\sum_{\gamma \in S_{m} \times S_{n}} \delta \left( T^{\mathcal{A}}_{p'_{3}} \gamma^{-1} \sigma_{1} \gamma \sigma_{2}^{-1} \right)
=\frac{m! n!|T^{\mathcal{A}}_{p'_{3}}|}{|T^{\mathcal{A}}_{p_{1}}| |T^{\mathcal{A}}_{p'_{2}}|} C^{p'_{3}; \mathcal{A}}_{ p_{1} p'_{2}} 
\end{equation}
where $C^{p'_{3}; \mathcal{A}}_{ p_{1} p'_{2}} $ is a structure constant of the $\mathcal{A}(m,n)$ PCA and $\sigma_{1}$, $\sigma_{2}^{-1}$ belong to equivalence classes labelled by $p_{1}$, $p'_{2}$ respectively.
\end{lemma} 
\begin{proof}

\begin{align}
\sum_{\gamma \in S_{m} \times S_{n}} &\delta(T^{\mathcal{A}}_{p'_{3}} \gamma^{-1} \sigma_{1} \gamma \sigma_{2}^{-1})  \nonumber
\\
&= \sum_{\mu_{1} \in S_{m} \times S_{n}} \delta\left(T^{\mathcal{A}}_{p'_{3}} (\mu_{1} \mu_{2})^{-1} \sigma_{1} (\mu_{1} \mu_{2}) \sigma_{2}^{-1}\right) \label{eqn:tmat_lem_relabel_gam}
\\
&= \sum_{\mu_{1} \in S_{m} \times S_{n}} \delta\left((\mu_{2}^{-1} \mu_{2})T^{\mathcal{A}}_{p'_{3}} \mu_{2}^{-1} \mu_{1}^{-1} \sigma_{1} \mu_{1} \mu_{2} \sigma_{2}^{-1}\right) \label{eqn:tmat_lem_insert_id}
\\
&=  \sum_{\mu_{1} \in S_{m} \times S_{n}} \delta \left( \underbrace{\mu_{2}T^{\mathcal{A}}_{p'_{3}} \mu_{2}^{-1}}_{=T^{\mathcal{A}}_{p'_{3}}} \mu_{1}^{-1} \sigma_{1} \mu_{1} \mu_{2} \sigma_{2}^{-1} \mu_{2}^{-1}\right)\label{eqn:tmat_lem_cycles_mu2inv}
\\
&=  \frac{1}{m!n!} \sum_{\mu_{1},\mu_{2} \in S_{m} \times S_{n}} \delta \left( T^{\mathcal{A}}_{p'_{3}} (\mu_{1}^{-1} \sigma_{1} \mu_{1}) (\mu_{2} \sigma_{2}^{-1} \mu_{2}^{-1}) \right)\label{eqn:tmat_lem_sum_over_mu2_added}
\\
&=  \frac{1}{m!n!}  \delta \left( T^{\mathcal{A}}_{p'_{3}} \left(\sum_{\mu_{1} \in S_{m} \times S_{n}}\mu_{1}^{-1} \sigma_{1} \mu_{1}\right) \left(\sum_{\mu_{2} \in S_{m} \times S_{n}} \mu_{2} \sigma_{2}^{-1} \mu_{2}^{-1}\right) \right) \label{eqn:tmat_lem_collecting_sums_in_delta}
\\
&= \frac{| \text{Aut}_{\mathcal{A}}\left(p_{1}\right)|| \text{Aut}_{\mathcal{A}}\left(p'_{2}\right)|}{m!n!}  \delta \left( T^{\mathcal{A}}_{p'_{3}} T^{\mathcal{A}}_{p_{1}} T^{\mathcal{A}}_{p'_{2}} \right) \label{eqn:tmat_lem_eles_and_autos_included}
\\
&= \frac{m!n!}{|T^{\mathcal{A}}_{p_{1}}||T^{\mathcal{A}}_{p'_{2}}|}  \delta \left( T^{\mathcal{A}}_{p'_{3}} T^{\mathcal{A}}_{p_{1}} T^{\mathcal{A}}_{p'_{2}} \right)
\label{eqn:tmat_lem_all_pca_eles_prime_removed} \,.
\end{align}
In the above, equation \eqref{eqn:tmat_lem_relabel_gam} replaces the sum over $\gamma$ for a sum over $\mu_{1}$ with $\gamma \to \mu_{1}\mu_{2}$. The identity element is added in \eqref{eqn:tmat_lem_insert_id}, followed by a rearranging of $\mu_{2}^{-1}$ and the identification of $ T^{\mathcal{A}}_{p'_{3}} = \mu_{2} T^{\mathcal{A}}_{p'_{3}} \mu_{2}^{-1}$ . A sum over $\mu_{2} \in S_{m} \times S_{n}$ is introduced in \eqref{eqn:tmat_lem_sum_over_mu2_added} along with required factor $1/(m!n!)$. Finally, between equations \eqref{eqn:tmat_lem_collecting_sums_in_delta} and \eqref{eqn:tmat_lem_all_pca_eles_prime_removed}, the sums are brought inside the delta function, the PCA elements and automorphism group size factors are introduced, and the orbit-stabiliser theorem is used to simplify the expression. Again, the prime labels are used to represent the partitions/classes of which an inverse permutation belongs, and we have chosen that $\sigma_{1}$ and $\sigma_{2}^{-1}$ belong to partitions/equivalence classes labelled by $p_{1}$ and $p'_{2}$ respectively.

The remaining delta function can be reduced using the multiplication properties of the PCA
\begin{equation}
\label{eqn:tm_delta_struct_cons_rel}
\delta(T^{\mathcal{A}}_{p'_{3}} T^{\mathcal{A}}_{p_{1}} T^{\mathcal{A}}_{p'_{2}}) =  \sum_{p_{k} }C^{p_{k};\mathcal{A}}_{p_{1}p'_{2}} \delta(T^{\mathcal{A}}_{p'_{3}} T^{\mathcal{A}}_{p_{k}} ) =   \sum_{p_{k}}C^{p_{k};\mathcal{A}}_{p_{1}p'_{2}} \delta_{p_{k} p'_{3}} |T^{\mathcal{A}}_{p_{k}}| =  |T^{\mathcal{A}}_{p'_{3}}|C^{p'_{3}; \mathcal{A}}_{p_{1}p'_{2}} \,.
\end{equation}
Plugging this back into equation \eqref{eqn:tmat_lem_all_pca_eles_prime_removed} produces the required result
\begin{equation}
\sum_{\gamma \in S_{m} \times S_{n}} \delta(T^{\mathcal{A}}_{p'_{3}} \gamma^{-1} \sigma_{1} \gamma \sigma_{2}^{-1}) = \frac{m!n!}{|T^{\mathcal{A}}_{p_{1}}||T^{\mathcal{A}}_{p'_{2}}|}  |T^{\mathcal{A}}_{p'_{3}}|C^{p'_{3}; \mathcal{A}}_{p_{1}p'_{2}} \,.
\end{equation}
\end{proof}
Combining the lemma above and equation \eqref{eqn:tm_bmf_trace_and_delta4}, the correlator is
\begin{equation}
\langle \mathcal{O}_{\sigma_{1}}(X, Y) (\mathcal{O}_{\sigma_{2}}(X, Y))^{\dagger} \rangle = \sum_{p_{3}} \frac{m!n!|T^{\mathcal{A}}_{p'_{3}}|}{|T^{\mathcal{A}}_{p_{1}}||T^{\mathcal{A}}_{p'_{2}}|}  C^{p'_{3}; \mathcal{A}}_{p_{1}p'_{2}} \mathcal{O}_{\sigma^{\left(p_{3}\right)}}(B_{x}, B_{y})\,.
\end{equation}
Finally, using the facts that $|T^{\mathcal{A}}_{p'_{i}}| = |T^{\mathcal{A}}_{p_{i}}|$, $\mathcal{O}_{\sigma_{1}} \equiv \mathcal{O}_{\sigma^{\left(p_{1}\right)}}$ and $\mathcal{O}_{\sigma_{2}^{-1}} \equiv \mathcal{O}_{\sigma^{(p'_{2})}}$, then rearranging the orbit size factors and applying $|T_{p_{i}}^{\mathcal{A}}|\mathcal{O}_{\sigma^{\left( p_{i}\right)}} = \mathcal{O}_{T^{\mathcal{A}}_{p_{i}}} \equiv \mathcal{O}_{p_{i}}$, we achieve the combinatorial basis representation of the correlator
\begin{equation}
\label{eqn:tm_bmf_simplified_final_corr_boxed}
\boxed{
\langle \mathcal{O}_{p_{1}}(X, Y) (\mathcal{O}_{p_{2}}(X, Y))^{\dagger} \rangle = m!n!\sum_{p_{3}}  C^{p'_{3}; \mathcal{A}}_{p_{1}p'_{2}} \mathcal{O}_{p_{3}}(B_{x}, B_{y})}
\end{equation}
Akin to the result of the one-matrix correlator with coupling matrix field of \eqref{eqn:aomb_correlator_result_boxed}, this two-matrix result shows that the insertion of combinatorial basis GIOs with fluctuating/quantum fields $X$ and $Y$, is equal to a linear combination of operators in the associated witness fields $B_{x}$ and $B_{y}$. As such, the structure constants can be evaluated by choosing the basis labels $(p_{1}, p'_{2}, p'_{3})$ for the quantum and witness field operators. It is worth noting that the structure constants $C^{p'_{3}; \mathcal{A}}_{p_{1}p'_{2}}$ are integer valued.


\subsection{Fourier/\texorpdfstring{$Q$}{Q}-basis for the two-matrix, two-point function}
\label{ss:two_matrix_schur_basis}	
As explained at the start of this section, the two-matrix observables can be enumerated using equivalence classes of permutations $ \sigma \in S_{ m+n} $ under the equivalence relation   $ \sigma \sim \gamma \sigma \gamma^{-1}$ where  $\gamma \in S_{m} \times S_{n} \subset S_{ m+n} $.
Following \cite{Balasubramanian:2002sa,Balasubramanian:2004nb,Bhattacharyya:2008rb,Pasukonis:2013ts,Ramgoolam:2016ciq,Mattioli:2016eyp,}, it will be useful to consider the decomposition of the irreducible representation $V_{R}^{(S_{m+n})} $ of $S_{ m+n}$ labelled by the Young diagram $R$ having $m+n$ boxes as a direct sum of  irreps $V_{R_{1}}^{(S_{m})} \otimes V_{R_{2}}^{(S_{n})}$ of the product group $S_{ m } \times S_n \subset S_{ m+n}$, where $R_1, R_2 $ have $m $ and $n$ boxes respectively.  This takes the form 
\begin{equation}
\label{eqn:spb_vector_decomp}
V_{R}^{(S_{m+n})} = \bigoplus_{R_{1}, R_{2}} V_{R_{1}}^{(S_{m})} \otimes V_{R_{2}}^{(S_{n})} \otimes V_{R_{1}, R_{2}}^{R}
\end{equation}
where $V_{R_{1},R_{2}}^{R}$ is the ``multiplicity space" with dimension equal to the number of times 
$V_{R_{1}}^{(S_{m})} \otimes V_{R_{2}}^{(S_{n})} $ appears in the decomposition. This dimension is the 
 Littlewood-Richardson coefficient denoted $ g ( R_1 , R_2 , R ) $. The states of this subgroup basis are labelled by
\begin{equation}
\label{eqn:spb_subgroup_basis_state}
\ket{R_{1}, R_{2}, m_{1}, m_{2}; \nu}
\end{equation}
where the $m_{1}$ and $m_{2}$ label the states of $V_{R_{1}}^{(S_{m})}$ and $V_{R_{2}}^{(S_{n})}$ respectively, while $\nu$ is the index which runs over a basis for the multiplicity  space $V_{R_{1}, R_{2}}^{R}$. As in the single matrix case, Schur polynomial operators can be built using characters. For the two-matrix case specifically, these operators are known as \textit{restricted Schur polynomials} \cite{Bhattacharyya:2008rb,deMelloKoch:2007rqf,Kimura:2008ac,Bhattacharyya:2008xy,phdthesis_steph,Bekker:2007ea,Collins:2008gc} due to the choice of subgroup basis and formed using the \textit{restricted character}:
\begin{equation}
\chi^{R}_{R_{1},R_{2},\mu,\nu}(\sigma)
\end{equation}
where again, the integers $\mu,\nu$ run over the multiplicity $g(R_1, R_2, R)$ of the branching $R \rightarrow R_1 \otimes R_2$:
$ 1 \leq \nu_{1}; \nu_{2} \leq g(R_1, R_2, R)$, and there is an associated Young diagram for each representation $R\,,R_{1}$ and $R_{2}$. The restricted character is analogous to $\chi_{R}(\sigma)$ in \eqref{eqn:spb_complex_character_schur}, but equipped with the necessary indices inherited from the decomposition of \eqref{eqn:spb_vector_decomp}. Using the projector-like operator 
\begin{equation}
\label{eqn:spb_projection_interwining_op}
P^{R}_{R_{1},R_{2},\mu,\nu} = \sum_{m_{1}}^{d_{R_{1}}}  \sum_{m_{2}}^{d_{R_{2}}} \ket{R_{1}, R_{2}, m_{1}, m_{2}; \mu} \bra{R_{1}, R_{2}, m_{1}, m_{2}; \nu}	
\end{equation}
where $d_{R_{i}}$ is the dimension of irrep $R_{i}$, the restricted characters are explicitly defined as
\begin{equation}
\label{eqn:spb_restricted_char_def}
\chi^{R}_{R_{1}, R_{2}, \mu, \nu}(\sigma)  = \text{Tr}_{V^{(R_{1},R_{2})}}\left[ P^{R}_{R_{1},R_{2},\mu,\nu} \left(D^{R}(\sigma)\right) \right] \,.
\end{equation}
Here, $D^{R}(\sigma)$ is the representation matrix of $\sigma$ in representation $R$ and the trace is taken over the subspace corresponding to the subduction $(R_{1}, R_{2})$ of $R$.
The projection-like operator's components in turn may be written in terms of branching coefficients as
\begin{equation}
\left(P^{R}_{R_{1},R_{2},\mu, \nu}\right)^{ij} = \sum_{m_{1},m_{2}} B^{R; i}_{R_{1}, R_{2}, \mu ;m_{1}, m_{2}} B^{R; j}_{R_{1}, R_{2}, \nu; m_{1}, m_{2}} 
\end{equation}
where the branching coefficients, $B^{R; i}_{R_{a}, R_{b}, \nu;m_{a}, m_{b}}$, are defined as the components of the vector $\ket{R_{a},R_{b},m_{a},m_{b}; \nu}$ in any given orthogonal basis for $R$: $B^{R; i}_{R_{a}, R_{b}, \nu ; m_{a}, m_{b}} = \braket{R;i|R_{a}, R_{b}, m_{a}, m_{b}; \nu}$.

The final step in producing the Schur polynomial operators, comes via the Wedderburn-Artin theorem, which states that the PCA $\mathcal{A}(m,n)$ may be decomposed as a sum of matrix algebras
\begin{equation}
\mathcal{A}(m,n) = \bigoplus_{\substack{R \\ R_{1},R_{2}}} \text{Span}\{Q^{R}_{R_{1},R_{2},\mu,\nu}; \mu, \nu 	 \}
\end{equation}
where these $Q$-basis elements, using the above terminology, are defined as 
\begin{align}
\label{eqn:q_basis_tm_definition}
Q^{R}_{R_{1},R_{2},\mu,\nu} = \frac{d_{R}}{(m+n)!}\sum_{\sigma \in S_{m+n}} \chi^{R}_{R_{1},R_{2},\mu,\nu} (\sigma) \sigma^{-1} \,,
\end{align}
where $d_{R}$ is the dimension of representation $R$ of $S_{m+n}$, and they have the property that they multiply as matrices in the multiplicity indices
\begin{equation}
\label{eqn:q_basis_tm_multiplication_rule}
Q^{R}_{R_{1},R_{2},\mu_{1},\nu_{1}}Q^{R}_{R_{1},R_{2},\mu_{2},\nu_{2}} =  \delta^{RS}\delta_{R_{1} S_{1}} \delta_{R_{2} S_{2}} \delta_{\nu_{1}\mu_{2}}  Q^{R}_{R_{1},R_{2},\mu_{1},\nu_{2}} \,.
\end{equation}
This is derived by generalising the multi-matrix model result in Appendix \ref{app:redone_wa_q_basis}, to the two-matrix model. Therefore, the associated restricted Schur polynomial operators, which shall be equivalently called $Q$-basis/Fourier basis operators throughout, can be formed by Fourier transforming the permutation parameterised operators of \eqref{eqn:atmb_gio_def} as follows
\begin{equation}
\label{eqn:q_basis_tm_operator_definition}
\mathcal{O}^{R}_{R_{1},R_{2},\mu, \nu}(X,Y)  = \sum_{\sigma \in S_{m+n}} \delta\left(Q^{R}_{R_{1},R_{2},\mu,\nu} \sigma^{-1} \right)\mathcal{O}_{\sigma}(X,Y) \,.
\end{equation}
Its Hermitian conjugate is 
\begin{align}
\label{eqn:q_basis_tm_operator_definition_conj}
\left(\mathcal{O}^{R}_{R_{1},R_{2},\mu, \nu}(X,Y) \right)^{\dagger}  &=  \sum_{\sigma\in S_{m+n}} \delta\left(Q^{R}_{R_{1},R_{2},\nu,\mu} \sigma \right)\mathcal{O}_{\sigma^{-1}}(X^{\dagger},Y^{\dagger})
\\
&=  \sum_{\tilde{\sigma}\in S_{m+n}} \delta\left(Q^{R}_{R_{1},R_{2},\nu,\mu} \tilde{\sigma}^{-1} \right)\mathcal{O}_{\tilde{\sigma}}(X^{\dagger},Y^{\dagger})
\\
&= \mathcal{O}^{R}_{R_{1},R_{2},\nu, \mu}(X^{\dagger},Y^{\dagger})
\,,
\end{align}
where properties $(Q^{R}_{R_{1},R_{2},\mu,\nu})^{\dagger} = Q^{R}_{R_{1},R_{2},\nu,\mu}$ (see Appendix \ref{app:redone_wa_q_basis}) and $\sigma^{\dagger} = \sigma^{-1}$ were used, and $\sigma \to \tilde{\sigma} = \sigma^{-1}$ was applied via sum invariance. The corresponding two-point function is hence
\begin{align}
{}&\langle \mathcal{O}^{R}_{R_{1},R_{2},\mu_{1}, \nu_{1}}(X,Y) \left(\mathcal{O}^{S}_{S_{1},S_{2},\mu_{2}, \nu_{2}}(X,Y) \right)^{\dagger}  \rangle  \nonumber
\\
{}&= \langle \mathcal{O}^{R}_{R_{1},R_{2},\mu_{1}, \nu_{1}}(X,Y) \mathcal{O}^{S}_{S_{1},S_{2},\nu_{2}, \mu_{2}}(X^{\dagger},Y^{\dagger})  \rangle 
\\
{}&= \sum_{\sigma,\tau \in S_{m+n}} \delta(Q^{R}_{R_{1},R_{2},\mu_{1}, \nu_{1}} \sigma^{-1}) \delta(Q^{S}_{S_{1},S_{2},\nu_{2}, \mu_{2}} \tau) \langle \mathcal{O}_{\sigma}(X,Y) \mathcal{O}_{\tau^{-1}}(X^{\dagger},Y^{\dagger}) \rangle 
\\
{}&= \sum_{\sigma,\tau, \rho \in S_{m+n}} \sum_{\gamma \in S_{m}\times S_{n}} \delta(Q^{R}_{R_{1},R_{2},\mu_{1}, \nu_{1}} \sigma^{-1}) \delta(Q^{S}_{S_{1},S_{2},\nu_{2}, \mu_{2}} \tau)   \mathcal{O}_{\rho}(B_{x},B_{y}) \delta(\rho^{-1} \gamma^{-1} \sigma \gamma \tau^{-1})\label{eqn:q_basis_tm_corr_with_perm_operators}
\end{align}
Equation \eqref{eqn:q_basis_tm_corr_with_perm_operators} introduces the permutation parameterised correlator result from \eqref{eqn:tm_bmf_trace_and_delta2_5}. Here we introduce a lemma to simplify the above expression
\begin{lemma}
\label{lem:two-matrix_lemma_fourier}
For $\sigma,\tau,\rho \in S_{m+n}$, $\gamma \in S_{m} \times S_{n}$, irreps $R,R_{1}, R_{2}$ and multiplicity indices $\mu_{1,2}, \nu_{1,2}$, the following expression holds 
\begin{align}
{}& \sum_{\sigma,\tau, \rho \in S_{m+n}} \sum_{\gamma \in S_{m}\times S_{n}} \delta(Q^{R}_{R_{1},R_{2},\mu_{1}, \nu_{1}} \sigma^{-1}) \delta(Q^{S}_{S_{1},S_{2},\nu_{2}, \mu_{2}} \tau)   \mathcal{O}_{\rho}(B_{x},B_{y}) \delta(\rho^{-1} \gamma^{-1} \sigma \gamma \tau^{-1}) \nonumber
\\
{}&= m!n! \delta^{RS}\delta_{R_{1}S_{1}}\delta_{R_{2}S_{2}} \delta_{\nu_{1}\nu_{2}} \mathcal{O}^{R}_{R_{1},R_{2},\mu_{1},\mu_{2}}(B_{x},B_{y})
\end{align}
\end{lemma}
\begin{proof}
\begin{align}
{}& \sum_{\sigma,\tau, \rho \in S_{m+n}} \sum_{\gamma \in S_{m}\times S_{n}} \delta(Q^{R}_{R_{1},R_{2},\mu_{1}, \nu_{1}} \sigma^{-1}) \delta(Q^{S}_{S_{1},S_{2},\nu_{2}, \mu_{2}} \tau)   \mathcal{O}_{\rho}(B_{x},B_{y}) \delta(\rho^{-1} \gamma^{-1} \sigma \gamma \tau^{-1}) \label{eqn:q_basis_tm_sub_in_perm_corr}
\\
{}&= \sum_{\rho\in S_{m+n}} \sum_{\gamma \in S_{m}\times S_{n}}  \mathcal{O}_{\rho}(B_{x},B_{y}) \delta(\rho^{-1} \gamma^{-1} Q^{R}_{R_{1},R_{2},\mu_{1}, \nu_{1}} \gamma Q^{S}_{S_{1},S_{2},\nu_{2}, \mu_{2}}) \label{eqn:q_basis_tm_sum_over_q_deltas}
\\
{}&= m!n! \sum_{\rho\in S_{m+n}} \mathcal{O}_{\rho}(B_{x},B_{y}) \delta(\rho^{-1} Q^{R}_{R_{1},R_{2},\mu_{1}, \nu_{1}} Q^{S}_{S_{1},S_{2},\nu_{2}, \mu_{2}}) \label{eqn:q_basis_tm_conj_invar}
\\
{}&= m!n! \delta^{RS}\delta_{R_{1}S_{1}}\delta_{R_{2}S_{2}} \delta_{\nu_{1}\nu_{2}} \left[ \sum_{\rho\in S_{m+n}}  \delta(Q^{R}_{R_{1},R_{2},\mu_{1}, \mu_{2}}\rho^{-1} )\mathcal{O}_{\rho}(B_{x},B_{y})\right] \label{eqn:q_basis_tm_basis_multiplication}
\\
{}&= m!n! \delta^{RS}\delta_{R_{1}S_{1}}\delta_{R_{2}S_{2}} \delta_{\nu_{1}\nu_{2}} \mathcal{O}^{R}_{R_{1},R_{2},\mu_{1},\mu_{2}}(B_{x},B_{y}) \label{eqn:q_basis_tm_q_operator_emerges} \,.
\end{align}
Equation \eqref{eqn:q_basis_tm_sum_over_q_deltas} computed the sum over the $Q$-basis delta functions, replacing the $\sigma$ and $\tau^{-1}$ in the final delta function. Equation \eqref{eqn:q_basis_tm_conj_invar} makes use of the invariance of $Q$-basis elements under conjugation by $\gamma \in S_{m} \times S_{n}$, while \eqref{eqn:q_basis_tm_basis_multiplication} is obtained by utilising their multiplication property \eqref{eqn:q_basis_tm_multiplication_rule}. Finally, \eqref{eqn:q_basis_tm_q_operator_emerges} yields a $Q$-basis operator in witness fields, following definition \eqref{eqn:q_basis_tm_operator_definition}. 
\end{proof}

Applying lemma \ref{lem:two-matrix_lemma_fourier} to equation \eqref{eqn:q_basis_tm_corr_with_perm_operators}, the two-matrix generalisation to that of the one-matrix outcome in Section \ref{ss:single_matrix_schur_basis}, is therefore
\begin{equation}
	\boxed{
	\begin{aligned}
\label{eqn:two_matrix_schur_basis_corr2}
\left\langle  \mathcal{O}^{R}_{R_{1}, R_{2}, \mu_{1}, \nu_{1}}(X, Y) \right. & \left. (\mathcal{O}^{S}_{S_{1}, S_{2}, \mu_{2}, \nu_{2}}(X, Y))^{\dagger} \right\rangle =
\\
&  m!n! \delta^{RS}\delta_{R_{1}S_{1}}\delta_{R_{2}S_{2}} \delta_{\nu_{1}\nu_{2}} \mathcal{O}^{R}_{R_{1},R_{2},\mu_{1},\mu_{2}}(B_{x},B_{y})
	\end{aligned}
	}
\end{equation}
This shows that the correlator of two Fourier basis operators in two matrices $X$ and $Y$, is orthogonal in its representations, and is proportional to a Fourier basis operator made of the witness fields, $B_{x}$ and $B_{y}$. 


\subsection{Observable functions of quantum and classical fields for two-matrix case}
\label{ss:back_fields_in_observables_for_tm}

If the observables are defined to include classical witness fields, and the action defined to have no coupling, then the same result can be acquired. Define the gauge invariant operator as
\begin{align}
\begin{split}
\label{eqntm_:bmf_in_ob_obs_definition}
\mathcal{O}_{\sigma}(XA_{x}, YA_{y}) &= 
\text{Tr}_{V^{\otimes m+n}_{N}}\left((XA_{x})^{\otimes m} \otimes (YA_{y})^{\otimes n} \cL_{\sigma} \right)
\\
&=(XA_{x})^{i_{1}}_{i_{\sigma(1)}} \dots(XA_{x})^{i_{m}}_{i_{\sigma(m)}}  (YA_{y})^{i_{m+1}}_{i_{\sigma(m+1)}} \dots (YA_{y})^{i_{m+n}}_{i_{\sigma(m+n)}}
\end{split}
\end{align}
and its Hermitian conjugate
\begin{align}
\begin{split}
\label{eqn:tm_bmf_in_ob_obs_definition_herm}
(\mathcal{O}_{\sigma}&(XA_{x}, YA_{y}))^{\dagger}
\\ 
&= \left(\text{Tr}_{V^{\otimes m+n}_{N}}\left( (XA_{x})^{\otimes m} \otimes (YA_{y})^{\otimes n} \cL_{\sigma}\right) \right)^{\dagger}
\\
&= (A_{x}^{\dagger} X^{\dagger})^{i_{1}}_{i_{\sigma^{-1}(1)}} \dots (A_{x}^{\dagger} X^{\dagger})^{i_{m}}_{i_{\sigma^{-1}(m)}}  (A_{y}^{\dagger} Y^{\dagger})^{i_{m+1}}_{i_{\sigma^{-1}(m+1)}} \dots (A_{y}^{\dagger} Y^{\dagger})^{i_{m+n}}_{i_{\sigma^{-1}(m+n)}} 
\\
&= \mathcal{O}_{\sigma^{-1}}(A_{x}^{\dagger}X^{\dagger}, A_{y}^{\dagger}Y^{\dagger})
\,.
\end{split}
\end{align}
where $A_{x}$ and $A_{y}$ are the initially defined classical witness fields. Using the partition function 
\begin{equation}
\label{eqn:tmob_no_bmf_partition_func}
\Sigma[0] = \int [dX][dY] e^{ - \text{Tr}\left(XX^{\dagger} \right) -\text{Tr}\left(YY^{\dagger} \right)}
\end{equation}
as well as the basic field correlators
\begin{align}
\begin{split}
\langle X^{i}_{j} (X^{\dagger})^{k}_{l} \rangle &= \frac{1}{\Sigma[0]}\int [dX][dY]  X^{i}_{j} (X^{\dagger})^{k}_{l} e^{-\text{Tr}(XX^{\dagger})-\text{Tr}(YY^{\dagger})} = \delta^{i}_{l} \delta^{k}_{j}\,,
\\
\langle Y^{i}_{j} (Y^{\dagger})^{k}_{l} \rangle &= \frac{1}{\Sigma[0]}\int [dX][dY]  Y^{i}_{j} (Y^{\dagger})^{k}_{l} e^{-\text{Tr}(XX^{\dagger})-\text{Tr}(YY^{\dagger})} = \delta^{i}_{l} \delta^{k}_{j}\,,
\end{split}
\end{align}
where all other basic field correlators vanish, one may calculate that the correlator of these GIOs is
\begin{align}
\langle {}& \mathcal{O}_{\sigma_{1}}(XA_{x}, YA_{y}) (\mathcal{O}_{\sigma_{2}}(XA_{x}, YA_{y}))^{\dagger}  \rangle \nonumber
\\
{}&= \frac{1}{\Sigma[0]} \int [dX][dY] \mathcal{O}_{\sigma_{1}}(XA_{x},YA_{y}) \left(\mathcal{O}_{\sigma_{2}}(XA_{x},YA_{y})\right)^{\dagger} e^{ - \text{Tr}\left(XX^{\dagger} \right) -\text{Tr}\left(YY^{\dagger} \right)} 
\\
\begin{split}
{}&= \sum_{I,J,K,L} (A_{x}^{\dagger})^{j_{1}}_{l_{1}} (A_{x})^{k_{1}}_{i_{\sigma_{1}(1)}}  \dots (A_{x}^{\dagger})^{j_{m}}_{l_{m}}(A_{x})^{k_{m}}_{i_{\sigma_{1}(m)}} ~~ (A_{y}^{\dagger})^{j_{m+1}}_{l_{m+1}}  (A_{y})^{k_{m+1}}_{i_{\sigma_{1}(m+1)}}  \dots  (A_{y}^{\dagger})^{j_{m+n}}_{l_{m+n}} (A_{y})^{k_{m+n}}_{i_{\sigma_{1}(m+n)}} 
\\
{}& \quad \times  \frac{1}{\Sigma[0]} \int [dX][dY] X^{i_{1}}_{k_{1}} \dots  X^{i_{m}}_{k_{m}}  Y^{i_{m+1}}_{k_{m+1}}  \dots Y^{i_{m+n}}_{k_{m+n}}  
\\
{}& \quad \times (X^{\dagger})^{l_{1}}_{j_{\sigma_{2}^{-1}(1)}} \dots (X^{\dagger})^{l_{m}}_{j_{\sigma_{2}^{-1}(m)}}  (Y^{\dagger})^{l_{m+1}}_{j_{\sigma_{2}^{-1}(m+1)}} \dots (Y^{\dagger})^{l_{m+n}}_{j_{\sigma_{2}^{-1}(m+n)}} 
e^{ - \text{Tr}\left(XX^{\dagger} \right) -\text{Tr}\left(YY^{\dagger} \right)} 
\end{split} 
\label{eqn:tm_bmf_in_ob_corr}
\end{align}
where the sum over $I,J,K,L$ represents a sum over all associated matrix indices. For notational convenience we define a function in the witness fields
\begin{multline}
f_{x,y}(A, A^{\dagger}; \vec{i}, \vec{j}, \vec{k},\vec{l}; \sigma_{1} ) = (A_{x}^{\dagger})^{j_{1}}_{l_{1}} (A_{x})^{k_{1}}_{i_{\sigma_{1}(1)}}  \dots (A_{x}^{\dagger})^{j_{m}}_{l_{m}}(A_{x})^{k_{m}}_{i_{\sigma_{1}(m)}} 
\\
\times (A_{y}^{\dagger})^{j_{m+1}}_{l_{m+1}}  (A_{y})^{k_{m+1}}_{i_{\sigma_{1}(m+1)}} \dots 
 (A_{y}^{\dagger})^{j_{m+n}}_{l_{m+n}} (A_{y})^{k_{m+n}}_{i_{\sigma_{1}(m+n)}} \,,
\end{multline}
where, for example, all possible $i$ matrix indices are denoted as vector $\vec{i}$ components. With this definition, the correlator becomes
\begin{align}
\langle {}& \mathcal{O}_{\sigma_{1}}(XA_{x}, YA_{y}) (\mathcal{O}_{\sigma_{2}}(XA_{x}, YA_{y}))^{\dagger}  \rangle \nonumber
\\
\begin{split}
{}&= \sum_{I,J,K,L} f_{x,y}(A, A^{\dagger}; \vec{i}, \vec{j}, \vec{k},\vec{l}; \sigma_{1} )  \frac{1}{\Sigma[0]} \int [dX][dY] X^{i_{1}}_{k_{1}} \dots  X^{i_{m}}_{k_{m}}  Y^{i_{m+1}}_{k_{m+1}}  \dots Y^{i_{m+n}}_{k_{m+n}} 
\\
{}& \quad \times (X^{\dagger})^{l_{1}}_{j_{\sigma_{2}^{-1}(1)}} \dots    (X^{\dagger})^{l_{m}}_{j_{\sigma_{2}^{-1}(m)}}   (Y^{\dagger})^{l_{m+1}}_{j_{\sigma_{2}^{-1}(m+1)}} \dots (Y^{\dagger})^{l_{m+n}}_{j_{\sigma_{2}^{-1}(m+n)}} 
e^{ - \text{Tr}\left(XX^{\dagger} \right) -\text{Tr}\left(YY^{\dagger} \right)}
\end{split} \label{eqn:tm_bmf_back_class_f_func_and_part_func}
\\
\begin{split}
{}&= \sum_{I,J,K,L} f_{x,y}(A, A^{\dagger}; \vec{i}, \vec{j}, \vec{k},\vec{l}; \sigma_{1} ) \langle  X^{i_{1}}_{k_{1}} \dots  X^{i_{m}}_{k_{m}}  Y^{i_{m+1}}_{k_{m+1}}  \dots Y^{i_{m+n}}_{k_{m+n}} 
\\
{}& \qquad \qquad \qquad \qquad \times (X^{\dagger})^{l_{1}}_{j_{\sigma_{2}^{-1}(1)}} \dots   (X^{\dagger})^{l_{m}}_{j_{\sigma_{2}^{-1}(m)}} (Y^{\dagger})^{l_{m+1}}_{j_{\sigma_{2}^{-1}(m+1)}} \dots (Y^{\dagger})^{l_{m+n}}_{j_{\sigma_{2}^{-1}(m+n)}}  \rangle 
\end{split} 
\\
\begin{split}
{}&= \sum_{I,J,K,L}  f_{x,y}(A, A^{\dagger}; \vec{i}, \vec{j}, \vec{k},\vec{l}; \sigma_{1} ) \sum_{\gamma \in S_{m} \times S_{n}} \delta^{i_{1}}_{j_{\gamma \sigma_{2}^{-1}(1)}} \delta^{l_{1}}_{k_{\gamma^{-1}(1)}} \dots \delta^{i_{m+n}}_{j_{\gamma \sigma_{2}^{-1}(m+n)}} \delta^{l_{m+n}}_{k_{\gamma^{-1}(m+n)}}
\end{split}
\\
\begin{split}
{}&= \sum_{J,K}  \sum_{\gamma \in S_{m} \times S_{n}}  (A_{x}^{\dagger})^{j_{1}}_{k_{\gamma^{-1}(1)}} (A_{x})^{k_{1}}_{j_{\sigma_{1} \gamma \sigma_{2}^{-1}(1)}}  \dots (A_{x}^{\dagger})^{j_{m}}_{k_{\gamma^{-1}(m)}}(A_{x})^{k_{m}}_{j_{\sigma_{1} \gamma \sigma_{2}^{-1}(m)}} 
\\
{}& \qquad \qquad \qquad \qquad \times (A_{y}^{\dagger})^{j_{m+1}}_{k_{\gamma^{-1}(m+1)}}  (A_{y})^{k_{m+1}}_{j_{\sigma_{1} \gamma \sigma_{2}^{-1}(m+1)}} \dots (A_{y}^{\dagger})^{j_{m+n}}_{k_{\gamma^{-1}(m+n)}} (A_{y})^{k_{m+n}}_{j_{\sigma_{1} \gamma \sigma_{2}^{-1}(m+n)}} 
\end{split} 
\\
\begin{split}
{}&= \sum_{j_{1}, \dots, j_{n}} \sum_{\gamma \in S_{m} \times S_{n}} \left(A^{\dagger}_{x} A_{x}\right)_{j_{\gamma^{-1} \sigma_{1} \gamma \sigma_{2}^{-1}(1)}}^{j_{1}} \dots \left(A^{\dagger}_{x} A_{x}\right)_{j_{\gamma^{-1} \sigma_{1} \gamma \sigma_{2}^{-1}(m)}}^{j_{m}} 
\\
{}& \qquad  \qquad  \qquad  \qquad \qquad  \qquad  \times \left(A^{\dagger}_{y} A_{y} \right)_{j_{\gamma^{-1} \sigma_{1} \gamma \sigma_{2}^{-1}(m+1)}}^{j_{m+1}} \dots \left(A^{\dagger}_{y} A_{y}\right)_{j_{\gamma^{-1} \sigma_{1} \gamma \sigma_{2}^{-1}(m+n)}}^{j_{m+n}} 
\end{split}  \label{eqn:tm_bmf_back_class_aadag}
\\
\begin{split}
{}&= \sum_{j_{1}, \dots, j_{n}} \sum_{\gamma \in S_{m} \times S_{n}} \left(B_{x}\right)_{j_{\gamma^{-1} \sigma_{1} \gamma \sigma_{2}^{-1}(1)}}^{j_{1}} \dots \left(B_{x}\right)_{j_{\gamma^{-1} \sigma_{1} \gamma \sigma_{2}^{-1}(m)}}^{j_{m}} 
\\
{}& \qquad  \qquad  \qquad  \qquad \qquad \qquad \qquad \times \left(B_{y}\right)_{j_{\gamma^{-1} \sigma_{1} \gamma \sigma_{2}^{-1}(m+1)}}^{j_{m+1}} \dots \left(B_{y}\right)_{j_{\gamma^{-1} \sigma_{1} \gamma \sigma_{2}^{-1}(m+n)}}^{j_{m+n}}  \label{eqn:tm_bmf_back_class_aadag_to_b}
\end{split} 
\\
{}&= \sum_{\gamma \in S_{m} \times S_{n}} \sum_{\sigma_{3} \in S_{m+n}}\mathcal{O}_{\sigma_{3}}(B_{x}, B_{y}) \delta(\sigma_{3}^{-1} \gamma^{-1} \sigma_{1} \gamma \sigma_{2}^{-1})
\label{eqn:tm_bmf_back_class_trace_result} 
\\
{}&= \sum_{\gamma \in S_{m} \times S_{n}} \sum_{p_{3} } \left( \sum_{\alpha \in \text{Orb}_{\mathcal{A}}(p_{3})} \mathcal{O}_{\alpha}(B_{x}, B_{y})\delta( \alpha^{-1} \gamma^{-1} \sigma_{1} \gamma \sigma_{2}^{-1}) \right) \label{eqn:tm_bmf_trace_and_delta2_pt2} 
\\
{}&=  \sum_{p_{3} }  \mathcal{O}_{\sigma^{\left(p_{3}\right)}}(B_{x}, B_{y})\sum_{\gamma \in S_{m} \times S_{n}} \delta\left( \sum_{\alpha \in \text{Orb}_{\mathcal{A}}(p_{3})} \alpha^{-1} \gamma^{-1} \sigma_{1} \gamma \sigma_{2}^{-1} \right)\label{eqn:tm_bmf_trace_and_delta3_pt2}  
\\
{}&=  \sum_{p_{3}}  \mathcal{O}_{\sigma^{\left(p_{3}\right)}}(B_{x}, B_{y})\sum_{\gamma \in S_{m} \times S_{n}} \delta\left( T^{\mathcal{A}}_{p'_{3}} \gamma^{-1} \sigma_{1} \gamma \sigma_{2}^{-1} \right) \label{eqn:tm_bmf_trace_and_delta4_pt2}
\\ 
{}&= \sum_{p_{3} } \frac{m! n!|T^{\mathcal{A}}_{p'_{3}}|}{|T^{\mathcal{A}}_{p_{1}}| |T^{\mathcal{A}}_{p'_{2}}|} C^{p'_{3};\mathcal{A}}_{p_{1} p'_{2}} \mathcal{O}_{\sigma^{\left(p_{3}\right)}}(B_{x}, B_{y})\,. \label{eqn:tm_struct_const_and_op_classical}
\end{align}
Initial steps from \eqref{eqn:tm_bmf_back_class_f_func_and_part_func} to \eqref{eqn:tm_bmf_back_class_aadag} use Wick's theorem, Kronecker equivariance and contraction of indices. In \eqref{eqn:tm_bmf_back_class_aadag_to_b} we set the classical matrices as $A^{\dagger}_{x} A_{x} = B_{x}$ and $A^{\dagger}_{y} A_{y} = B_{y}$ to match the notation of equation \eqref{eqn:tm_bmf_same_indices_single_line_perms}, while a delta function and the GIO definition from \eqref{eqn:atmb_gio_def} were used in \eqref{eqn:tm_bmf_back_class_trace_result}. Equations \eqref{eqn:tm_bmf_trace_and_delta2_pt2}-\eqref{eqn:tm_bmf_trace_and_delta4_pt2} split the $\sigma_{3}$ sum into a sum over equivalence classes/partitions, labelled by $p_{3}$, and set a PCA element in the delta function. Note that $\alpha$ and $\sigma^{(p_{3})}$ belong in the same orbit/equivalence class, $\text{Orb}_{\mathcal{A}}(p_{3})$, since $\sigma^{\left(p_{3}\right)}$ is any permutation from the orbit. Lemma \ref{lem:two-matrix_lemma} was used to reach the final line \eqref{eqn:tm_struct_const_and_op_classical}. Using $|T_{p_{i}}^{\mathcal{A}}|\mathcal{O}_{\sigma^{\left( p_{i}\right)}} \equiv |T_{p_{i}}^{\mathcal{A}}|\mathcal{O}_{\sigma_{i}} = \mathcal{O}_{p_{i}}$ and rearranging orbit size factors, we achieve the same combinatorial basis result as equation \eqref{eqn:tm_bmf_simplified_final_corr_boxed} for the coupling matrix derivation
\begin{equation}
\boxed{
\langle \mathcal{O}_{p_{1}}(XA_{x}, YA_{y}) (\mathcal{O}_{p_{2}}(XA_{x}, YA_{y}))^{\dagger}  \rangle = m!n!\sum_{p_{3} } C^{p'_{3};\mathcal{A}}_{p_{1} p'_{2}} \mathcal{O}_{p_{3}}(B_{x}, B_{y})
}
\end{equation}
This two-matrix correlator using classical witness fields has an associated diagram given in Figure \ref{fig:two_mat_corr_bmf_wihtin_obs_png}. The swap in indices occurs now for both $A_{x}$ and $A_{y}$ matrices when the correlator is expanded in terms of a sum over permutations.

\begin{figure}[htb!]
\begin{center}
\includegraphics[width=13cm, height=15cm]{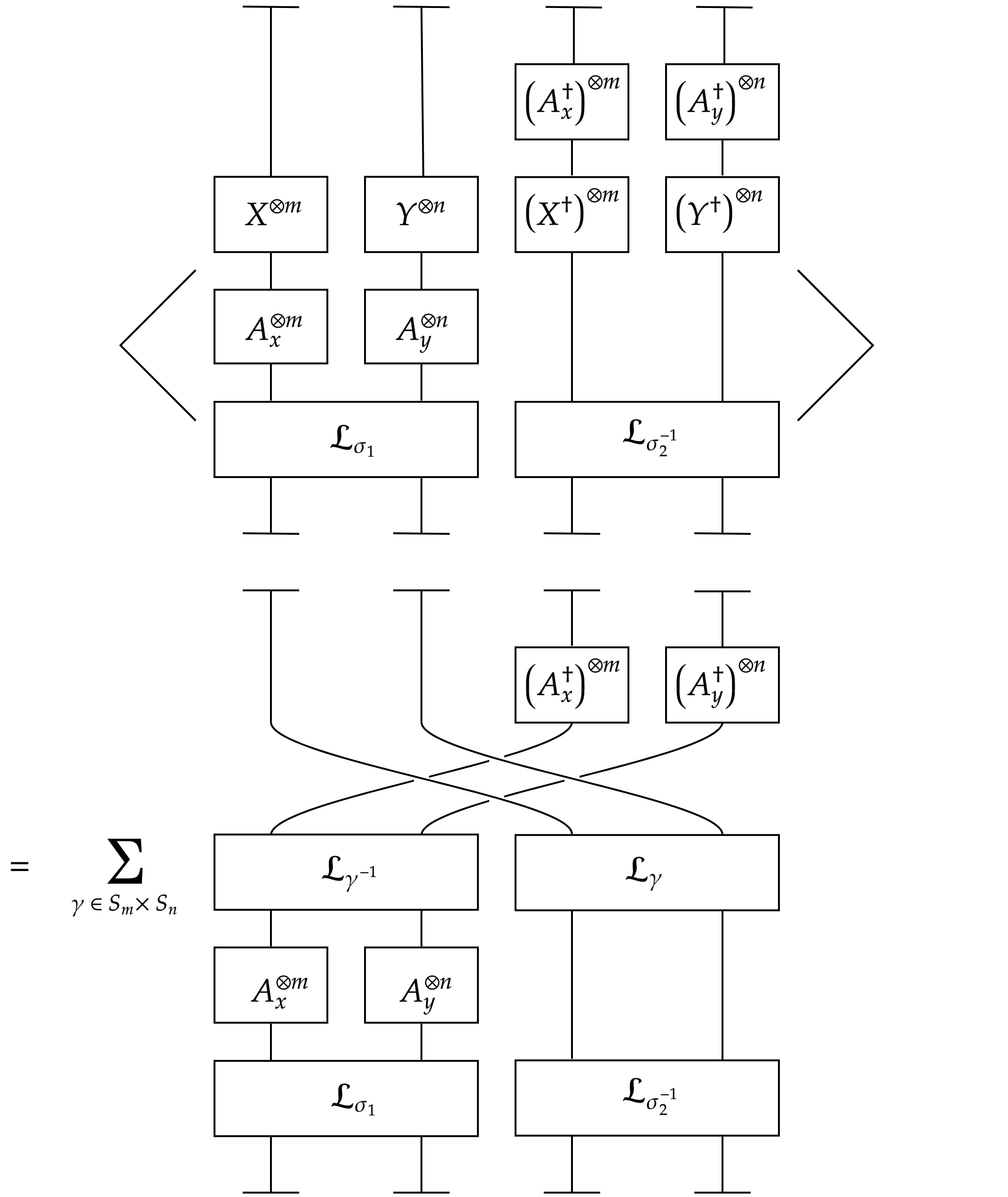}
\end{center}
\caption{The diagrammatic representation of the general two-matrix correlator with classical matrix fields. The presence of the two index swaps in the lower diagram is to account for the two classical matrices $A_{x}$ and $A_{y}$, appearing after applying Wick's theorem.}
\label{fig:two_mat_corr_bmf_wihtin_obs_png}
\end{figure}

\subsection{Fourier/\texorpdfstring{$Q$}{Q} basis for two-matrix, two-point function with classical fields}
\label{ss:schur_basis_back_fields_in_observables_two_mat}

To adopt the Fourier basis for the classical witness field result, define the gauge invariant Fourier operator
\begin{equation}
\label{eqn:spb_class_back_two_mat_schur_gen_obs_}
\mathcal{O}^{R}_{R_{1},R_{2},\mu, \nu}(XA_{x},YA_{y}) = \sum_{\sigma \in S_{m+n}} \delta\left(Q^{R}_{R_{1},R_{2},\mu,\nu} \sigma^{-1} \right)\mathcal{O}_{\sigma}(XA_{x},YA_{y})
\end{equation}
with Hermitian conjugate 
\begin{align}
\label{eqn:spb_class_back_two_mat_schur_gen_conj}
\left(\mathcal{O}^{R}_{R_{1},R_{2},\mu, \nu}(XA_{x},YA_{y}) \right)^{\dagger}  &= \sum_{\sigma\in S_{m+n}} \delta\left(Q^{R}_{R_{1},R_{2},\nu,\mu} \sigma \right)\mathcal{O}_{\sigma^{-1}}(A_{x}^{\dagger}X^{\dagger},A_{y}^{\dagger}Y^{\dagger})
\\
&= \sum_{\tilde{\sigma}\in S_{m+n}} \delta\left(Q^{R}_{R_{1},R_{2},\nu,\mu} \tilde{\sigma}^{-1} \right)\mathcal{O}_{\tilde{\sigma}}(A_{x}^{\dagger}X^{\dagger},A_{y}^{\dagger}Y^{\dagger})
\\
&= \mathcal{O}^{R}_{R_{1},R_{2},\nu, \mu}(A_{x}^{\dagger}X^{\dagger},A_{y}^{\dagger}Y^{\dagger})
\,.
\end{align}
where again, we used $(Q^{R}_{R_{1},R_{2},\mu,\nu})^{\dagger} = Q^{R}_{R_{1},R_{2},\nu,\mu}$, $\sigma^{\dagger} = \sigma^{-1}$ and $ \mathcal{O}_{\sigma}(XA_{x}, YA_{y})$ is the permutation parameterised operator, defined previously in \eqref{eqntm_:bmf_in_ob_obs_definition}.
Using these Fourier basis operators, the correlator is therefore
\begin{align}
\label{eqn:spb_class_back_corr1}
\langle {}& \mathcal{O}^{R}_{R_{1},R_{2},\mu_{1}, \nu_{1}}(XA_{x},YA_{y}) \left(\mathcal{O}^{R}_{R_{1},R_{2},\mu_{2}, \nu_{2}}(XA_{x},YA_{y}) \right)^{\dagger} \rangle \nonumber 
\\
{}&= \sum_{\substack{ \sigma, \tau \\ \in S_{m+n}}} \delta(Q^{R}_{R_{1}, R_{2}, \mu_{1}, \nu_{1}}\sigma^{-1}) \delta(Q^{S}_{S_{1}, S_{2}, \nu_{2},\mu_{2}}\tau) \langle \mathcal{O}_{\sigma}(XA_{x}, YA_{y}) \left(\mathcal{O}_{\tau}(XA_{x}, YA_{y})\right)^{\dagger} \rangle
\end{align}
A convenient form of the permutation parameterised correlator seen on the right hand side of the above equation is
\begin{align}
\begin{split}
\langle \mathcal{O}_{\sigma}(XA_{x}, YA_{y}) & \left(\mathcal{O}_{\tau}(XA_{x}, YA_{y})\right)^{\dagger} \rangle 
\\
&=\sum_{\gamma \in S_{m} \times S_{n}} \sum_{\rho \in S_{m+n}}\mathcal{O}_{\rho}(B_{x},B_{y}) \delta(\rho^{-1} \gamma^{-1} \sigma \gamma \tau^{-1}) \,.
\end{split}
\end{align}
This was achieved in \eqref{eqn:tm_bmf_back_class_trace_result} and sets $B_{x} = A_{x}^{\dagger} A_{x}$ and $B_{y} = A_{y}^{\dagger} A_{y}$. Plugging this into \eqref{eqn:spb_class_back_corr1} and then utilising lemma \ref{lem:two-matrix_lemma_fourier}, the final correlator expression is once again obtained
\begin{equation}
	\boxed{
	\begin{aligned}
\left\langle  \mathcal{O}^{R}_{R_{1}, R_{2}, \mu_{1}, \nu_{1}}(XA_{x}, YA_{y}) \right. & \left. (\mathcal{O}^{S}_{S_{1}, S_{2}, \mu_{2}, \nu_{2}}(XA_{x}, YA_{y}))^{\dagger} \right\rangle =
\\
& m!n! \delta^{RS} \delta_{R_{1}S_{1}} \delta_{R_{2}S_{2}} 
 \delta_{\nu_{1}\nu_{2}}  \mathcal{O}^{R}_{R_{1}, R_{2}, \mu_{1}, \mu_{2} }(B_{x}, B_{y})
	\end{aligned}
	}
\end{equation}
The result of this classical field witness case has the exact same form as the coupling witness field case \eqref{eqn:two_matrix_schur_basis_corr2}.

\newpage

\clearpage
\section{Algebras and multi-matrix correlators with multi-matrix witnesses }
\label{s:multi_matrix_model}

The extension to an arbitrary number of witness matrix fields takes gauge invariants of the form
\begin{align}
\begin{split}
\label{eqn:mmm_observable}
\mathcal{O}_{\sigma} (X_{l})   & = \text{Tr}_{V^{\otimes \mathbf{m}}_{N} }\left( X_{1}^{\otimes m_{1}} \otimes  X_{2}^{\otimes m_{2}} \cdots  X_{l-1}^{\otimes m_{l-1}} \otimes X_{l}^{\otimes m_{l}}  \cL_{\sigma} \right) \\
& = \text{Tr}_{V^{\otimes \mathbf{m}}_{N}}\left( \bigotimes_{\alpha = 1}^{l} X_{\alpha}^{\otimes m_{\alpha}}  \cL_{\sigma}\right) \,.
\end{split}
\end{align}
Here $l$ labels the unique/distinguishable matrices, ${\bf{m}} = m_{1} + m_{2} + \dots + m_{l}$ and $\cL_{\sigma}$ has action 
\begin{equation}
\label{eqn:mmm_lin_op_action}
\mathcal{L}_{\sigma} \ket{e_{i_{1}} \otimes \dots \otimes e_{i_{\mathbf{m}}}} = \ket{e_{i_{\sigma(1)}} \otimes \dots \otimes e_{i_{\sigma(\mathbf{m})}}} \,.
\end{equation}
The invariance of these operators is captured by equivalence relation 
\begin{equation}
\label{eqn:mmm_equiv_rel}
\sigma \sim \gamma \sigma \gamma^{-1}
\end{equation}
where $\sigma \in S_{\sum_{i=1}^{l} m_{i}} \equiv S_{{\bf{m}}}$ and $\gamma \in S_{m_{1}} \times \cdots \times S_{m_{l}}$, such that
\begin{equation}
\label{eqn:mmm_obs_invariance_multi_matrix}
\mathcal{O}_{\sigma}(X_{l}) = \mathcal{O}_{\gamma \sigma \gamma^{-1}}(X_{l}) \,.
\end{equation}
The associated PCA stemming from such equivalence is an extension to $\mathcal{A}(m,n)$ (the Necklace PCA) and hence it is denoted as: $\mathcal{A}(m_{1}, m_{2}, \dots, m_{l}) := \mathcal{A}(\vec{m})$. This $\mathcal{A}(\vec{m})$ is defined as the sub-algebra of $\mathbb{C}[S_{\mathbf{m}}]$ that commutes with $\mathbb{C}[S_{m_{1}} \times \dots \times S_{m_{l}}]$ and to express its elements, we introduce the following notation. $\text{Aut}_{\mathcal{\vec{A}}}(p) \equiv \text{Aut}_{S_{m_{1}} \times \dots \times  S_{m_{l}}}\left(\sigma^{\left(p\right)}\right)$ is the subgroup of $S_{m_{1}} \times \dots \times S_{m_{l}}$ comprised of elements that leave $\sigma^{\left(p\right)} \in S_{\mathbf{m}}$ invariant under conjugation by $\gamma \in S_{m_{1}} \times \dots \times S_{m_{l}}$, and has size $|\text{Aut}_{\mathcal{\vec{A}}}(p)|$. The orbit, $\text{Orb}_{\mathcal{\vec{A}}}(p) \equiv \text{Orbit}\left(\sigma^{\left(p\right)}, S_{m_{1}} \times \dots \times S_{m_{l}}\right)$, is the set of elements from $S_{\mathbf{m}}$ obtained by acting on $\sigma^{\left(p\right)} \in S_{\mathbf{m}}$ with $\gamma \in S_{m_{1}} \times \cdots \times S_{m_{l}}$ by conjugation, i.e. $\gamma \sigma^{\left(p\right)} \gamma^{-1}$, and the size of each orbit/equivalence class is denoted by $|T_{p}^{\mathcal{\vec{A}}}| = |\text{Orb}_{\mathcal{\vec{A}}}(p)| = \prod_{\beta=1}^{l} m_{\beta}!/|\text{Aut}_{\mathcal{\vec{A}}}(p)|$. For clarity, $\sigma^{\left(p\right)} \in S_{\mathbf{m}}$ is any permutation within the orbit $\text{Orb}_{\mathcal{\vec{A}}}(p)$. As such, an element of $\mathcal{A}(\vec{m})$ is defined as
\begin{equation}
\label{eqn:mmm_algebra_element_mult_mat}
T^{\mathcal{\vec{A}}}_{p} = \frac{1}{|\text{Aut}_{\mathcal{\vec{A}}} (p)|} \sum_{\gamma \in S_{m_{1}} \times \dots \times S_{m_{l}} } \gamma \sigma^{\left(p\right)} \gamma^{-1} = \sum_{\alpha \in \text{Orb}_{\mathcal{\vec{A}}}(p)} \alpha \quad \in \mathcal{A}(\vec{m})\,,
\end{equation}
where the label $\mathcal{\vec{A}}$ on $T^{\mathcal{\vec{A}}}_{p}$, $\text{Aut}_{\mathcal{\vec{A}}} (p)$ and $\text{Orb}_{\mathcal{\vec{A}}}(p)$ implies that these quantities are associated to the $\mathcal{A}(\vec{m})$ PCA. Such PCA elements are therefore sums over elements within the same orbit/class, obtained from the equivalence relation of \eqref{eqn:mmm_equiv_rel}.
 
Applying similar techniques as previously encountered, the following sections explore the combinatorial and Schur/Fourier basis correlator results for multi-matrix models with classical and coupling witness fields.


\subsection{Two point function of general operators with multi-matrix-couplings}
\label{ss:gen_gauge_obs_multi_mat_with_bmf}

The partition function with coupling witness matrices for the multi-matrix model is
\begin{equation}
\label{eqn:mmm_bmf_path_integral_two_mat}
\Sigma[0] =  \prod_{\eta=1}^{l}\int [dX_{\eta}] e^{-\text{Tr}(X_{\eta} A_{\eta} X_{\eta}^{\dagger})} \,,
\end{equation}
where in what follows, $(A_{\eta})^{-1} = B_{\eta}$  defines the coupling matrix associated to variable $X_{\eta}$ and the $[0]$ in $\Sigma[0]$ implies a sourceless partition function. By introducing sources in the same manner as Appendix \ref{app:background_matrix_field_correlator}, the two-point functions for the $X_{\eta}$ field variable can be derived 
\begin{equation}
\label{eqn:mmm_correlators_with_bmf}
\langle (X_{\alpha})^{i}_{j} (X^{\dagger}_{\beta})^{k}_{l}  \rangle  =  \prod_{\eta=1}^{l} \frac{1}{\Sigma[0]}\int [dX_{\eta}] (X_{\alpha})^{i}_{j} (X^{\dagger}_{\beta})^{k}_{l} e^{-\text{Tr}(X_{\eta} A_{\eta} X_{\eta}^{\dagger})} = \delta_{\alpha \beta} \delta^{i}_{l} ((A_{\alpha})^{-1})^{k}_{j} = \delta_{\alpha \beta} \delta^{i}_{l} (B_{\alpha})^{k}_{j} \,.
\end{equation}
Using the permutation parameterised GIO
\begin{equation}
\label{eqn:mmm_observable_2}
\mathcal{O}_{\sigma} (X_{l}) = \text{Tr}_{V^{\otimes \mathbf{m}}_{N}}\left( X_{1}^{\otimes m_{1}} \otimes  X_{2}^{\otimes m_{2}} \otimes \cdots \otimes  X_{l-1}^{\otimes m_{l-1}} \otimes X_{l}^{\otimes m_{l}}  \cL_{\sigma} \right) 
\end{equation}
and its Hermitian conjugate
\begin{align}
\begin{split}
\left(\mathcal{O}_{\sigma}(X_{l})\right)^{\dagger} 
&= \text{Tr}_{V^{\otimes \mathbf{m}}_{N}}\left( (X_{1}^{\dagger})^{\otimes m_{1}} \otimes (X_{2}^{\dagger})^{\otimes m_{2}} \otimes \cdots \otimes (X_{l-1}^{\dagger})^{\otimes m_{l-1}} \otimes (X_{l}^{\dagger})^{\otimes m_{l}}  \cL_{\sigma^{-1}} \right) 
\\
&= \mathcal{O}_{\sigma^{-1}}(X^{\dagger}_{l}) \,,
\end{split}
\end{align}
the correlator is then constructed as follows
\begin{align}
\langle {}&\mathcal{O}_{\sigma_{1}} (X_{l}) \left(\mathcal{O}_{\sigma_{2}}(X_{l}) \right)^{\dagger} \rangle \nonumber
\\
{}&=  \left\langle \text{Tr}_{V^{\otimes \mathbf{m}}_{N}}\left(  X_{1}^{\otimes m_{1}} \otimes \cdots \otimes X_{l}^{\otimes m_{l}} \cL_{\sigma_{1}} \right)  \text{Tr}_{V^{\otimes \mathbf{m}}_{N}}\left( (X_{1}^{\dagger})^{\otimes m_{1}} \otimes \cdots \otimes (X_{l}^{\dagger})^{\otimes m_{l}}  \cL_{\sigma_{2}^{-1}} \right)\right\rangle \label{eqn:mmm_gio_with_bmf1}
\\
\begin{split}
{}&= \sum_{\substack{i_{1}, \dots, i_{{\mathbf{m}}} \\ j_{1}, \dots, j_{\mathbf{m}} }} \left\langle \right. (X_{1})^{i_{1}}_{i_{\sigma_{1}(1)}} \cdots (X_{1})^{i_{m_{1}}}_{i_{\sigma_{1}(m_{1})}} \cdots \cdots (X_{l})^{i_{\mathbf{m} - m_{l}}}_{i_{\sigma_{1}(\mathbf{m} - m_{l})}} \cdots (X_{l})^{i_{\mathbf{m}}}_{i_{\sigma_{1}(\mathbf{m})}}
 \\ 
{}& \qquad \qquad \times(X_{1}^{\dagger})^{j_{1}}_{j_{\sigma_{2}^{-1}(1)}} \cdots (X_{1}^{\dagger})^{j_{m_{1}}}_{j_{\sigma_{2}^{-1}(m_{1})}} 
	\cdots \cdots (X_{l}^{\dagger})^{j_{\mathbf{m} - m_{l}}}_{j_{\sigma_{2}^{-1}(\mathbf{m} - m_{l})}} \cdots (X_{l}^{\dagger})^{j_{\mathbf{m}}}_{j_{\sigma_{2}^{-1}(\mathbf{m})}} \left. \right\rangle 
\end{split}
\\
\begin{split}
&{}= \sum_{\gamma \in S_{m_{1}} \times \dots \times S_{m_{l}}}\sum_{\substack{i_{1}, \dots, i_{\mathbf{m}} \\ j_{1}, \dots, j_{\mathbf{m}} }} \delta^{i_{1}}_{j_{\gamma \sigma_{2}^{-1}(1)}} \cdots \delta^{i_{m_{1}}}_{j_{\gamma \sigma_{2}^{-1}(m_{1})}} 
\cdots \cdots \delta^{i_{\mathbf{m} - m_{l}}}_{j_{ \gamma \sigma_{2}^{-1}(\mathbf{m} - m_{l})}} \cdots \delta^{i_{\mathbf{m}}}_{j_{\gamma\sigma_{2}^{-1}(\mathbf{m})}} 
\\
&{} \qquad \qquad \times (B_{1})^{j_{1}}_{i_{\gamma^{-1}\sigma_{1}(1)}} \cdots (B_{1})^{j_{m_{1}}}_{i_{\gamma^{-1}\sigma_{1}(m_{1})}} 
	\cdots \cdots (B_{l})^{j_{\mathbf{m} - m_{l}}}_{i_{\gamma^{-1}\sigma_{1}(\mathbf{m} - m_{l})}} \cdots (B_{l})^{j_{\mathbf{m}}}_{i_{\gamma^{-1}\sigma_{1}(\mathbf{m})}}  \label{eqn:mmm_gio_with_bmf2}
\end{split}
\\
\begin{split}
{}&= \sum_{\gamma \in S_{m_{1}} \times \dots \times S_{m_{l}}}\sum_{\substack{i_{1}, \dots, i_{\mathbf{m}} \\ j_{1}, \dots, j_{\mathbf{m}} }} \delta^{i_{\gamma^{-1} \sigma_{1}(1)}}_{j_{\gamma^{-1} \sigma_{1} \gamma \sigma_{2}^{-1}(1)}} \cdots \delta^{i_{\gamma^{-1} \sigma_{1}(m_{1})}}_{j_{\gamma^{-1} \sigma_{1}\gamma \sigma_{2}^{-1}(m_{1})}} \cdots \cdots 
\delta^{i_{\gamma^{-1} \sigma_{1}(\mathbf{m} - m_{l})}}_{j_{\gamma^{-1} \sigma_{1} \gamma \sigma_{2}^{-1}(\mathbf{m} - m_{l})}} \cdots \delta^{i_{\gamma^{-1} \sigma_{1}(\mathbf{m})}}_{j_{\gamma^{-1} \sigma_{1}\gamma\sigma_{2}^{-1}(\mathbf{m})}} 
\\
{}& \qquad \qquad \times(B_{1})^{j_{1}}_{i_{\gamma^{-1}\sigma_{1}(1)}} \cdots (B_{1})^{j_{m_{1}}}_{i_{\gamma^{-1}\sigma_{1}(m_{1})}} 
	\cdots \cdots (B_{l})^{j_{\mathbf{m} - m_{l}}}_{i_{\gamma^{-1}\sigma_{1}(\mathbf{m} - m_{l})}} \cdots (B_{l})^{j_{\mathbf{m}}}_{i_{\gamma^{-1}\sigma_{1}(\mathbf{m})}} \label{eqn:mmm_gio_with_bmf2_pt5}
\end{split}
\\
\begin{split}
{}&= \sum_{\gamma \in S_{m_{1}} \times \dots \times S_{m_{l}}} \sum_{j_{1}, \dots j_{\mathbf{m}}} (B_{1})^{j_{1}}_{j_{\gamma^{-1}\sigma_{1} \gamma \sigma_{2}^{-1}(1)}} \dots (B_{1})^{j_{m_{1}}}_{j_{\gamma^{-1}\sigma_{1} \gamma \sigma_{2}^{-1}(m_{1})}}  \cdots
\\
{}& \qquad \qquad \qquad \qquad \qquad \qquad \cdots \times  (B_{l})^{j_{(\mathbf{m}-m_{l})}}_{j_{\gamma^{-1}\sigma_{1} \gamma \sigma_{2}^{-1}(\mathbf{m}-m_{l})}} \cdots (B_{l})^{j_{(\mathbf{m})}}_{j_{\gamma^{-1}\sigma_{1} \gamma \sigma_{2}^{-1}(\mathbf{m})}} \,. \label{eqn:mmm_gio_with_bmf3}
\end{split} 
\end{align}
Between equations \eqref{eqn:mmm_gio_with_bmf2} and \eqref{eqn:mmm_gio_with_bmf3} Kronecker equivariance was applied and the set of $i$ indices were contracted. Further simplification leads to 
\begin{align}
\langle \mathcal{O}_{\sigma_{1}}(X_{l})& \left(\mathcal{O}_{\sigma_{2}}(X_{l}) \right)^{\dagger} \rangle  \nonumber
\\
&= \sum_{\substack{ \sigma_{3} \in S_{\mathbf{m}} \\ \gamma \in S_{m_{1}} \times \dots \times S_{m_{l}} }} \text{Tr}_{V_{N}^{\otimes \mathbf{m}}}(B_{1}^{\otimes m_{1}} \otimes \dots \otimes B_{l}^{\otimes m_{l}} \cL_{\sigma_{3}}) \delta(\sigma_{3}^{-1} \gamma^{-1} \sigma_{1} \gamma \sigma_{2}^{-1})\label{eqn:mmm_2_point_delta_func}
\\
&=  \sum_{\substack{ \sigma_{3} \in S_{\mathbf{m}} \\ \gamma \in S_{m_{1}} \times \dots \times S_{m_{l}} }} \mathcal{O}_{\sigma_{3}}(B_{l}) \delta(\sigma_{3}^{-1} \gamma^{-1} \sigma_{1} \gamma \sigma_{2}^{-1}) \label{eqn:mmm_2_point_result}
\\
&=  \sum_{p_{3} } \sum_{\gamma \in S_{m_{1}} \times \dots \times S_{m_{l}}}  \left(  \sum_{\alpha \in \text{Orb}_{\mathcal{\vec{A}}}(p_{3})}\mathcal{O}_{\alpha}(B_{l}) \delta(\alpha^{-1} \gamma^{-1} \sigma_{1} \gamma \sigma_{2}^{-1}) \right) \label{eqn:mmm_sum_over_equiv_classes}
\\
&=   \sum_{p_{3} } \mathcal{O}_{\sigma^{\left(p_{3}\right)}}(B_{l})\sum_{\gamma \in S_{m_{1}} \times \dots \times S_{m_{l}}} \delta(T^{\mathcal{\vec{A}}}_{p'_{3}} \gamma^{-1} \sigma_{1} \gamma \sigma_{2}^{-1}) \label{eqn:mmm_sum_over_equiv_classes2} \,,
\end{align}
where $p_{3}$ labels equivalence classes/orbits. In \eqref{eqn:mmm_2_point_delta_func} trace notation was adopted and a delta function was used. This is followed by setting the witness field GIO in equation \eqref{eqn:mmm_2_point_result}, using the definition of equation \eqref{eqn:mmm_observable}. Splitting the sum over $\sigma_{3}$ into a sum over equivalence classes/orbits labelled by $p_{3}$ is completed in \eqref{eqn:mmm_sum_over_equiv_classes}, and finally the sum over $\alpha$  -- the sum over elements in the equivalence class/orbit $\text{Orb}_{\mathcal{\vec{A}}}(p_{3})$ of which $\sigma^{\left(p_{3}\right)}$ and $\alpha$ belong -- is brought inside the delta function to produce the PCA element $T_{p'_{3}}^{\mathcal{\vec{A}}}$ in \eqref{eqn:mmm_sum_over_equiv_classes2}. The prime label indicates an orbit associated to an inverse permutation. 
\begin{lemma} 
\label{lem:multi-matrix_lemma}
For $\gamma \in S_{m_{1}} \times \dots \times  S_{m_{l}}$, $\sigma_{i} \in S_{\mathbf{m}}$ and $T^{\mathcal{\vec{A}}}_{p_{i}} \in \mathcal{A}(\vec{m})$, the following equality holds 
\begin{equation}
\sum_{\gamma \in S_{m_{1}} \times \dots \times S_{m_{l}}} \delta(T^{\mathcal{\vec{A}}}_{p'_{3}} \gamma^{-1} \sigma_{1} \gamma \sigma_{2}^{-1})
=\frac{|T^{\mathcal{\vec{A}}}_{p'_{3}}|\prod_{\alpha=1}^{l}m_{\alpha}!}{|T^{\mathcal{\vec{A}}}_{p_{1}}| |T^{\mathcal{\vec{A}}}_{p'_{2}}|} C^{p'_{3};\mathcal{\vec{A}}}_{p_{1} p'_{2}} 
\end{equation}
where $C^{p'_{3};\mathcal{\vec{A}}}_{p_{1} p'_{2}}$ is the PCA structure constant of $\mathcal{A}(\vec{m})$ and $\sigma_{1}$, $\sigma_{2}^{-1}$ belong to equivalence classes labelled by $p_{1}$, $p'_{2}$ respectively.
\end{lemma} 
\begin{proof}
\begin{align}
&\sum_{\gamma \in S_{m_{1}} \times \dots \times S_{m_{l}}}  \delta(T^{\mathcal{\vec{A}}}_{p'_{3}} \gamma^{-1} \sigma_{1} \gamma \sigma_{2}^{-1}) \nonumber
\\
&= \sum_{\mu_{1} \in S_{m_{1}} \times \dots \times S_{m_{l}}} \delta\left( T^{\mathcal{\vec{A}}}_{p'_{3}} (\mu_{1}\mu_{2})^{-1} \sigma_{1} (\mu_{1}\mu_{2}) \sigma_{2}^{-1} \right) \label{eqn:mmm_pcalem_gamma_replaced}
\\
&= \sum_{\mu_{1} \in S_{m_{1}} \times \dots \times S_{m_{l}}} \delta\left( (\mu_{2}^{-1} \mu_{2}) T^{\mathcal{\vec{A}}}_{p'_{3}}\mu_{2}^{-1} \mu_{1}^{-1} \sigma_{1} \mu_{1} \mu_{2} \sigma_{2}^{-1} \right) \label{eqn:mmm_pcalem_identity_insert}
\\
&= \sum_{\mu_{1} \in S_{m_{1}} \times \dots \times S_{m_{l}}} \delta\left( \underbrace{\mu_{2} T^{\mathcal{\vec{A}}}_{p'_{3}}\mu_{2}^{-1}}_{T^{\mathcal{\vec{A}}}_{p'_{3}}} \mu_{1}^{-1} \sigma_{1} \mu_{1} \mu_{2} \sigma_{2}^{-1} \mu_{2}^{-1}\right) \label{eqn:mmm_pcalem_simplify_pca_conj_and_cycle}
\\
&= \frac{1}{\prod_{\alpha = 1}^{l} m_{\alpha}!}\sum_{\mu_{1},\mu_{2} \in S_{m_{1}} \times \dots \times S_{m_{l}}} \delta\left(  T^{\mathcal{\vec{A}}}_{p'_{3}} (\mu_{1}^{-1} \sigma_{1} \mu_{1}) (\mu_{2} \sigma_{2}^{-1} \mu_{2}^{-1}) \right)  \label{eqn:mmm_pcalem_mu2_sum_added}
\\
&= \frac{1}{\prod_{\alpha = 1}^{l} m_{\alpha}!}\delta\left(  T^{\mathcal{\vec{A}}}_{p'_{3}} \left(\sum_{\mu_{1} \in S_{m_{1}} \times \dots \times S_{m_{l}}} \mu_{1}^{-1} \sigma_{1} \mu_{1} \right) \left(\sum_{\mu_{2} \in S_{m_{1}} \times \dots \times S_{m_{l}}} \mu_{2} \sigma_{2}^{-1} \mu_{2}^{-1} \right) \right) \label{eqn:mmm_pcalem_sums_in_delta}
\\
&= \frac{|\text{Aut}_{\mathcal{\vec{A}}(p_{1})}||\text{Aut}_{\mathcal{\vec{A}}(p'_{2})}|}{\prod_{\alpha = 1}^{l} m_{\alpha}!} \delta\left(  T^{\mathcal{\vec{A}}}_{p'_{3}}  T^{\mathcal{\vec{A}}}_{p_{1}}  T^{\mathcal{\vec{A}}}_{p'_{2}}  \right)\label{eqn:mmm_pcalem_pca_eles_and_auto_factors}
\\
&= \frac{\prod_{\alpha = 1}^{l} m_{\alpha}!}{|T^{\mathcal{\vec{A}}}_{p_{1}} | |T^{\mathcal{\vec{A}}}_{p'_{2}}|} \delta\left(  T^{\mathcal{\vec{A}}}_{p'_{3}}  T^{\mathcal{\vec{A}}}_{p_{1}}  T^{\mathcal{\vec{A}}}_{p'_{2}}  \right) \label{eqn:mmm_pcalem_orb_stab_theorem}\,.
\end{align}
In equation \eqref{eqn:mmm_pcalem_gamma_replaced}, the sum over $\gamma$ is replaced by a sum over $\mu_{1}$, taking $\gamma \to \mu_{1}\mu_{2}$. An identity permutation ($e = \mu_{2}^{-1} \mu_{2}$) is inserted in \eqref{eqn:mmm_pcalem_identity_insert}, followed by cycling $\mu_{2}^{-1}$ in the delta function next to $\sigma_{2}^{-1}$ and identifying that $\mu_{2} T^{\mathcal{\vec{A}}}_{p'_{3}} \mu_{2}^{-1} = T^{\mathcal{\vec{A}}}_{p'_{3}}$ in \eqref{eqn:mmm_pcalem_simplify_pca_conj_and_cycle}. The subsequent steps introduce a sum over $\mu_{2}$ in \eqref{eqn:mmm_pcalem_mu2_sum_added} (including the required $1/|S_{m_{1}} \times \dots \times S_{m_{l}}| = 1/\prod_{\alpha=1}^{l}m_{\alpha}!$ factor), place the sums in the delta function to establish the PCA elements and automorphism group size factors between \eqref{eqn:mmm_pcalem_sums_in_delta} and \eqref{eqn:mmm_pcalem_pca_eles_and_auto_factors}, and finally, use the orbit-stabiliser theorem to rewrite $|\text{Aut}_{\mathcal{\vec{A}}}(p_{i})| = |S_{m_{1}} \times \dots \times S_{m_{l}}|/|T^{\mathcal{\vec{A}}}_{p_{i}}|$ in equation \eqref{eqn:mmm_pcalem_orb_stab_theorem}. Here, as in previous sections, the prime notation on the class labels imply the equivalence class of an inverse permutation, and we choose that $\sigma_{1}$ and $\sigma_{2}^{-1}$ belong to equivalence classes labelled by $p_{1}$ and $p'_{2}$ respectively.

The remaining delta function can be reduced, revealing the structure constant using the multiplication properties of the $\mathcal{A}(\vec{m})$ PCA elements
\begin{equation}
\label{eqn:mmm_delta_struct_cons_rel}
\delta\left(  T^{\mathcal{\vec{A}}}_{p'_{3}}  T^{\mathcal{\vec{A}}}_{p_{1}}  T^{\mathcal{\vec{A}}}_{p'_{2}}  \right) =  \sum_{p_{k} }C^{p_{k}; \mathcal{\vec{A}}}_{p_{1}p'_{2}} \delta( T^{\mathcal{\vec{A}}}_{p'_{3}} T^{\mathcal{\vec{A}}}_{p_{k}}) =  \sum_{p_{k} }C^{p_{k}; \mathcal{\vec{A}}}_{p_{1} p'_{2}} \delta_{p'_{3} p_{k}} |T^{\mathcal{\vec{A}}}_{p_{k}}| =  |T^{\mathcal{\vec{A}}}_{p'_{3}}|C^{p'_{3}; \mathcal{\vec{A}}}_{p_{1} p'_{2}} \,.
\end{equation}
Therefore
\begin{equation}
\sum_{\gamma \in S_{m_{1}} \times \dots \times S_{m_{l}}}  \delta(T^{\mathcal{\vec{A}}}_{p'_{3}} \gamma^{-1} \sigma_{1} \gamma \sigma_{2}^{-1})  =  \frac{\prod_{\alpha = 1}^{l} m_{\alpha}!}{|T^{\mathcal{\vec{A}}}_{p_{1}} | |T^{\mathcal{\vec{A}}}_{p'_{2}}|} |T^{\mathcal{\vec{A}}}_{p'_{3}}|C^{p'_{3}; \mathcal{\vec{A}}}_{p_{1} p'_{2}} \,.
\end{equation}
\end{proof}
Using this lemma, the correlator expression is
\begin{equation}
\langle \mathcal{O}_{\sigma_{1}}(X_{l}) \left(\mathcal{O}_{\sigma_{2}}(X_{l})\right)^{\dagger} \rangle
=  \sum_{p_{3} } \frac{|T^{\mathcal{\vec{A}}}_{p'_{3}}|\prod_{\alpha = 1}^{l} m_{\alpha}!}{|T^{\mathcal{\vec{A}}}_{p_{1}} | |T^{\mathcal{\vec{A}}}_{p'_{2}}|} C^{p'_{3}; \mathcal{\vec{A}}}_{p_{1} p'_{2}} \mathcal{O}_{\sigma^{\left(p_{3}\right)}}(B_{l}) \,.
\end{equation}
Finally, using $|T^{\mathcal{\vec{A}}}_{p'_{i}}| = |T^{\mathcal{\vec{A}}}_{p_{i}}|$, $\mathcal{O}_{\sigma_{1}} \equiv \mathcal{O}_{\sigma^{(p_{1})}} $, $\mathcal{O}_{\sigma_{2}^{-1}} \equiv \mathcal{O}_{\sigma^{(p'_{2})}} $, rearranging the orbit sizes and taking $|T^{\mathcal{\vec{A}}}_{p_{i}}|\mathcal{O}_{\sigma^{\left( p_{i}\right)}} = \mathcal{O}_{T^{\mathcal{\vec{A}}}_{p_{i}}} \equiv \mathcal{O}_{p_{i}}$ the correlator in combinatorial basis becomes
\begin{equation}
\label{eqn:mmm_final_correlator_boxed}
\boxed{
\langle \mathcal{O}_{p_{1}}(X_{l}) \left(\mathcal{O}_{p_{2}}(X_{l})\right)^{\dagger} \rangle = \prod_{\alpha=1}^{l} m_{\alpha}! \sum_{p_{3} } C^{p'_{3}; \mathcal{\vec{A}}}_{p_{1} p'_{2}} \mathcal{O}_{p_{3}}(B_{l})
}
\end{equation}

As in the previous sections, we obtain a linear combination of combinatorial basis GIOs in witness fields $B_{l}$, upon evaluating the quantum field $X_{l}$ GIO two-point function. Supplying the basis label data $(p_{1}, p'_{2}, p'_{3})$ will therefore provide the required information to reconstruct the structure constants of the permutation centraliser algebra $\mathcal{A}(\vec{m})$. Additionally, structure constant $C^{p'_{3}; \mathcal{\vec{A}}}_{p_{1} p'_{2}}$ is an integer by virtue that each orbit/equivalence class is of a set number of elements.


\subsection{Fourier/\texorpdfstring{$Q$}{Q} basis for multi-matrix correlator}
\label{ss:multi_matrix_schur_basis}	
The representation decomposition for the multi-matrix case follows much the same route as \S \ref{ss:two_matrix_schur_basis} but is generalised to a larger tensor space to account for the larger number of unique matrices\footnote{The notation in this section closely follows \cite{Pasukonis:2013ts} which provides further details/identities in their appendices A, B and D.}. Generalising for $\sigma \in S_{\mathbf{m}}$ and $\gamma \in S_{m_{1}} \times \dots \times S_{m_{l}}$ we have a decomposition of form
\begin{equation}
\label{eqn:mmm_spb_vector_decomp}
V_{R}^{(S_{\mathbf{m}})} = \bigoplus_{R_{1}, R_{2}, \dots ,R_{l}} V_{R_{1}}^{(S_{m_{1}})} \otimes V_{R_{2}}^{(S_{m_{2}})} \otimes \dots \otimes V_{R_{l}}^{(S_{m_{l}})} \otimes V_{R_{1}, R_{2}, \dots, R_{l}}^{R}.
\end{equation}
A state in such a space can be labelled by
\begin{equation}
\label{eqn:mmm_spb_state_labelling}
\ket{R; \mathbf{r}, \nu, \mathbf{t}}
\end{equation} 
where $R$ is the representation of $S_{\mathbf{m}}$ and $\mathbf{r}=(R_{1}, R_{2}, \dots, R_{l})$ which denotes the set of individual representations $R_{i}$, for $i = 1, \dots, l$. Multiplicity label $\nu$ indicates/counts how many times the tensor product of the $\mathbf{r}$ irreps appear in the decomposition, and $\mathbf{t} = (t_{1}, t_{2}, \dots, t_{l})$ are a set of labels used to denote the states of each individual representation space e.g. $V_{R_{1}}^{(S_{m_{1}})}$ has states labelled by $t_{1}$. Following on from this notation, the characters from which the GIO in Fourier basis are generated, are now established. 
Define the projection-like operator
\begin{equation}
\label{eqn:mmm_spb_projection_operator}
P^{R}_{\mathbf{r},\mu, \nu} =  \sum_{\mathbf{t}} \ket{R; \mathbf{r}, \mu, \mathbf{t}}\bra{R; \mathbf{r}, \nu, \mathbf{t}} 
\end{equation}
with components given in terms of branching coefficients\footnote{The branching coefficients $B^{R;i}_{\mathbf{r}, \nu; \mathbf{t}}$ are defined as the components of the vector $\ket{R; \mathbf{r}, \nu, \mathbf{t}}$ in any given orthogonal basis for $R$: $B^{R;i}_{\mathbf{r}, \nu; \mathbf{t}} = \braket{R;i|R; \mathbf{r}, \nu, \mathbf{t}}$.}
\begin{equation}
\left(P^{R}_{\mathbf{r},\mu, \nu}\right)^{ij} = \sum_{\mathbf{t}} B^{R;i}_{\mathbf{r}, \mu; \mathbf{t}} B^{R;j}_{\mathbf{r}, \nu; \mathbf{t}} \,,
\end{equation}
where  $\sum_{\mathbf{t}}$ implies a sum over all state spaces in $\mathbf{t} = (t_{1}, t_{2}, \dots, t_{l})$. The character is then constructed from these operators  
\begin{equation}
\chi^{R}_{\mathbf{r}, \mu, \nu}(\sigma) = \text{Tr}_{V^{\mathbf{r}}}\left[P^{R}_{\mathbf{r},\mu, \nu}\left(D^{R}(\sigma) \right) \right]
\end{equation}
which leads naturally to the $Q$-basis element definition
\begin{align}
\label{eqn:q_basis_mmm_definition}
Q^{R}_{\mathbf{r},\mu,\nu} = \frac{d_{R}}{\mathbf{m}!} \sum_{\sigma \in S_{\mathbf{m}}} \chi^{R}_{\mathbf{r},\mu,\nu} (\sigma) \sigma^{-1} \,.
\end{align}
where $d_{R}$ is the dimension of representation $R$ of $S_{\mathbf{m}}$. An important property of these basis elements is that they multiply like matrices 
\begin{equation}
\label{eqn:q_basis_mmm_mult_property}
Q^{R}_{\mathbf{r},\mu_{1},\nu_{1}} Q^{S}_{\mathbf{s},\mu_{2},\nu_{2}}  = \delta^{RS}\delta_{\mathbf{r}\mathbf{s}}  \delta_{\nu_{1}\mu_{2}}  Q^{R}_{\mathbf{r},\mu_{1},\nu_{2}} \,.
\end{equation}
See appendix \ref{app:redone_wa_q_basis}, for the derivation of this rule. The GIO in this Fourier basis for the multi-matrix model is hence defined as 
\begin{equation}
\label{eqn:mmm_spb_schur_gen_obs}
\mathcal{O}^{R}_{\mathbf{r}, \mu, \nu}(X_{l}) = \sum_{\sigma \in S_{\mathbf{m}}} \delta\left(Q^{R}_{\mathbf{r}, \mu, \nu} \sigma^{-1} \right) \mathcal{O}_{\sigma}(X_{l}) \,,
\end{equation}
with conjugate 
\begin{align}
\label{eqn:mmm_spb_schur_gen_obs_conj}
(\mathcal{O}^{R}_{\mathbf{r}, \mu, \nu}(X_{l}))^{\dagger} 
&= \sum_{\sigma \in S_{\mathbf{m}}} \delta\left(Q^{R}_{\mathbf{r}, \nu, \mu} \sigma \right) \mathcal{O}_{\sigma^{-1}}(X_{l}^{\dagger})
\\
&= \sum_{\tilde{\sigma} \in S_{\mathbf{m}}} \delta\left(Q^{R}_{\mathbf{r}, \nu, \mu} \tilde{\sigma}^{-1} \right) \mathcal{O}_{\tilde{\sigma}}(X_{l}^{\dagger})
\\
&= \mathcal{O}^{R}_{\mathbf{r}, \nu, \mu}(X_{l}^{\dagger})  \,.
\end{align}
where $(Q^{R}_{\mathbf{r}, \mu, \nu})^{\dagger} = Q^{R}_{\mathbf{r}, \nu, \mu}$ was applied and $\sigma^{\dagger} = \sigma^{-1}$. Additionally, $\mathcal{O}_{\sigma}(X_{l})$ is the permutation parameterised operator established in \eqref{eqn:mmm_observable_2}.
The correlator of these GIOs is then written as
\begin{align}
&\left\langle  \mathcal{O}^{R}_{\mathbf{r}, \mu_{1}, \nu_{1}}(X_{l}) (\mathcal{O}^{S}_{\mathbf{s},  \mu_{2}, \nu_{2}}(X_{l}))^{\dagger}  \right\rangle \nonumber 
\\
&= \sum_{\sigma, \tau \in S_{\mathbf{m}}} \delta\left(Q^{R}_{\mathbf{r}, \mu_{1}, \nu_{1}} \sigma^{-1}\right) \delta\left( Q^{S}_{\mathbf{s},\nu_{2}, \mu_{2}} \tau \right) \langle \mathcal{O}_{\sigma}(X_{l})\mathcal{O}_{\tau^{-1}}(X_{l}^{\dagger})\rangle
\\
{}&= \sum_{\sigma, \tau, \rho \in S_{\mathbf{m}}} \sum_{\gamma \in S_{m_{1}} \times \dots \times S_{m_{l}}}\delta\left(Q^{R}_{\mathbf{r}, \mu_{1}, \nu_{1}} \sigma^{-1}\right) \delta\left( Q^{S}_{\mathbf{s},\nu_{2}, \mu_{2}} \tau \right)  \mathcal{O}_{\rho}(B_{l}) \delta(\rho^{-1} \gamma^{-1} \sigma \gamma \tau^{-1}) \label{eqn:mmm_spb_corr_with_combo_basis_inserted}  \,,
\end{align}
where the result of the permutation parameterised basis, multi-matrix correlator from equation \eqref{eqn:mmm_2_point_result} was inserted in equation \eqref{eqn:mmm_spb_corr_with_combo_basis_inserted}.
\begin{lemma}
\label{lem:multi_matrix_schur_basis_lem}
For $\rho, \sigma, \tau \in S_{\mathbf{m}}$ and $\gamma \in S_{m_{1}} \times \dots \times S_{m_{l}}$ the following equality holds
\begin{align}
{}&\sum_{\sigma, \tau, \rho \in S_{\mathbf{m}}} \sum_{\gamma \in S_{m_{1}} \times \dots \times S_{m_{l}}}\delta\left(Q^{R}_{\mathbf{r}, \mu_{1}, \nu_{1}} \sigma^{-1}\right) \delta\left( Q^{S}_{\mathbf{s},\nu_{2}, \mu_{2}} \tau \right)  \mathcal{O}_{\rho}(B_{l}) \delta(\rho^{-1} \gamma^{-1} \sigma \gamma \tau^{-1}) \nonumber
\\
{}&= h \delta^{R S} \delta_{\mathbf{r} \mathbf{s}} \delta_{\nu_{1} \nu_{2}}  \mathcal{O}^{R}_{\mathbf{r}, \mu_{1}, \mu_{2}}( B_{l})\,,
\end{align}
where $h = \prod_{\alpha=1}^{l} m_{\alpha}! =|S_{m_{1}} \times \dots \times S_{m_{l}}|$.
\end{lemma}
\begin{proof}
\begin{align}
&\sum_{\sigma, \tau, \rho \in S_{\mathbf{m}}} \sum_{\gamma \in S_{m_{1}} \times \dots \times S_{m_{l}}}\delta\left(Q^{R}_{\mathbf{r}, \mu_{1}, \nu_{1}} \sigma^{-1}\right) \delta\left( Q^{S}_{\mathbf{s},\nu_{2}, \mu_{2}} \tau \right)  \mathcal{O}_{\rho}(B_{l}) \delta(\rho^{-1} \gamma^{-1} \sigma \gamma \tau^{-1}) \nonumber
\\
{}& = \sum_{ \rho \in S_{\mathbf{m}}}{} \sum_{\gamma \in S_{m_{1}} \times \dots \times S_{m_{l}}}\mathcal{O}_{\rho}(B_{l}) \delta(\rho^{-1} \gamma^{-1} Q^{R}_{\mathbf{r},\mu_{1},\nu_{1}} \gamma Q^{S}_{\mathbf{s},\nu_{2},\mu_{2}})  \label{eqn:mmm_spb_delta_sum}
\\
{}& = h \sum_{\rho \in S_{\mathbf{m}}} \mathcal{O}_{\rho}(B_{l}) \delta(\rho^{-1}  Q^{R}_{\mathbf{r},\mu_{1},\nu_{1}} Q^{S}_{\mathbf{s},\nu_{2},\mu_{2}} ) \label{eqn:mmm_spb_char_cyclicity}
\\
& = h \delta^{R S} \delta_{\mathbf{r} \mathbf{s}}  \delta_{\nu_{1} \nu_{2}} \left(  \sum_{ \rho \in S_{\mathbf{m}}}  \delta\left( Q^{R}_{\mathbf{r}, \mu_{1}, \mu_{2}} \rho^{-1} \right) \mathcal{O}_{\rho}(B_{l}) \right)\label{eqn:mmm_spb_deltas_in_rep_quants}
\\
& = h \delta^{R S} \delta_{\mathbf{r} \mathbf{s}}   \delta_{\nu_{1} \nu_{2}}   \mathcal{O}^{R}_{\mathbf{r}, \mu_{1}, \mu_{2}}(B_{l})
\end{align}
Equation \eqref{eqn:mmm_spb_delta_sum} sums over delta functions containing the $Q$-basis elements, inserting them in place of $\sigma$ and $\tau^{-1}$ in the final delta function. In equation \eqref{eqn:mmm_spb_char_cyclicity}, the invariance of  $Q^{R}_{\mathbf{r}, \mu_{1},\nu_{1}}$ under  conjugation by $\gamma$ was used to remove the $\gamma$ dependence, after which the sum was computed, introducing a factor of $h$. Multiplying the $Q$-basis elements as matrices following \eqref{eqn:q_basis_mmm_mult_property} produces equation \eqref{eqn:mmm_spb_deltas_in_rep_quants}. The final line simply identifies a Fourier basis operator in the witness fields, as per definition \eqref{eqn:mmm_spb_schur_gen_obs}.
\end{proof}
Using the above lemma, the correlator of \eqref{eqn:mmm_spb_corr_with_combo_basis_inserted} is 
\begin{equation}
\label{eqn:mmm_spb_corr_result_boxed}
\boxed{
\left\langle  \mathcal{O}^{R}_{\mathbf{r}, \mu_{1}, \nu_{1}}(X_{l}) (\mathcal{O}^{S}_{\mathbf{s}, \mu_{2}, \nu_{2}}(X_{l}))^{\dagger}  \right\rangle = h \delta^{R S} \delta_{\mathbf{r} \mathbf{s}}  \delta_{\nu_{1} \nu_{2}}  \mathcal{O}^{R}_{\mathbf{r}, \mu_{1} \mu_{2}}(B_{l})
}
\end{equation}
This exhibits the orthogonality relationship as observed in the one and two-matrix cases and shows the correlator is proportional to a Fourier operator composed of witness fields $B_{l}$.


\subsection{Observable functions of  multi-matrix quantum and classical fields} 
\label{ss:back_fields_in_observables_for_mmm}

Obtaining the multi-matrix correlator result through classical witness fields is possible and follows the same procedure as \S\ref{ss:back_fields_in_observables} and \S\ref{ss:back_fields_in_observables_for_tm}. 
Define the partition function
\begin{equation}
\label{eqn:mmm_bmf_part_func_no_bmf}
\Sigma[0] =  \prod_{\eta=1}^{l}\int [dX_{\eta}] e^{-\text{Tr}(X_{\eta} X_{\eta}^{\dagger})} \,.
\end{equation}
which has basic correlator 
\begin{equation}
\label{eqn:mmm_no_bmf_basic_corr}
\langle (X_{\alpha})^{i}_{j} (X_{\beta}^{\dagger})^{k}_{l} \rangle =  \prod_{\eta=1}^{l} \frac{1}{\Sigma[0]}\int [dX_{\eta}](X_{\alpha})^{i}_{j} (X_{\beta}^{\dagger})^{k}_{l} e^{-\text{Tr}(X_{\eta} X_{\eta}^{\dagger})} = \delta_{\alpha, \beta} \delta^{i}_{l} \delta^{k}_{j}\,.
\end{equation}
Then, using permutation parameterised GIO 
\begin{align}
\begin{split}
\label{eqn:mmm_back_class_obs_definition}
\mathcal{O}_{\sigma}(Z_{l}A_{l}) 
&= \text{Tr}_{V_{N}^{\otimes \mathbf{m}}} \left( (Z_{1}A_{1})^{\otimes m_{1}} \otimes \dots \otimes (Z_{l}A_{l})^{\otimes m_{l}} \cL_{\sigma} \right) 
\\
&= (Z_{1}A_{1})^{i_{1}}_{i_{\sigma(1)}} \dots(Z_{1}A_{1})^{i_{m_{1}}}_{i_{\sigma(m_{1})}}  \dots (Z_{l}A_{l})^{i_{\mathbf{m}-m_{l}}}_{i_{\sigma(\mathbf{m}-m_{l})}} \dots(Z_{l}A_{l})^{i_{\mathbf{m}}}_{i_{\sigma(\mathbf{m})}} 
\end{split}
\end{align}
and its Hermitian conjugate
\begin{align}
\begin{split}
\label{eqn:mmm_back_class_definition_herm}
\left( \mathcal{O}_{\sigma}(Z_{l}A_{l}) \right)^{\dagger} 
&= \text{Tr}_{V_{N}^{\otimes \mathbf{m}}} \left( (A_{1}^{\dagger} Z_{1}^{\dagger})^{\otimes m_{1}} \otimes \dots \otimes ( A_{l}^{\dagger}Z_{l}^{\dagger})^{\otimes m_{l}}\cL_{\sigma^{-1}} \right) 
\\
&= (A_{1}^{\dagger} Z_{1}^{\dagger})^{i_{1}}_{i_{\sigma^{-1}(1)}} \dots(A_{1}^{\dagger} Z_{1}^{\dagger})^{i_{m_{1}}}_{i_{\sigma^{-1}(m_{1})}}  \dots (A_{l}^{\dagger} Z_{l}^{\dagger})^{i_{\mathbf{m}-m_{l}}}_{i_{\sigma^{-1}(\mathbf{m}-m_{l})}} \dots(A_{l}^{\dagger} Z_{l}^{\dagger})^{i_{\mathbf{m}}}_{i_{\sigma^{-1}(\mathbf{m})}}  
\\
&= \mathcal{O}_{\sigma^{-1}}(A_{l}^{\dagger} Z_{l}^{\dagger})
 \,,
\end{split}
\end{align}

the GIO two-point function of classical witness fields is then 
\begin{align}
\langle  \mathcal{O}_{\sigma_{1}} {}& (X_{l} A_{l}) \left(\mathcal{O}_{\sigma_{2}}(X_{l }A_{l}) \right)^{\dagger}  \rangle \nonumber
\\
={}& \prod_{\eta=1}^{l} \frac{1}{\Sigma[0]} \int [dX_{\eta}] \mathcal{O}_{\sigma_{1}}(X_{l} A_{l}) \left(\mathcal{O}_{\sigma_{2}}(X_{l }A_{l}) \right)^{\dagger} e^{-\text{Tr}(X_{\eta} X_{\eta}^{\dagger})} \label{eqn:mmm_bmf_in_ob_corr}
\\
\begin{split}
={}& \sum_{I,J,K,Q} (A_{1})^{k_{1}}_{i_{\sigma_{1}(1)}} \dots (A_{1})^{k_{m_{1}}}_{i_{\sigma_{1}(m_{1})}}  \dots (A_{l})^{k_{\mathbf{m}- m_{l}}}_{i_{\sigma_{1}(\mathbf{m}- m_{l})}} \dots (A_{l})^{k_{\mathbf{m}}}_{i_{\sigma_{1}(\mathbf{m})}} (A^{\dagger}_{1})^{j_{1}}_{q_{1}} \dots (A^{\dagger}_{1})^{j_{m_{1}}}_{q_{m_{1}}} 
\\
{}& \times (A^{\dagger}_{l})^{j_{\mathbf{m} - m_{l}}}_{q_{\mathbf{m} - m_{l}}} \dots (A^{\dagger}_{l})^{j_{\mathbf{m}}}_{q_{\mathbf{m}}} \prod_{\eta=1}^{l} \frac{1}{\Sigma[0]} \int [dX_{\eta}] (Z_{1})^{i_{1}}_{k_{1}} \dots (Z_{1})^{i_{m_{1}}}_{k_{m_{1}}} \dots  (Z_{l})^{i_{\mathbf{m} - m_{l}}}_{k_{\mathbf{m} - m_{l}}} \dots (Z_{l})^{i_{\mathbf{m}}}_{k_{\mathbf{m}}} 
\\
{}& \times  (Z_{1}^{\dagger})^{q_{1}}_{j_{\sigma_{2}^{-1}(1)}} \dots (Z_{1}^{\dagger})^{q_{m_{1}}}_{j_{\sigma_{2}^{-1}(m_{1})}} \dots (Z_{l}^{\dagger})^{q_{\mathbf{m}-m_{l}}}_{j_{\sigma_{2}^{-1}(\mathbf{m}-m_{l})}} \dots (Z_{l}^{\dagger})^{q_{\mathbf{m}}}_{j_{\sigma_{2}^{-1}(\mathbf{m})}} e^{-\text{Tr}(X_{\eta} X_{\eta}^{\dagger})} \,,
\end{split}
\end{align}
where the sum over $I,J,K,Q$ covers all matrix indices. To simplify the notation, a function in the witness fields is defined
\begin{align}
\begin{split}
f_{l}(A, A^{\dagger}; \vec{i}, \vec{j}, \vec{k},\vec{q}; \sigma_{1} )
={} (A_{1})^{k_{1}}_{i_{\sigma_{1}(1)}}  \dots &(A_{1})^{k_{m_{1}}}_{i_{\sigma_{1}(m_{1})}}  \dots (A_{l})^{k_{\mathbf{m}- m_{l}}}_{i_{\sigma_{1}(\mathbf{m}- m_{l})}} \dots (A_{l})^{k_{\mathbf{m}}}_{i_{\sigma_{1}(\mathbf{m})}}
\\
& \times (A^{\dagger}_{1})^{j_{1}}_{q_{1}} \dots (A^{\dagger}_{1})^{j_{m_{1}}}_{q_{m_{1}}} \dots (A^{\dagger}_{l})^{j_{\mathbf{m} - m_{l}}}_{q_{\mathbf{m} - m_{l}}} \dots (A^{\dagger}_{l})^{j_{\mathbf{m}}}_{q_{\mathbf{m}}} \,.
\end{split}
\end{align}
Using this notation the correlator is
\begin{align}
\langle  \mathcal{O}_{\sigma_{1}} {}& (X_{l} A_{l}) \left(\mathcal{O}_{\sigma_{2}}(X_{l }A_{l}) \right)^{\dagger}  \rangle \nonumber
\\
\begin{split}
={}& \sum_{I,J,K,Q} f_{l}(A, A^{\dagger}; \vec{i}, \vec{j}, \vec{k},\vec{q}; \sigma_{1} )  \prod_{\eta=1}^{l} \frac{1}{\Sigma[0]} \int [dX_{\eta}] (Z_{1})^{i_{1}}_{k_{1}} \dots (Z_{1})^{i_{m_{1}}}_{k_{m_{1}}} \dots  (Z_{l})^{i_{\mathbf{m} - m_{l}}}_{k_{\mathbf{m} - m_{l}}}  
\\
{}& \times (Z_{l})^{i_{\mathbf{m}}}_{k_{\mathbf{m}}}  \dots (Z_{1}^{\dagger})^{q_{1}}_{j_{\sigma_{2}^{-1}(1)}} \dots (Z_{1}^{\dagger})^{q_{m_{1}}}_{j_{\sigma_{2}^{-1}(m_{1})}} \dots (Z_{l}^{\dagger})^{q_{\mathbf{m}-m_{l}}}_{j_{\sigma_{2}^{-1}(\mathbf{m}-m_{l})}} \dots (Z_{l}^{\dagger})^{q_{\mathbf{m}}}_{j_{\sigma_{2}^{-1}(\mathbf{m})}} e^{-\text{Tr}(X_{\eta} X_{\eta}^{\dagger})}
\end{split}
\\
\begin{split}
={}& \sum_{I,J,K,Q} f_{l}(A, A^{\dagger}; \vec{i}, \vec{j}, \vec{k},\vec{q}; \sigma_{1} )  \langle (Z_{1})^{i_{1}}_{k_{1}} \dots (Z_{1})^{i_{m_{1}}}_{k_{m_{1}}} \dots  (Z_{l})^{i_{\mathbf{m} - m_{l}}}_{k_{\mathbf{m} - m_{l}}} \dots (Z_{l})^{i_{\mathbf{m}}}_{k_{\mathbf{m}}}  
\\
& \qquad \qquad \qquad \times (Z_{1}^{\dagger})^{q_{1}}_{j_{\sigma_{2}^{-1}(1)}} \dots (Z_{1}^{\dagger})^{q_{m_{1}}}_{j_{\sigma_{2}^{-1}(m_{1})}} \dots (Z_{l}^{\dagger})^{q_{\mathbf{m}-m_{l}}}_{j_{\sigma_{2}^{-1}(\mathbf{m}-m_{l})}} \dots (Z_{l}^{\dagger})^{q_{\mathbf{m}}}_{j_{\sigma_{2}^{-1}(\mathbf{m})}}  \rangle \label{eqn:mmm_bmf_f_apply_wicks}
\end{split}
\\
\begin{split}
={}& \sum_{I,J,K,Q} \sum_{\gamma \in S_{m_{1}} \times \dots \times S_{m_{l}}} f_{l}(A, A^{\dagger}; \vec{i}, \vec{j}, \vec{k},\vec{q}; \sigma_{1} )\delta^{i_{1}}_{j_{\gamma \sigma_{2}^{-1}(1)}} \delta^{q_{1}}_{k_{\gamma^{-1}(1)}} \dots \dots \delta^{i_{\mathbf{m}}}_{j_{\gamma \sigma_{2}^{-1}(\mathbf{m})}} \delta^{q_{\mathbf{m}}}_{k_{\gamma^{-1}(\mathbf{m})}} \label{eqn:mmm_bmf_f_func_and_deltas}
\end{split}
\\
\begin{split}
={}&  \sum_{j_{1},  \dots, j_{\mathbf{m}}} \sum_{\gamma \in S_{m_{1}} \times \dots \times S_{m_{l}}} \left(A^{\dagger}_{1} A_{1}\right)_{j_{\gamma^{-1} \sigma_{1} \gamma \sigma_{2}^{-1}(1)}}^{j_{1}} \dots \left(A^{\dagger}_{1} A_{1}\right)_{j_{\gamma^{-1} \sigma_{1} \gamma \sigma_{2}^{-1}(m_{1})}}^{j_{m_{1}}}   \dots
\\
&{} \qquad \qquad \qquad \qquad \qquad  \dots \times \left(A^{\dagger}_{l} A_{l}\right)_{j_{\gamma^{-1} \sigma_{1} \gamma \sigma_{2}^{-1}(\mathbf{m}-m_{l})}}^{j_{\mathbf{m}-m_{l}}} \dots \left(A^{\dagger}_{l} A_{l}\right)_{j_{\gamma^{-1} \sigma_{1} \gamma \sigma_{2}^{-1}(\mathbf{m})}}^{j_{\mathbf{m}}}  \label{eqn:mmm_bmf_A_Adag_with_trace_on_j}
\end{split}
\\
\begin{split}
={}& \sum_{j_{1}, \dots j_{\mathbf{m}}} \sum_{\gamma \in S_{m_{1}} \times \dots \times S_{m_{l}}} (B_{1})^{j_{1}}_{j_{\gamma^{-1}\sigma_{1} \gamma \sigma_{2}^{-1}(1)}} \dots (B_{1})^{j_{m_{1}}}_{j_{\gamma^{-1}\sigma_{1} \gamma \sigma_{2}^{-1}(m_{1})}}  \dots 
\\
&{} \qquad \qquad \qquad \qquad \qquad \dots \times   (B_{l})^{j_{(\mathbf{m}-m_{l})}}_{j_{\gamma^{-1}\sigma_{1} \gamma \sigma_{2}^{-1}(\mathbf{m}-m_{l})}} \cdots (B_{l})^{j_{(\mathbf{m})}}_{j_{\gamma^{-1}\sigma_{1} \gamma \sigma_{2}^{-1}(\mathbf{m})}} \label{eqn:mmm_back_class_setting_a_mats_to_b}
\end{split}
\\
={}& \sum_{\gamma \in S_{m_{1}} \times \dots \times S_{m_{l}}}\text{Tr}_{V_{N}^{\otimes \mathbf{m}}}\left((B_{1})^{\otimes m_{1}} \otimes \dots \otimes (B_{l})^{\otimes m_{l}} \cL_{\gamma^{-1} \sigma_{1} \gamma \sigma_{2}^{-1}} \right)
\\
={}& \sum_{\sigma_{3} \in S_{\mathbf{m}}} \sum_{ \gamma \in S_{m_{1}} \times \dots \times S_{m_{l}} } \mathcal{O}_{\sigma_{3}}(B_{l}) \delta(\sigma_{3}^{-1} \gamma^{-1} \sigma_{1} \gamma \sigma_{2}^{-1}))\label{eqn:mmm_back_class_delta_and_gio_in_back}
\\
={}& \sum_{p_{3} } \sum_{\gamma \in S_{m_{1}} \times \dots \times S_{m_{l}}}  \left(  \sum_{\alpha \in \text{Orb}_{\mathcal{\vec{A}}}(p_{3})}\mathcal{O}_{\alpha}(B_{l}) \delta(\alpha^{-1} \gamma^{-1} \sigma_{1} \gamma \sigma_{2}^{-1}) \right) \label{eqn:mmm_back_class_equiv_class_split}
\\
={}&    \sum_{p_{3} } \mathcal{O}_{\sigma^{\left(p_{3}\right)}}(B_{l})\sum_{\gamma \in S_{m_{1}} \times \dots \times S_{m_{l}}} \delta(T^{\mathcal{\vec{A}}}_{p'_{3}} \gamma^{-1} \sigma_{1} \gamma \sigma_{2}^{-1}) \label{eqn:mmm_back_class_alg_ele_in_delta}
\\
={}& \sum_{p_{3} } \frac{|T^{\mathcal{\vec{A}}}_{p'_{3}}|\prod_{\alpha = 1}^{l} m_{\alpha}!}{|T^{\mathcal{\vec{A}}}_{p_{1}} | |T^{\mathcal{\vec{A}}}_{p'_{2}}|} C^{p'_{3}; \mathcal{\vec{A}}}_{p_{1} p'_{2}} \mathcal{O}_{\sigma^{\left(p_{3}\right)}}(B_{l})\label{eqn:mmm_back_class_corr_with_pca_struct_const} 
\end{align}
In the above, Kronecker equivariance, Wick's theorem and index contractions were implemented to arrive at \eqref{eqn:mmm_bmf_A_Adag_with_trace_on_j} from \eqref{eqn:mmm_bmf_f_apply_wicks}. $A^{\dagger}_{\eta} A_{\eta} = B_{\eta}$ was set in \eqref{eqn:mmm_back_class_setting_a_mats_to_b}, and the GIO of classical witness fields along with a delta function were used in \eqref{eqn:mmm_back_class_delta_and_gio_in_back}. Equations \eqref{eqn:mmm_back_class_equiv_class_split} and \eqref{eqn:mmm_back_class_alg_ele_in_delta} split the sum into equivalence classes/orbits (labelled $p_{3}$), and a sum over the  elements of orbit $\text{Orb}_{\mathcal{\vec{A}}}(p_{3})$ (labelled $\alpha$). The PCA element $T_{p'_{3}}^{\mathcal{\vec{A}}}$ is identified by taking the $\alpha$ sum inside the delta function, to sum over $\alpha^{-1}$. Finally, lemma \ref{lem:multi-matrix_lemma} was utilised in \eqref{eqn:mmm_back_class_corr_with_pca_struct_const} to bring the correlator to the form with manifest PCA structure constant dependence.  The prime label on $p'_{i}$ is to indicate that it is the partition/equivalence class of an inverse permutation. Employing $|T_{p_{i}}^{\mathcal{\vec{A}}}|\mathcal{O}_{\sigma^{\left( p_{i}\right)}} \equiv |T_{p_{i}}^{\mathcal{\vec{A}}}|\mathcal{O}_{\sigma_{i}} = \mathcal{O}_{p_{i}}$, and $|T_{p'_{i}}| =  |T_{p_{i}}|$, we achieve the same combinatorial basis result as equation \eqref{eqn:mmm_final_correlator_boxed} for the coupling matrix derivation after rearranging the orbit size factors
\begin{equation}
\boxed{
\langle \mathcal{O}_{p_{1}}(X_{l}A_{l}) \left(\mathcal{O}_{p_{2}}(X_{l}A_{l})\right)^{\dagger} \rangle = \prod_{\alpha=1}^{l} m_{\alpha}! \sum_{p_{3} } C^{p'_{3}; \mathcal{\vec{A}}}_{p_{1} p'_{2}} \mathcal{O}_{p_{3}}(B_{l})
}
\end{equation}
Figures may be generated for the multi-matrix case as a straightforward generalisation to both Figure \ref{fig:tm_two_perm_bmf_correlator_diagram} and \ref{fig:two_mat_corr_bmf_wihtin_obs_png}, by inserting the required number of box operators to match the desired  number of witness matrices. 


\subsection{Fourier/\texorpdfstring{$Q$}{Q} basis for multi-matrix, two-point function with quantum and classical fields}
\label{ss:schur_basis_back_fields_in_observables_mult_mat}
Define the GIO for the classical witness fields in Fourier basis as 
\begin{equation}
\label{eqn:mmm_spb_schur_gen_obs_back_class}
\mathcal{O}^{R}_{\mathbf{r}, \mu, \nu}(X_{l}A_{l}) = \sum_{\sigma \in S_{\mathbf{m}}} \delta\left(Q^{R}_{\mathbf{r}, \mu, \nu} \sigma^{-1} \right) \mathcal{O}_{\sigma}(X_{l}A_{l}) \,,
\end{equation}
with conjugate 
\begin{align}
\label{eqn:mmm_spb_schur_gen_obs_back_class_conj}
(\mathcal{O}^{R}_{\mathbf{r}, \mu, \nu}(X_{l}A_{l}))^{\dagger} 
&= \sum_{\sigma \in S_{\mathbf{m}}} \delta\left(Q^{R}_{\mathbf{r}, \nu, \mu} \sigma \right) \mathcal{O}_{\sigma^{-1}}(A_{l}^{\dagger} X_{l}^{\dagger})
\\
&= \sum_{\tilde{\sigma} \in S_{\mathbf{m}}} \delta\left(Q^{R}_{\mathbf{r}, \nu, \mu} \tilde{\sigma}^{-1} \right) \mathcal{O}_{\tilde{\sigma}}(A_{l}^{\dagger} X_{l}^{\dagger})
\\
&= \mathcal{O}^{R}_{\mathbf{r}, \nu, \mu}(A_{l}^{\dagger} X_{l}^{\dagger})  \,.
\end{align}
where $(Q^{R}_{\mathbf{r}, \mu, \nu})^{\dagger} =Q^{R}_{\mathbf{r}, \nu, \mu}$, and $\sigma^{\dagger} = \sigma^{-1}$ were applied.
Using the permutation parameterised correlator result from \eqref{eqn:mmm_back_class_delta_and_gio_in_back}, repeated below for convenience, 
\begin{align}
\langle \mathcal{O}_{\sigma}(X_{l} A_{l}) \left(\mathcal{O}_{\tau}(X_{l }A_{l}) \right)^{\dagger}  \rangle 
{}&= \langle \mathcal{O}_{\sigma}(X_{l} A_{l}) \mathcal{O}_{\tau^{-1}}(A_{l}^{\dagger}X_{l}^{\dagger})  \rangle  
\\
{}& = \sum_{\rho \in S_{\mathbf{m}}} \sum_{ \gamma \in S_{m_{1}} \times \dots \times S_{m_{l}} } \mathcal{O}_{\rho}(B_{l}) \delta(\rho^{-1} \gamma^{-1} \sigma \gamma \tau^{-1}) \,,
\end{align}
where $B_{l} = A_{l}^{\dagger} A_{l}$ was set, the classical witness field correlator in Fourier basis is 
\begin{align}
\left\langle \right. &  \mathcal{O}^{R}_{\mathbf{r}, \mu_{1}, \nu_{1}}(X_{l}A_{l})  \left(\mathcal{O}^{S}_{\mathbf{s}, \mu_{2}, \nu_{2}}(X_{l}A_{l})\right)^{\dagger}  \left. \right\rangle  \nonumber
\\
&= \sum_{\sigma, \tau \in S_{\mathbf{m}}} \delta\left(Q^{R}_{\mathbf{r},\mu_{1},\nu_{1}} \sigma^{-1}\right) \delta\left( Q^{S}_{\mathbf{s},\nu_{2},\mu_{2}}\tau \right) \langle \mathcal{O}_{\sigma}(X_{l}A_{l}) \left(\mathcal{O}_{\tau}(X_{l}A_{l} )\right)^{\dagger} \rangle
\\
&= \sum_{\sigma, \tau, \rho \in S_{\mathbf{m}}}  \sum_{\gamma \in S_{m_{1}} \times \dots \times S_{m_{l}}} \delta\left(Q^{R}_{\mathbf{r},\mu_{1},\nu_{1}} \sigma^{-1}\right) \delta\left( Q^{S}_{\mathbf{s},\nu_{2},\mu_{2}}\tau \right) \mathcal{O}_{\rho}(B_{l}) \delta(\rho^{-1} \gamma^{-1} \sigma \gamma \tau^{-1}) 
\\
&= h  \delta^{R S} \delta_{\mathbf{r} \mathbf{s}} \delta_{\nu_{1} \nu_{2}}  \mathcal{O}^{R}_{\mathbf{r}, \mu_{1}, \mu_{2}}(B_{l}) \,,
\end{align}
where lemma \ref{lem:multi_matrix_schur_basis_lem} was used to obtain the final line, and again $h = \prod_{\alpha=1}^{l
} m_{\alpha}!$. This correlator result
\begin{equation}
\label{eqn:mmm_schur_classical_corr_result_boxed}
\boxed{
\left\langle   \mathcal{O}^{R}_{\mathbf{r}, \mu_{1}, \nu_{1}}(X_{l}A_{l})  \left(\mathcal{O}^{S}_{\mathbf{s}, \mu_{2}, \nu_{2}}(X_{l}A_{l})\right)^{\dagger} \right\rangle  =  h  \delta^{R S} \delta_{\mathbf{r} \mathbf{s}} \delta_{\nu_{1} \nu_{2}}  \mathcal{O}^{R}_{\mathbf{r}, \mu_{1}, \mu_{2}}(B_{l})
}
\end{equation}

generalises the previous
findings of  \S\ref{ss:schur_basis_back_fields_in_observables_two_mat} to an arbitrary number of witness matrices, while again confirming the orthogonality relations of the two GIO correlator by virtue of the representation theoretic delta functions.

\newpage

\section{Algebras and tensor correlators with tensor witnesses }
\label{s:tensor_model}
A further extension to the witness field discussion is to introduce tensor invariants \cite{BenGeloun:2013lim,BenGeloun:2017vwn}. Tensor operators of rank $d$ are invariant under $U(N)^{\times d}$ and labelled by permutations ($\sigma_{1}, \sigma_{2}, \dots, \sigma_{d}$). As such they are constructed by contracting a given tensor, $Z^{i_{1} \dots i_{d}}$, with its complex conjugate\footnote{In this tensor case, complex conjugated quantities will feature indices down out of convention.}, $\bar{Z}_{i_{1}, \dots i_{d}}$. A general permutation parameterised GIO is then written as 
\begin{align}
\label{eqn:tens_gio_def_for_tensor}
\begin{split}
\mathcal{O}_{\sigma_{1}, \dots, \sigma_{d}} &(\bar{Z}, Z) 
\\
&= \text{Tr}_{V_{1}^{\otimes m} \otimes \dots \otimes V_{d}^{\otimes m}} \left( Z^{\otimes m}   \bar{Z}^{\otimes m}  (\cL_{ \sigma_{1} }  \otimes \dots \otimes \cL_{ \sigma_{d} } )  \right) \\
&= \sum_{I} \bar{Z}_{(i_{1})_{\sigma_{1}(1)} \dots (i_{d})_{\sigma_{d}(1)}} \dots \bar{Z}_{(i_{1})_{\sigma_{1}(m)} \dots (i_{d})_{\sigma_{d}(m)}} Z^{(i_{1})_{1} \dots (i_{d})_{1}} \dots Z^{(i_{1})_{m} \dots (i_{d})_{m}} 
\end{split}
\end{align}
where the sum over $I$ is over the full set of indices $I = \{(i_{1})_{1}, \dots, (i_{d})_{1}, \dots,  (i_{1})_{m}, \dots, (i_{d})_{m} \}$. The tensor invariants have the underlying equivalence relation
\begin{equation}
\label{eqn:tens_tensor_equiv_relation}
(\sigma_{1}, \sigma_{2}, \dots, \sigma_{d}) \sim (\mu_{1} \sigma_{1} \mu_{2}, \mu_{1} \sigma_{2} \mu_{2},\dots, \mu_{1} \sigma_{d} \mu_{2}) \,,
\end{equation}
where $\sigma_{i}\,, \mu_{1}\,, \mu_{2} \in S_{m} $.	
This is known as the equivalence under left-right diagonal action of $S_{m}$ on $S_{m}^{\times d}$, and means that equivalent permutation tuples produce the same operator
\begin{equation}
\label{eqn:tc_obs_invariance}
\mathcal{O}_{\sigma_{1}, \sigma_{2}, \dots, \sigma_{d}}(\bar{Z}, Z) = \mathcal{O}_{\mu_{1}\sigma_{1}\mu_{2}, \mu_{1}\sigma_{2}\mu_{2}, \dots, \mu_{1}\sigma_{d}\mu_{2}}(\bar{Z}, Z) \,.
\end{equation}
In keeping with the previous matrix discussion, such a relation can be linked to a permutation centralizer algebra (PCA) which may be used to characterise tensor correlators. To help produce an explicit and concise demonstration of such correlator calculations, from here on we specialise to $d=3$. The PCA associated with the tensor model under study is denoted $\mathcal{K}(m)$ and otherwise referred to as the ``Kronecker PCA". This was originally discussed in detail in \cite{BenGeloun:2017vwn} where both a gauged version $\mathcal{K}(m)$ and an un-gauged version $\mathcal{K}_{\text{un}}(m)$ were investigated for their connection to the counting of tensor invariants\footnote{$\mathcal{K}_{\text{un}}(m)$ and $\mathcal{K}(m)$ are isomorphic and can be used interchangeably, both being referred to as the Kronecker PCA.}. There, $\mathcal{K}_{\text{un}}(m)$ was defined as the vector space and sub-algebra of $\mathbb{C}(S_{m})^{\otimes 3}$ which is invariant under the left and right action of the diagonal symmetric group algebra\footnote{$\mathcal{K}_{\text{un}}(m)$ is also called the ``double coset algebra" in \cite{BenGeloun:2017vwn}.} $\text{Diag}[\mathbb{C}(S_{m})]$, i.e.
\begin{equation}
\label{eqn:tc_kn_algebra_tensor}
\mathcal{K}_{\text{un}}(m) = \text{Span}_{\mathbb{C}} \left\{ \sum_{\gamma_{1}, \gamma_{2} \in S_{m}}  \gamma_{1} \sigma_{1} \gamma_{2} \otimes \gamma_{1} \sigma_{2} \gamma_{2} \otimes \gamma_{1} \sigma_{3} \gamma_{2}\,, \quad \sigma_{1}\,,\sigma_{2}\,,\sigma_{3} \in S_{m}\right\}  \,.
\end{equation}
This is also the definition used in the previous matrix models, where the algebra elements are formed from group averages. An important relation between this PCA and graphs was at the centre of \cite{BenGeloun:2013lim,BenGeloun:2017vwn,BenGeloun:2020yau,Geloun:2022kma}, where a graph basis for the algebra was implemented. This basis is simply a re-description of the combinatorial/permutation basis so far used, by labelling the algebra in terms of unique graphs. Counting these graphs connects directly to counting the elements in $\mathcal{K}_{\text{un}}(m)$, as they both share the same equivalence relation of \eqref{eqn:tens_tensor_equiv_relation} (see \S 2 of \cite{BenGeloun:2020yau}
for more details on the construction of these graphs). 

This relation and fact \eqref{eqn:tc_obs_invariance} apply to $\mathcal{K}_{\text{un}}(m)$ specifically and in terms of graph nomenclature, $\mathcal{K}_{\text{un}}(m)$ corresponds to bipartite graphs with $m$ trivalent vertices with three incoming coloured edges and $m$ trivalent vertices with three outgoing coloured edges. Unique  graphs are hence equivalence classes of triples $(\tau_{1}, \tau_{2}, \tau_{3}) \in S_{m} \times S_{m} \times S_{m}$. The expression $\text{Col}(m)$ is used to denote the set of such equivalence classes, where $m$ labels the number of vertices, and $r \in \{1, \dots, |\text{Col}(m)| \}$ is the label for each individual graph of chosen $m$. Alluding to the counting mentioned earlier, it is  a fact that $|\text{Col}(m)| = |\mathcal{K}_{\text{un}}(m)|$, i.e. the number of graphs is equal to the dimension of the algebra. To define the PCA elements used in the subsequent discussion, we additionally specify $\text{Aut}_{\mathcal{K}}(r)$ as the subgroup of $\text{Diag}[S_{m}] \times \text{Diag}[S_{m}]$ leaving triples $(\sigma_{1}^{(r)},\sigma_{2}^{(r)},\sigma_{3}^{(r)}) \in S_{m} \times S_{m} \times S_{m}$ fixed, and $\text{Orb}_{\mathcal{K}}(r)$ as the set of unique triples of permutations obtained from the left-right diagonal action of $S_{m}$. The orbit-stabilizer theorem then allows one to connect the size of these orbits, $|\text{Orb}_{\mathcal{K}}(r)|$, to the size of the $\text{Aut}_{\mathcal{K}}(r)$ group, $|\text{Aut}_{\mathcal{K}}(r)|$. More precisely, the theorem states that there is an isomorphism such that
\begin{equation}
 \frac{1}{|\text{Orb}_{\mathcal{K}}(r)|} = \frac{|\text{Aut}_{\mathcal{K}}(r)|}{|\text{Diag}[S_{m}] \times \text{Diag}[S_{m}]|} = \frac{|\text{Aut}_{\mathcal{K}}(r)|}{(m!)^2}
\end{equation}
Having briefly introduced the background and notation, consider a graph basis element for $\mathcal{K}_{\text{un}}(m)$ as
\begin{align}
E_{r} &= \frac{1}{(m!)^2} \sum_{\mu_{1}, \mu_{2} \in S_{m}} \mu_{1} \sigma_{1}^{(r)} \mu_{2} \otimes \mu_{1} \sigma_{2}^{(r)} \mu_{2} \otimes \mu_{1} \sigma_{3}^{(r)} \mu_{2} \label{eqn:tc_tens_e_basis_1}\\
& = \frac{|\text{Aut}_{\mathcal{K}}(r)|}{(m!)^2} \sum_{a \in \text{Orb}_{\mathcal{K}}(r)} \sigma_{1}^{(r)}(a) \otimes \sigma_{2}^{(r)}(a) \otimes \sigma_{3}^{(r)}(a) \\
&= \frac{1}{|\text{Orb}_{\mathcal{K}}(r)|} \sum_{a \in \text{Orb}_{\mathcal{K}}(r)} \sigma_{1}^{(r)}(a) \otimes \sigma_{2}^{(r)}(a) \otimes \sigma_{3}^{(r)}(a) \quad \in \mathcal{K}_{\text{un}}(m) \label{eqn:tc_tens_e_basis_3}\,.
\end{align}
Here $a$ labels the distinct permutation triples that are in the same orbit, such that $(\sigma_{1}^{(r)}(a) \otimes \sigma_{2}^{(r)}(a) \otimes \sigma_{3}^{(r)}(a))$ is a representative triple of the orbit denoted $\text{Orb}_{\mathcal{K}}(r)$.  As an aside, the $T$ basis used in the previous matrix models (for example, equation \eqref{eqn:atmb_algebra_element_necklace} for the Necklace PCA) is straightforwardly related to the graph $E$ basis for the Kronecker PCA via normalisation
\begin{equation}
E_{r} = \frac{T^{\mathcal{K}}_{r}}{|T^{\mathcal{K}}_{r}|}\,, \qquad T^{\mathcal{K}}_{r} = \sum_{a \in \text{Orb}_{\mathcal{K}}(r)} \sigma_{1}^{(r)}(a) \otimes \sigma_{2}^{(r)}(a) \otimes \sigma_{3}^{(r)}(a)\,.
\end{equation}
i.e. $T^{\mathcal{K}}_{r}$ is the un-normalised basis element of $\mathcal{K}_{\text{un}}(m)$ and $|T^{\mathcal{K}}_{r}| = |\text{Orb}_{\mathcal{K}}(r)|$.

Having addressed the tensor model's algebra, the correlator of tensor GIOs with classical witness fields is calculated in \S \ref{ss:background_tensor_fields}. Note that using a coupling witness field faces difficulties due to a lack of well-defined tensor inverse, hence only classical witness fields are implemented. Additionally, since $\mathcal{K}_{\text{un}}(m)$ is semi-simple and associative, the Wedderburn-Artin theorem \cite{ram1991representation} allows us to describe it as a direct sum of matrix algebras. Using this fact, the correlator calculations can be organised in terms of representation theoretic quantities by Fourier transforming the GIOs and subsequently using the multiplication of the Fourier basis matrix elements. This serves to generalise the findings of \cite{BenGeloun:2017vwn} to include a classical tensor field in an additional basis, and is the subject of \S \ref{ss:tensor_fourier_basis}.


\subsection{Two-point function of general tensor operators with tensor quantum and classical fields}
\label{ss:background_tensor_fields}

To calculate tensor correlators with classical tensor witnesses, we define the gauge invariant operators to include both classical field, $\Lambda$, and field variable $Z$, as rank-3 tensors. The permutation $\sigma_{i} \in S_{m}$ then acts on the $i$-th index of the tensor as per definition in \eqref{eqn:tens_gio_def_for_tensor}. Explicitly, 
we will use the  gauge invariant operators : 
\begin{align}
\label{eqn:tenmod_observable_d_3}
\begin{split}
\mathcal{O}_{\sigma_{1}, \sigma_{2}, \sigma_{3}}(\bar{\Lambda}, Z)  
&= \text{Tr}_{V_{1}^{\otimes m} \otimes V_{2}^{\otimes m} \otimes  V_{3}^{\otimes m}} \left[  Z^{\otimes m}
 \bar{\Lambda}^{\otimes m}  \left( \cL_{ \sigma_{1}}  \otimes \cL_{ \sigma_{2} } \otimes \cL_{ \sigma_{3}}  \right)  \right] 
\\
& = \sum_{I,J,K} \bar{\Lambda}_{i_{\sigma_{1}(1)} j_{\sigma_{2}(1)} k_{\sigma_{3}(1)}} \cdots \bar{\Lambda}_{{i_{\sigma_{1}(m)}} j_{\sigma_{2}(m)} k_{\sigma_{3}(m)}}  Z^{i_{1} j_{1} k_{1}} \cdots Z^{i_{m} j_{m} k_{m}}  \,,
\end{split}
\end{align} 
The permutation operators have action
\begin{align}
\begin{split}
{}&(\cL_{\sigma_{1}} \otimes \cL_{\sigma_{2}}  \otimes \cL_{\sigma_{3}} )\ket{(e_{i_{1}} \otimes \dots \otimes e_{i_{m}}) \otimes (e_{j_{1}} \otimes \dots \otimes e_{j_{m}}) \otimes (e_{k_{1}} \otimes \dots \otimes e_{k_{m}}) }
\\
{}&= \ket{(e_{i_{\sigma_{1}(1)}} \otimes \dots \otimes e_{i_{\sigma_{1}(m)}}) \otimes (e_{j_{\sigma_{2}(1)}} \otimes \dots \otimes e_{j_{\sigma_{2}(m)}})\otimes (e_{k_{\sigma_{3}(1)}} \otimes \dots \otimes e_{k_{\sigma_{3}(m)}}) }
\end{split}
\end{align}
It is also convenient to write this, in shorthand notation, as : 
\bea 
{}&(\cL_{\sigma_{1}} \otimes \cL_{\sigma_{2}}  \otimes \cL_{\sigma_{3}} ) \ket { e_{ i_1 j_1 k_1 } \otimes e_{ i_2 j_2 k_2 } \otimes \cdots e_{ i_m j_m k_m } } 
= \ket { e_{ i_{ \sigma_1 ( 1 ) }  j_{ \sigma_2(1)}  k_{\sigma_3 (1)}  }  \otimes \cdots  \otimes e_{ i_{ \sigma_1( m ) }
 j_{ \sigma_2 ( m) }  k_{ \sigma_3 ( m )  } }}  \cr 
 &&
\eea
We also have, in the shorthand notation, 
\bea 
&& \bar \Lambda  \ket {  e_{ijk}  } = \bar \Lambda_{ ijk} \cr 
&& \bra{ e^{ ijk } } Z = Z^{ ijk} 
\eea
Using these we can convert the  trace in tensor space to the formula in terms of indices using 
\bea      
&& \text{Tr}_{V_{1}^{\otimes m} \otimes V_{2}^{\otimes m} \otimes  V_{3}^{\otimes m}} \left[  Z^{\otimes m}
 \bar{\Lambda}^{\otimes m}  \left( \cL_{ \sigma_{1}}  \otimes \cL_{ \sigma_{2} } \otimes \cL_{ \sigma_{3}}  \right)  \right]  \cr 
&& = \sum_{ I , J , K } \bra{ e^{ i_1 j_1 k_1} \otimes \cdots \otimes e^{ i_m j_m k_m } }  Z^{ \otimes m } \bar \Lambda^{ \otimes m  } 
 (\cL_{\sigma_{1}} \otimes \cL_{\sigma_{2}}  \otimes \cL_{\sigma_{3}} ) \ket { e_{ i_1 j_1 k_1 } \otimes \cdots \otimes  e_{ i_m j_m k_m } } 
 \cr 
 && 
\eea
This translates directly into a diagrammatic form, as given on the left side of the top diagram in Figure \ref{fig:tensor_diagram_3} (see Appendix \ref{app:corr_diagram_interpret} for more details on interpreting this diagrammatic tensor GIO). 
To get the   conjugate operator we take  $\bar{\Lambda} \to \Lambda$, $Z \to \bar{Z}$ and the inverse of the permutations:
\begin{align}
\begin{split}
\label{eqn:tenmod_observable_d_3_conj}
{}&\overline{\mathcal{O}_{\tau_{1}, \tau_{2}, \tau_{3}}(\bar{\Lambda},Z)} 
\\
{}&= \mathcal{O}_{\tau_{1}^{-1}, \tau_{2}^{-1}, \tau_{3}^{-1}}(\Lambda, \bar{Z})
\\ 
{}&= \text{Tr}_{V_{1}^{\otimes m} \otimes V_{2}^{\otimes m} \otimes  V_{3}^{\otimes m}} \left[ \Lambda^{\otimes m}  \bar{Z}^{\otimes m}  \left( \cL_{ \tau_{1}^{-1} } \otimes \cL_{ \tau_{2}^{-1} }  \otimes \cL_{ \tau_{3}^{-1} }  \right) \right] \\
{}&= \sum_{I,J,K}  \bar{Z}_{i_{\tau_{1}^{-1}(1)} j_{\tau_{2}^{-1}(1)} k_{\tau_{3}^{-1}(1)}} \cdots \bar{Z}_{i_{\tau_{1}^{-1}(m)} j_{\tau_{2}^{-1}(m)} k_{\tau_{3}^{-1}(m)}}  \Lambda^{i_{1} j_1 k_1} \cdots \Lambda^{{i_m} j_{m} k_{m}}  \,,
\end{split}
\end{align} 
where the sum over $I, J, K$ implies the sum over the full set of $i\,, j$ and $k$ indices (e.g. $I = \{i_{1}, \dots i_{m} \}$) and the bar on the fields indicates a covariant field. The conjugate operator has the diagrammatic form in the top-right of Figure \ref{fig:tensor_diagram_3}. The tensor field partition function is 
\begin{equation}
\Sigma = \int [dZ][d\bar{Z}] e^{-\frac{1}{2}\sum_{i,j,k} \bar{Z}_{ijk} Z^{ijk}} \,,
\end{equation}
and
\begin{equation}
\langle  \bar{Z}_{j_{1}j_{2}j_{3}}Z^{i_{1}i_{2}i_{3}} \rangle = \delta^{i_{1}}_{j_{1}}\delta^{i_{2}}_{j_{2}}\delta^{i_{3}}_{	j_{3}}
\end{equation}
is the two-point function of field variables $Z, \bar{Z}$.
The GIO correlator with classical tensor fields is therefore
\begin{align}
\langle {}& \mathcal{O}_{\sigma_{1}, \sigma_{2}, \sigma_{3}}(\bar{\Lambda}, Z) \overline{\mathcal{O}_{\tau_{1}, \tau_{2}, \tau_{3}}(\bar{\Lambda}, Z)} \rangle \nonumber
\\
{}&= \frac{1}{\Sigma} \int [dZ][d\bar{Z}] \mathcal{O}_{\sigma_{1}, \sigma_{2}, \sigma_{3}}(\bar{\Lambda}, Z) \overline{\mathcal{O}_{\tau_{1}, \tau_{2}, \tau_{3}}(\bar{\Lambda}, Z)}  e^{-\frac{1}{2}\sum_{i,j,k} \bar{Z}_{ijk} Z^{ijk}} \label{eqn:tenmod_two_point_func}
\\
\begin{split}
{}&= \sum_{\substack{I, J, K \\ I', J', K'}} \langle     \bar{\Lambda}_{i'_{\sigma_{1}(1)} j'_{\sigma_{2}(1)} k'_{\sigma_{3}(1)}} \cdots \bar{\Lambda}_{{i'_{\sigma_{1}(m)}} j'_{\sigma_{2}(m)} k'_{\sigma_{3}(m)}}  Z^{i'_{1} j'_{1} k'_{1}} \cdots Z^{i'_{m} j'_{m} k'_{m}} 
\\
{}& \qquad \qquad \qquad  \times \bar{Z}_{i_{\tau_{1}^{-1}(1)} j_{\tau_{2}^{-1}(1)} k_{\tau_{3}^{-1}(1)}} \cdots \bar{Z}_{i_{\tau_{1}^{-1}(m)} j_{\tau_{2}^{-1}(m)} k_{\tau_{3}^{-1}(m)}} \Lambda^{i_{1} j_1 k_1}  \cdots \Lambda^{{i_m} j_{m} k_{m}}\rangle
\end{split}
\\
\begin{split}
{}&= \sum_{\substack{I, J, K \\ I', J', K'}}  \sum_{\gamma \in S_{m}} \bar{\Lambda}_{i'_{\sigma_{1}(1)} j'_{\sigma_{2}(1)} k'_{\sigma_{3}(1)}} \cdots \bar{\Lambda}_{{i'_{\sigma_{1}(m)}} j'_{\sigma_{2}(m)} k'_{\sigma_{3}(m)}} \Lambda^{i_{1} j_1 k_1}  \cdots \Lambda^{{i_m} j_{m} k_{m}} 
\\
&{} \qquad \qquad \qquad \qquad \qquad \qquad \times \delta^{i'_{1}}_{i_{\gamma \tau_{1}^{-1}(1)}} \delta^{j'_{1}}_{j_{\gamma \tau_{2}^{-1}(1)}} \delta^{k'_{1}}_{k_{\gamma \tau_{3}^{-1}(1)}} \dots 
\delta^{i'_{m}}_{i_{\gamma \tau_{1}^{-1}(m)}} \delta^{j'_{m}}_{j_{\gamma \tau_{2}^{-1}(m)}} \delta^{k'_{m}}_{k_{\gamma \tau_{3}^{-1}(m)}} 
\end{split}
\\
\begin{split}
{}&= \sum_{ \substack{I, J, K \\ \gamma \in S_{m}}} \bar{\Lambda}_{i_{\sigma_{1}\gamma \tau_{1}^{-1}(1)} j_{\sigma_{2}\gamma \tau_{2}^{-1}(1)} k_{\sigma_{3}\gamma\tau_{3}^{-1}(1)}} \cdots \bar{\Lambda}_{{i_{\sigma_{1}\gamma \tau_{1}^{-1}(m)}} j_{\sigma_{2}\gamma \tau_{2}^{-1}(m)} k_{\sigma_{3}\gamma \tau_{3}^{-1}(m)}} 
\\
{}& \qquad \qquad \qquad  \qquad \qquad \qquad  \qquad \qquad \qquad \qquad  \qquad \qquad \qquad  \times \Lambda^{{i_{1}} j_{1} k_{1}} \cdots \Lambda^{{i_m} j_{m} k_{m}} \,.
\end{split} \label{eqn:tenmod_two_point_func_lambdas}
\end{align}
where the expectation value on field variables $Z$ and $\bar{Z}$ has been expanded into a sum over permutations $\gamma$ via Wick's theorem, and Kronecker equivariance was used to arrive at \eqref{eqn:tenmod_two_point_func_lambdas}. Therefore we have in trace notation
\begin{multline}
\langle \mathcal{O}_{\sigma_{1}, \sigma_{2}, \sigma_{3}}(\bar{\Lambda}, Z) \overline{\mathcal{O}_{\tau_{1}, \tau_{2}, \tau_{3}}(\bar{\Lambda}, Z)} \rangle= 
\\ 
\sum_{\gamma \in S_{m}} \text{Tr}_{V_{1}^{\otimes m} \otimes V_{2}^{\otimes m} \otimes  V_{3}^{\otimes m}} \left( \Lambda^{\otimes m} \bar{\Lambda}^{\otimes m}   (\cL_{ \sigma_{1} }  \otimes \cL_{ \sigma_{2} }  \otimes \cL_{ \sigma_{3} } ) (\cL_{ \gamma})^{\otimes 3} (\cL_{ \tau_{1}^{-1} } \otimes \cL_{ \tau_{2}^{-1}} \otimes \cL_{\tau_{3}^{-1}})\right)\,.
\end{multline}
A diagram representing this tensor, two-point GIO correlator and its equivalent description in terms of the $\gamma$ permutations using Wick's theorem, is given in Figure \ref{fig:tensor_diagram_3}.
\begin{figure}[htb!]
\begin{center}
\includegraphics[width=11cm, height=14cm]{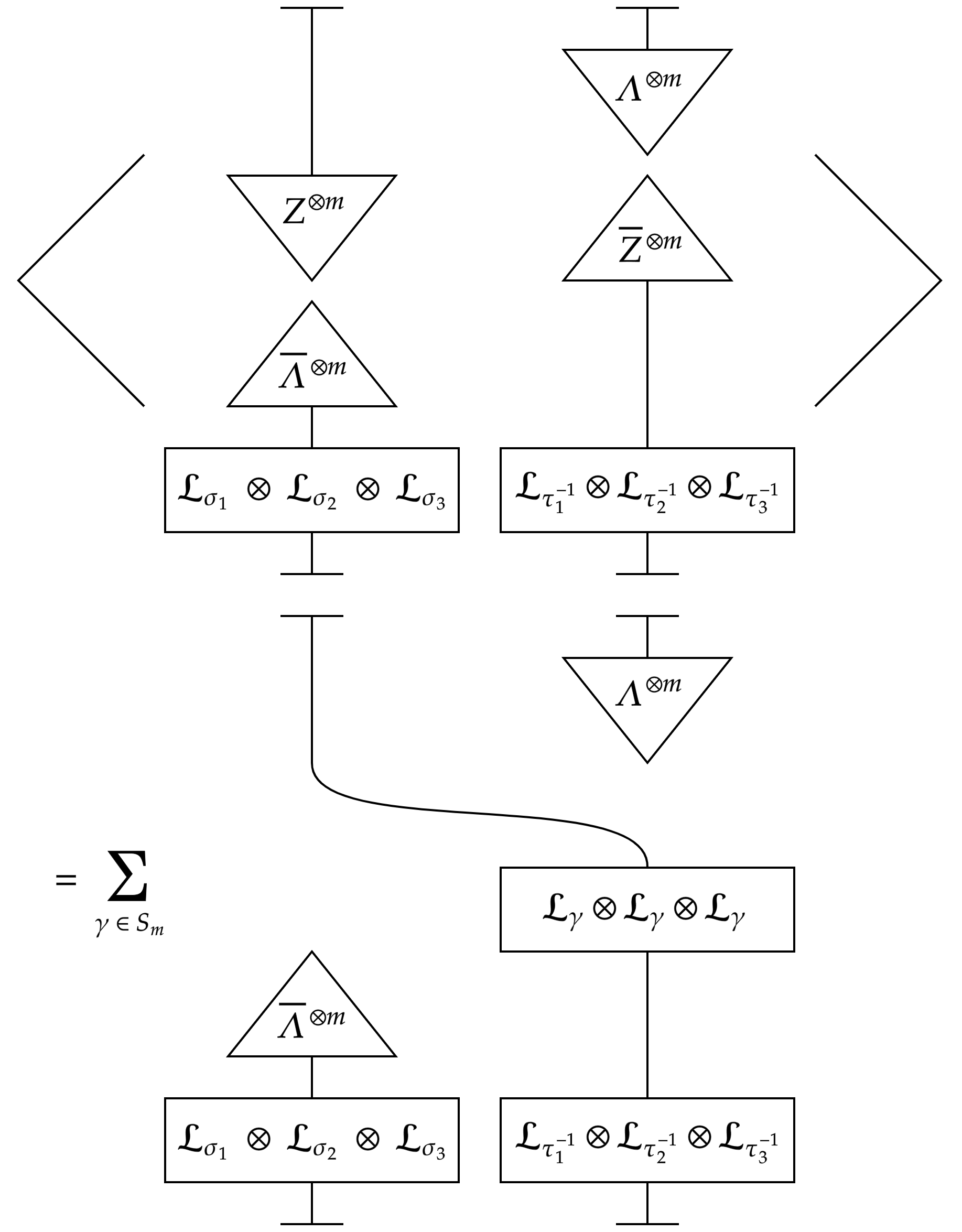}
\end{center}
\caption{The diagrammatic representation of the tensor model two-point function of gauge invariant operators with classical tensor witnesses. This example takes $\Lambda$ and $Z$ as contravariant tensors and their conjugates as covariant tensors, where all fields are rank-3 and displayed using triangles. The figure implies (via their associated linear operators) that each of the three indices on every term in the tensor product, is permuted by $\sigma_{1},\sigma_{2}$ and $ \sigma_{3}$ (or $\tau_{1}^{-1}, \tau_{2}^{-1}$ and $\tau_{3}^{-1}$) respectively. Once again, applying Wick's theorem introduces a permutation $\gamma$ that is summed over, accompanied by a twist in the indices.}
\label{fig:tensor_diagram_3}
\end{figure}
 Next we simplify \eqref{eqn:tenmod_two_point_func_lambdas} with the introduction of delta functions and elements $\alpha_{l}$ for $l = 1,2,3$,
\begin{align}
\langle \mathcal{O}_{\sigma_{1}, \sigma_{2}, \sigma_{3}}(\bar{\Lambda}, Z) {}& \overline{\mathcal{O}_{\tau_{1}, \tau_{2}, \tau_{3}}(\bar{\Lambda}, Z)} \rangle  \nonumber
\\
\begin{split}
{}&= \sum_{\substack{ I, J, K \\ \alpha_{l}, \gamma \in S_{m}}}  \bar{\Lambda}_{i_{\alpha_{1}(1)} j_{\alpha_{2}(1)} k_{\alpha_{3}(1)}} \cdots
 \bar{\Lambda}_{{i_{\alpha_{1}(m)}} j_{\alpha_{2}(m)} k_{\alpha_{3}(m)}} 
\\
{}& \qquad \qquad \qquad \qquad \qquad \qquad \times \Lambda^{i_{1} j_{1} k_{1}}  \cdots \Lambda^{i_{m} j_{m} k_{m}}  \prod_{l=1}^{3} \delta(\alpha_{l}^{-1} \sigma_{l} \gamma \tau_{l}^{-1})\,,
\end{split}
\label{eqn:tenmod_two_point_func_pt2}
\end{align}
where $\sum_{\alpha_{l}}$ implies a sum over $\alpha_{1}, \alpha_{2}$ and $\alpha_{3}$.
The sum of the classical tensor fields, $\Lambda$ and $\bar{\Lambda}$, over their index sets $I,J,K$ now takes the form of an operator as defined in equation \eqref{eqn:tenmod_observable_d_3} 
i.e.  
\begin{align}
\mathcal{O}_{\alpha_{1}, \alpha_{2}, \alpha_{3}}( \bar{\Lambda}, \Lambda) 
&= 
\sum_{I, J, K}  \bar{\Lambda}_{i_{\alpha_{1}(1)} j_{\alpha_{2}(1)} k_{\alpha_{3}(1)}} \cdots
 \bar{\Lambda}_{{i_{\alpha_{1}(m)}} j_{\alpha_{2}(m)} k_{\alpha_{3}(m)}} \Lambda^{i_{1} j_{1} k_{1}}  \cdots \Lambda^{i_{m} j_{m} k_{m}}  
\\
&= \text{Tr}_{V_{1}^{\otimes m} \otimes V_{2}^{\otimes m} \otimes  V_{3}^{\otimes m}} \left(\Lambda^{\otimes m}  \bar{\Lambda}^{\otimes m} (\mathcal{L}_{\alpha_{1}} \otimes \mathcal{L}_{\alpha_{2}} \otimes \mathcal{L}_{\alpha_{3}})   \right)
\end{align}
such that
\begin{equation}
\label{eqn:tenmod_two_point_func_pt3}
\langle \mathcal{O}_{\sigma_{1}, \sigma_{2}, \sigma_{3}}(\bar{\Lambda}, Z) \overline{\mathcal{O}_{\tau_{1}, \tau_{2}, \tau_{3}}(\bar{\Lambda}, Z)} \rangle    \\
= \sum_{ \alpha_{l}, \gamma \in S_{m}} \mathcal{O}_{\alpha_{1}, \alpha_{2}, \alpha_{3}}(\bar{\Lambda},\Lambda)
\prod_{l=1}^{3} \delta(\alpha_{l}^{-1} \sigma_{l} \gamma \tau_{l}^{-1} ) \,.
\end{equation}
The expression is further simplified by introducing $\gamma = \gamma_{1} \gamma_{2}$ and, by using the invariance of the GIO in classical tensor fields, $\mathcal{O}_{\alpha_{1},\alpha_{2},\alpha_{3}} = \mathcal{O}_{\mu_{2}\alpha_{1}\mu_{1},\mu_{2}\alpha_{2}\mu_{1},\mu_{2}\alpha_{3}\mu_{1}}$, permutations $\mu_{2}$ and $\mu_{1}$ can be included in the correlator (see Appendix \ref{app:delta_function_sums} for more details)
\begin{align}
\begin{split}
\label{eqn:tenmod_two_point_func_pt4}
\langle \mathcal{O}_{\sigma_{1}, \sigma_{2}, \sigma_{3}}(\bar{\Lambda}, Z)& \overline{\mathcal{O}_{\tau_{1}, \tau_{2}, \tau_{3}}(\bar{\Lambda}, Z)} \rangle  \\
&= \frac{1}{m!} \sum_{ \alpha_{l}, \gamma_{1}, \gamma_{2} \in S_{m}} \mathcal{O}_{\alpha_{1}, \alpha_{2}, \alpha_{3}}(\bar{\Lambda},\Lambda)
\prod_{l=1}^{3} \delta( \alpha_{l}^{-1}  \sigma_{l} \gamma_{1} \gamma_{2} \tau_{l}^{-1} ) 
\\
& = \frac{1}{(m!)^3} \sum_{ \alpha_{l}, \gamma_{1}, \gamma_{2}, \mu_{1},\mu_{2} \in S_{m}} \mathcal{O}_{\alpha_{1}, \alpha_{2}, \alpha_{3}}(\bar{\Lambda},\Lambda)
\prod_{l=1}^{3} \delta(\mu_{2} \alpha_{l}^{-1} \mu_{1} \sigma_{l} \gamma_{1} \gamma_{2} \tau_{l}^{-1} ) \,.
\end{split}
\end{align}
Using the cyclicity of permutations in the delta functions, in explicit tensor product form the correlator is
\begin{multline}
\label{eqn:tenmod_two_point_func_pt7}
\langle \mathcal{O}_{\sigma_{1}, \sigma_{2}, \sigma_{3}}(\bar{\Lambda}, Z) \overline{\mathcal{O}_{\tau_{1}, \tau_{2}, \tau_{3}}(\bar{\Lambda}, Z)} \rangle = \frac{1}{(m!)^3} \sum_{ \alpha_{l}, \gamma_{1}, \gamma_{2}, \mu_{1},\mu_{2} \in S_{m}}  \mathcal{O}_{\alpha_{1}, \alpha_{2}, \alpha_{3}}(\bar{\Lambda},\Lambda) 
\\
\times \bm{\delta}_{3} \left( (\alpha_{1}^{-1} \otimes \alpha_{2}^{-1} \otimes \alpha_{3}^{-1}) (\mu_{1}^{\otimes 3}) (\sigma_{1} \otimes \sigma_{2} \otimes \sigma_{3}) (\gamma_{1}^{\otimes 3})(\gamma_{2}^{\otimes 3}) (\tau_{1}^{-1} \otimes \tau_{2}^{-1} \otimes \tau_{3}^{-1})(\mu_{2}^{\otimes 3}) \right) \,,
\end{multline}
where
\begin{equation}
\bm{\delta}_{3}(\beta_{1} \otimes \beta_{2} \otimes \beta_{3}) \equiv \prod_{l}^{3}\delta(\beta_{l})\, \quad \text{for} \quad \beta_{l} \in S_{m} \,.
\end{equation}
The two-point function now takes the form of an operator strictly in the classical tensor fields multiplying a delta function of tensored permutations of $S_{m}$. This can be further evaluated  using the elements of the Kronecker PCA, $\mathcal{K}_{\text{un}}(m)$. There are two permutation triples with left-right $S_{m} \times S_{m}$ diagonal permutations acting on them in the delta function of \eqref{eqn:tenmod_two_point_func_pt7}. Bringing the sums over $\mu_{i}, \gamma_{i}$ inside, the correlator is equivalently rewritten as
\begin{multline}
\label{eqn:tenmod_two_point_func_alg_rib_basis}
\langle \mathcal{O}_{\sigma_{1}, \sigma_{2}, \sigma_{3}}(\bar{\Lambda}, Z) \overline{\mathcal{O}_{\tau_{1}, \tau_{2}, \tau_{3}}(\bar{\Lambda}, Z)} \rangle  = \frac{1}{(m!)^{3}} \sum_{ \alpha_{l} \in S_{m}}  \mathcal{O}_{\alpha_{1}, \alpha_{2}, \alpha_{3}}(\bar{\Lambda},\Lambda) 
\\
\times \bm{\delta}_{3} \Bigg( (\alpha_{1}^{-1} \otimes \alpha_{2}^{-1} \otimes \alpha_{3}^{-1}) 
\left[ \sum_{\mu_{1}, \gamma_{1}} \mu_{1} \sigma_{1} \gamma_{1} \otimes \mu_{1} \sigma_{2} \gamma_{1} \otimes\mu_{1} \sigma_{3} \gamma_{1}  \right] 
\\ 
\times \left[ \sum_{\mu_{2}, \gamma_{2}} \gamma_{2} \tau_{1}^{-1} \mu_{2} \otimes  \gamma_{2}  \tau_{2}^{-1} \mu_{2} \otimes \gamma_{2} \tau_{3}^{-1} \mu_{2} \right] \Bigg) \,.
\end{multline}
These square bracket terms are just the graph basis elements of \eqref{eqn:tc_tens_e_basis_1} up to the normalisation factors and choice of graph label. Since the GIO are functions of equivalence class, we can equally write an operator defined by permutations $(\sigma_{1}, \sigma_{2}, \sigma_{3})$, as an operator defined by $(\sigma^{(r)}_{1}, \sigma^{(r)}_{2}, \sigma^{(r)}_{3})$ instead. This $(r)$ notation implies any triple of permutations from the orbit $\text{Orb}_{\mathcal{K}}(r)$, i.e. from equivalence class $r$. Hence, associating the $\sigma_{i}$ and $\tau_{i}$ triples with graph labels $r$ and $s$ respectively, the correlator is now defined in reference to this basis as
\begin{align}
\label{eqn:tenmod_two_point_func_alg_rib_basis_2}
\langle \mathcal{O}_{E_{r}}(\bar{\Lambda}, Z) \overline{\mathcal{O}_{E_{s}}(\bar{\Lambda}, Z)} \rangle 
&=\langle \mathcal{O}_{\sigma_{1}, \sigma_{2}, \sigma_{3}}(\bar{\Lambda}, Z) \overline{\mathcal{O}_{\tau_{1}, \tau_{2}, \tau_{3}}(\bar{\Lambda}, Z)} \rangle 
\\
& \equiv \langle \mathcal{O}_{\sigma_{1}^{(r)}, \sigma_{2}^{(r)}, \sigma_{3}^{(r)}}(\bar{\Lambda}, Z) \overline{\mathcal{O}_{\tau_{1}^{(s)}, \tau_{2}^{(s)}, \tau_{3}^{(s)}}(\bar{\Lambda}, Z)} \rangle 
\\
&= m! \sum_{ \alpha_{l} \in S_{m}}  \mathcal{O}_{\alpha_{1}, \alpha_{2}, \alpha_{3}}( \bar{\Lambda},\Lambda) \bm{\delta}_{3} \left( (\alpha_{1}^{-1} \otimes \alpha_{2}^{-1} \otimes \alpha_{3}^{-1}) E_{r} E_{s'} \right) \,.
\end{align}
The prime notation on $E_{s'}$, is used to indicate that this basis algebra element was formed from an triple of inverse permutations, namely $(\tau_{1}^{-1}, \tau_{2}^{-1}, \tau_{3}^{-1})$. Decomposing the $\alpha_{l} \in S_{m}$ sums into sums over equivalence classes (labelled by $p$) and sums over elements in each class (labelled by $d$)\footnote{In graph nomenclature, this is decomposing the sum into a sum over all unique graphs and a sum over the elements in the same orbit of each graph.}, means the $\alpha^{-1}$ triple can also be exchanged for a Kronecker PCA basis element
\begin{align}
\MoveEqLeft[1] \langle \mathcal{O}_{E_{r}}(\bar{\Lambda}, Z) \overline{\mathcal{O}_{E_{s}}(\bar{\Lambda}, Z)} \rangle \nonumber
\\
\begin{split}
={}& m! \sum_{p} \Biggl[\sum_{d \in \text{Orb}_{\mathcal{K}}(p)} \mathcal{O}_{\alpha_{1}^{(p)}, \alpha_{2}^{(p)}, \alpha_{3}^{(p)}} (\bar{\Lambda},\Lambda)  
\\
&{} \qquad \qquad \qquad \qquad \times  \bm{\delta}_{3} \left( \left( (\alpha_{1}^{-1})^{(p')}(d) \otimes (\alpha_{2}^{-1})^{(p')}(d) \otimes (\alpha_{3}^{-1})^{(p')}(d) \right) E_{r}E_{s'}\right) \Biggr] \label{eqn:tenmod_two_point_func_alg_rib_basis_3}
\end{split}
\\
\begin{split}
={}& m! \sum_{p} \Biggl( \mathcal{O}_{\alpha_{1}^{(p)}, \alpha_{2}^{(p)}, \alpha_{3}^{(p)}} (\bar{\Lambda},\Lambda)
\\ 
& \qquad \qquad \qquad \times  \bm{\delta}_{3} \Biggl[ \sum_{d \in \text{Orb}_{\mathcal{K}}(p)} \left( (\alpha_{1}^{-1})^{(p')}(d) \otimes (\alpha_{2}^{-1})^{(p')}(d) \otimes (\alpha_{3}^{-1})^{(p')}(d) \right) E_{r}E_{s'}\Biggr] \Biggr)\,, \label{eqn:tenmod_op_in_witness_moved_outside_sum}
\end{split}
\end{align}
where the orbits/equivalence classes in this sum decomposition are labelled by $p$, and again, the prime notation has been adopted for the inverse triple $(\alpha^{-1}_{1},\alpha^{-1}_{2},\alpha^{-1}_{3})$\footnote{We may think of the sum over $d \in \text{Orb}_{\mathcal{K}}(p)$ as running over all permutation triples in a given orbit $p$. The implied effect this has on triples $(\alpha_{1}^{-1}(d),\alpha_{2}^{-1}(d),\alpha_{3}^{-1}(d))$, is to run over all triples in the $p'$ orbit, which is the orbit generated by inverse permutations: $\alpha_{i}^{-1}$.}. Hence,
\begin{equation}
\label{eqn:tenmod_tensor_delta}
\langle \mathcal{O}_{E_{r}}(\bar{\Lambda}, Z) \overline{\mathcal{O}_{E_{s}}(\bar{\Lambda}, Z)} \rangle = \sum_{p} m!|\text{Orb}_{\mathcal{K}}(p')|  \mathcal{O}_{\alpha_{1}^{(p)}, \alpha_{2}^{(p)}, \alpha_{3}^{(p)}} (\bar{\Lambda},\Lambda) \bm{\delta}_{3} \left(E_{p'}E_{r}E_{s'}\right)\,,
\end{equation}
Note that since the GIO are invariant under the equivalence relation \eqref{eqn:tc_obs_invariance}, the operator $\mathcal{O}_{\alpha_{1}^{(p)}, \alpha_{2}^{(p)}, \alpha_{3}^{(p)}}$ is the same for all $d \in \text{Orb}_{\mathcal{K}}(p)$, and so the $d$-label on the operator is omitted. For this same reason, it is free to move outside the sum over $d$ as seen in \eqref{eqn:tenmod_op_in_witness_moved_outside_sum}. From \S 2.3 of \cite{BenGeloun:2020yau}, this delta function of graph basis elements can also be simplified.  The following multiplication rule holds 
\begin{equation}
\label{eqn:tenmod_alg_mult_delta_identity}
E_{r}E_{s'} = \sum_{t = 1}^{|\text{Col}(m)|} C^{t;\mathcal{K}}_{rs'} E_{t} \quad \implies \quad \bm{\delta}_{3}(E_{p'} E_{r} E_{s'}) = \sum_{t=1}^{|\text{Col}(m)|} C^{t;\mathcal{K}}_{rs'} \bm{\delta}_{3}(E_{p'} E_{t})\,.
\end{equation}
where $C^{t;\mathcal{K}}_{rs'}$ is the $\mathcal{K}_{\text{un}}(m)$ structure constant in the graph basis. Setting $\sigma^{(p)}(a) = \sigma_{1}^{(p)}(a) \otimes \sigma_{2}^{(p)}(a) \otimes \sigma_{3}^{(p)}(a)$ for convenience, the inner product $g$ on $\mathbb{C}(S_{m})^{\otimes 3}$ can be defined as
\begin{equation}
\label{eqn:tc_inner_product}
g(\alpha^{(p)}(a),\beta^{(t)}(b)) = \bm{\delta}_{3}\left(\left(\alpha^{-1}\right)^{(p)}(a) \beta^{(t)}(b)  \right) \,,
\end{equation}
and one can identify that 
\begin{align}
\label{eqn:tc_delta_of_two_e_bases}
\bm{\delta}_{3}(E_{p'} E_{t}) 
&= \frac{1}{|\text{Orb}_{\mathcal{K}}(p')||\text{Orb}_{\mathcal{K}}(t)|}\sum_{a \in \text{Orb}_{\mathcal{K}}(p')} \sum_{b \in \text{Orb}_{\mathcal{K}}(t)} \bm{\delta}_{3}\left((\alpha^{-1})^{(p')}(a)\beta^{(t)}(b) \right) 
\\
&= \frac{1}{|\text{Orb}_{\mathcal{K}}(p')||\text{Orb}_{\mathcal{K}}(t)|}\sum_{a \in \text{Orb}_{\mathcal{K}}(p')} \sum_{b \in \text{Orb}_{\mathcal{K}}(t)} g\left(\alpha^{(p')}(a),\beta^{(t)}(b) \right) 
\\
&= \frac{1}{|\text{Orb}_{\mathcal{K}}(p')||\text{Orb}_{\mathcal{K}}(t)|}\sum_{a \in \text{Orb}_{\mathcal{K}}(p')} \sum_{b \in \text{Orb}_{\mathcal{K}}(t)} \delta_{p't}\delta_{ab}
\\
&= \frac{1}{|\text{Orb}_{\mathcal{K}}(p')|}\delta_{p't} \label{eqn:tc_delta_of_two_e_bases_final}\,.
\end{align}
Combining equations \eqref{eqn:tc_delta_of_two_e_bases_final} and \eqref{eqn:tenmod_alg_mult_delta_identity} then substituting this into equation \eqref{eqn:tenmod_tensor_delta}, the final expression for the correlator in the graph basis becomes
\begin{align}
\langle \mathcal{O}_{E_{r}}(\bar{\Lambda}, Z) \overline{\mathcal{O}_{E_{s}}(\bar{\Lambda}, Z)} \rangle 
&= \sum_{p}\sum_{t=1}^{|\text{Col}(m)|} m!|\text{Orb}_{\mathcal{K}}(p')|  \mathcal{O}_{\alpha_{1}^{(p)}, \alpha_{2}^{(p)}, \alpha_{3}^{(p)}} (\bar{\Lambda}, \Lambda)  C^{t;\mathcal{K}}_{rs'} \frac{1}{|\text{Orb}_{\mathcal{K}}(p')|}\delta_{p't}   
\\
&=  \sum_{p} m! C^{p';\mathcal{K}}_{rs'}  \mathcal{O}_{\alpha_{1}^{(p)}, \alpha_{2}^{(p)}, \alpha_{3}^{(p)}} (\bar{\Lambda},\Lambda	) 
\\
&= m! \sum_{p}  C^{p';\mathcal{K}}_{rs'}  \mathcal{O}_{E_{p}} (\bar{\Lambda},\Lambda	) \,,
\end{align}
where the final line identifies the $\alpha_{i}^{(p)}$ triple with algebra element $E_{p}$ (see equation \eqref{eqn:tenmod_two_point_func_alg_rib_basis_2} for similar application). This graph basis result 
\begin{equation}
\label{eqn:tc_boxed_comb_basis_corr}
\boxed{
\langle \mathcal{O}_{E_{r}}(\bar{\Lambda}, Z) \overline{\mathcal{O}_{E_{s}}(\bar{\Lambda}, Z)} \rangle = m! \sum_{p}  C^{p' ;\mathcal{K}}_{rs'}  \mathcal{O}_{E_{p}} (\bar{\Lambda},\Lambda	)
}
\end{equation}
is again in keeping with the previous cases: a linear combination of operators composed of the classical (tensor) fields with multiplicative structure constant factors from the corresponding PCA. Hence, given a set of basis labels $(E_{r}, E_{s}, E_{p})$, these structure constants can be computed. Note that for the graph basis of $\mathcal{K}_{\text{un}}(m)$ used, the structure constant appearing in the final expression, $C^{p' ;\mathcal{K}}_{rs'}$, is an integer as shown in \cite{BenGeloun:2020yau}.


\subsection{Fourier/\texorpdfstring{$Q$}{Q}  basis}
\label{ss:tensor_fourier_basis}
This section utilises the Fourier basis to write the correlator in terms of representation theoretic quantities. The Fourier basis we use is referred to as the ``$Q_{\text{un}}$-basis of $\mathcal{K}_{\text{un}}(m)$" in \cite{BenGeloun:2017vwn} where again, ``un" stands for ``ungauged". These $Q_{\text{un}}$ matrix basis elements, transformed from the combinatorial basis in \eqref{eqn:tc_kn_algebra_tensor}, are given by\footnote{Explicit transformation between combinatorial/permutation and representation/Fourier basis (gauged form) is given by equation (39) of \cite{BenGeloun:2017vwn}. It can also be derived by starting with the basis of algebra $\mathbb{C}(S_{m})$, then forming a tensor product, see Appendix B1 of \cite{BenGeloun:2017vwn}. Finally, direct transformation between graph and $Q$-basis is given in Appendix B of \cite{BenGeloun:2020yau}.} 
\begin{equation}
\label{eqn:tc_Qun_basis_element_tensors}
Q^{R,S,T}_{\text{un}; \lambda, \rho} = \kappa_{R,S,T} \sum_{\sigma_{l} \in S_{m}} \sum_{i_{l}, j_{l}} C^{R,S;T, \lambda}_{i_{1}, i_{2}, i_{3}} C^{R,S;T, \rho}_{j_{1}, j_{2}, j_{3}} D_{i_{1} j_{1}}^{R}(\sigma_{1}) D_{i_{2} j_{2}}^{S}(\sigma_{2}) D_{i_{3} j_{3}}^{T}(\sigma_{3}) \sigma_{1} \otimes \sigma_{2} \otimes \sigma_{3}
\end{equation}
where the sums over $\sigma_{l}, i_{l},j_{l}$ are taken to mean summing over all $l$ values: $l = 1$, $2$, $3$. Here $\kappa_{R,S,T} = \frac{d(R)d(S)d(T)}{(m!)^3}$, where $d(A) $ is the dimension of representation $A$ of $S_{m}$. In other words, the matrices $D^{A}(\sigma)$ of representation $A$ of permutation element $\sigma \in S_{m}$, are $d(A) \times d(A)$ in size. The $ C^{R,S;T, \lambda}_{i_{1}, i_{2}, i_{3}}$ terms are Clebsch-Gordon coefficients, which arise from the decomposition of a tensor product representation into the space of irreps tensored by the multiplicity space of the irreps. Indices $i_{1}, j_{1}$ run from 1 to $d(R)$, $i_{2}, j_{2}$ from 1 to $d(S)$ and $i_{3}, j_{3}$ from 1 to $d(T)$, while $\lambda, \rho \in [1, \mathsf{C}(R,S,T)]$ where $\mathsf{C}(R,S,T)$ is the Kronecker coefficient: the multiplicity of the irrep $T$ in
the tensor product of the irreps $R$ and $S$. Such a decomposition may be  written as 
\begin{equation}
\label{eqn:tc_tensor_prod_decomp_R_S}
V_{R} \otimes V_{S} = \bigoplus_{T \vdash m} V_{T} \otimes V_{\lambda}
\end{equation}
with $V_{R}$, $V_{S}$ and $V_{T}$ the vector space representations for $S_{m}$, and $V_{\lambda}$ the multiplicity space. Multiplication of two $Q_{\text{un}}$-basis elements gives
\begin{equation}
\label{eqn:tc_Qun_multiplication}
Q^{R,S,T}_{\text{un}; \lambda_{1}, \rho_{1}}  Q^{R',S',T'}_{\text{un}; \lambda_{2}, \rho_{2}} = \delta_{RR'}\delta_{SS'} \delta_{TT'}\delta_{\rho_{1} \lambda_{2}} Q^{R,S,T}_{\text{un}; \lambda_{1},\rho_{2}} \,,
\end{equation}
and the stability of the $Q_{\text{un}}$-basis elements under left and right actions of the diagonal group algebra $\text{Diag} [\mathbb{C}(S_{n})] $ elements implies\footnote{See appendix B of \cite{BenGeloun:2017vwn} for proofs.}
\begin{equation}
\label{eqn:tc_Qun_right_left_equiv}
(\gamma_{1}^{\otimes 3})Q^{R,S,T}_{\text{un}; \lambda, \rho} (\gamma_{2}^{\otimes 3}) = Q^{R,S,T}_{\text{un}; \lambda, \rho} \,.
\end{equation}
As for the operators, they take Fourier form
\begin{align}
\label{eqn:tc_fourier_gio_form_for_tensors}
\mathcal{O}^{R,S,T}_{\lambda, \rho}(\bar{\Lambda}, Z) = \sum_{\sigma_{1}, \sigma_{2}, \sigma_{3}} \bm{\delta}_{3} \left(Q^{R,S,T}_{\text{un}; \lambda, \rho} (\sigma_{1}^{-1} \otimes \sigma_{2}^{-1} \otimes \sigma_{3}^{-1}) \right) \mathcal{O}_{\sigma_{1},\sigma_{2},\sigma_{3}}(\bar{\Lambda}, Z)
\end{align}
where $\mathcal{O}_{\sigma_{1} , \sigma_{2}, \sigma_{3}}(\bar{\Lambda}, Z)$ is defined in \eqref{eqn:tenmod_observable_d_3}, while the conjugate Fourier operator is
\begin{align}
\label{eqn:tc_fourier_gio_form_for_tensors_conj}
\overline{\mathcal{O}^{R',S',T'}_{\lambda, \rho}(\bar{\Lambda}, Z)} 
{}&= \sum_{\sigma_{1}, \sigma_{2}, \sigma_{3}} \bm{\delta}_{3} \left(Q^{R',S',T'}_{\text{un}; \lambda, \rho} (\sigma_{1}^{-1} \otimes \sigma_{2}^{-1} \otimes \sigma_{3}^{-1}) \right) \overline{\mathcal{O}_{\sigma_{1},\sigma_{2},\sigma_{3}}(\bar{\Lambda}, Z)}
\\
&{}= \sum_{\sigma_{1}, \sigma_{2}, \sigma_{3}} \bm{\delta}_{3} \left(Q^{R',S',T'}_{\text{un}; \rho, \lambda} (\sigma_{1}\otimes \sigma_{2} \otimes \sigma_{3}) \right) \overline{\mathcal{O}_{\sigma_{1},\sigma_{2},\sigma_{3}}(\bar{\Lambda}, Z)}\label{eqn:tc_fourier_gio_form_mult_indices_swapped}
\end{align}
with $\overline{\mathcal{O}_{\sigma_{1},\sigma_{2},\sigma_{3}}(\bar{\Lambda}, Z)}$ defined in \eqref{eqn:tenmod_observable_d_3_conj}. Using fact $D^{R}_{ij}(\sigma) = D^{R}_{ji}(\sigma^{-1})$ in equation \eqref{eqn:tc_Qun_basis_element_tensors}, the multiplicity indices $\lambda, \rho$ have swapped positions in \eqref{eqn:tc_fourier_gio_form_mult_indices_swapped}, accompanied by the inverse permutations of the delta function. The two-point function in this basis is then
\begin{align}
{}&\left\langle  \mathcal{O}^{R,S,T}_{\lambda_{1}, \rho_{1}}(\bar{\Lambda}, Z) \overline{\mathcal{O}^{R',S',T'}_{\lambda_{2}, \rho_{2}}(\bar{\Lambda}, Z)}  \right\rangle \nonumber
\\
\begin{split}
{}&= \left\langle \sum_{\sigma_{1}, \sigma_{2}, \sigma_{3}} \bm{\delta}_{3} \left(Q^{R,S,T}_{\text{un}; \lambda_{1}, \rho_{1}} (\sigma_{1}^{-1} \otimes \sigma_{2}^{-1} \otimes \sigma_{3}^{-1}) \right) \mathcal{O}_{\sigma_{1},\sigma_{2},\sigma_{3}}(\bar{\Lambda}, Z) \right.  
\\
{}& \qquad \qquad  \qquad \qquad  \left. \times \sum_{\tau_{1}, \tau_{2}, \tau_{3}} \bm{\delta}_{3} \left(Q^{R',S','T}_{\text{un};  \rho_{2}, \lambda_{2}} (\tau_{1} \otimes \tau_{2} \otimes \tau_{3}) \right) \overline{\mathcal{O}_{\tau_{1},\tau_{2},\tau_{3}}(\bar{\Lambda}, Z)} \right\rangle \label{eqn:tc_two_pt_func_tensor_bmf_fourier_basis}
\end{split}
\\
\begin{split}
{}&= \sum_{ \substack{\sigma_{1},\sigma_{2}, \sigma_{3} \\ \tau_{1}, \tau_{2}, \tau_{3}}} \bm{\delta}_{3} \left(Q^{R,S,T}_{\text{un}; \lambda_{1}, \rho_{1}} (\sigma_{1}^{-1} \otimes \sigma_{2}^{-1} \otimes \sigma_{3}^{-1}) \right) 
\\
{}& \qquad \qquad \qquad  \times \bm{\delta}_{3} \left(Q^{R',S','T}_{\text{un}; \rho_{2}, \lambda_{2} } (\tau_{1} \otimes \tau_{2} \otimes \tau_{3}) \right) \left\langle  \mathcal{O}_{\sigma_{1},\sigma_{2},\sigma_{3}}(\bar{\Lambda}, Z) \overline{\mathcal{O}_{\tau_{1},\tau_{2},\tau_{3}}(\bar{\Lambda}, Z)} \right\rangle \label{eqn:tc_two_pt_func_tensor_bmf_fourier_basis_1_5}
\end{split}
\end{align}
From here, the result of \eqref{eqn:tenmod_two_point_func_pt7}, repeated below for convenience,
\begin{multline}
\label{eqn:tenmod_two_point_func_pt7_2nd}
\langle \mathcal{O}_{\sigma_{1}, \sigma_{2}, \sigma_{3}}(\bar{\Lambda}, Z) \overline{\mathcal{O}_{\tau_{1}, \tau_{2}, \tau_{3}}(\bar{\Lambda}, Z)} \rangle = \frac{1}{(m!)^3} \sum_{ \alpha_{l}, \gamma_{1}, \gamma_{2}, \mu_{1},\mu_{2} \in S_{m}}  \mathcal{O}_{\alpha_{1}, \alpha_{2}, \alpha_{3}}(\bar{\Lambda},\Lambda)
\\
\times \bm{\delta}_{3} \left( (\alpha_{1}^{-1} \otimes \alpha_{2}^{-1} \otimes \alpha_{3}^{-1}) (\mu_{1}^{\otimes 3}) (\sigma_{1} \otimes \sigma_{2} \otimes \sigma_{3}) (\gamma_{1}^{\otimes 3})(\gamma_{2}^{\otimes 3}) (\tau_{1}^{-1} \otimes \tau_{2}^{-1} \otimes \tau_{3}^{-1})(\mu_{2}^{\otimes 3}) \right) \,,
\end{multline}
may be utilised by substituting this correlator into the  expression \eqref{eqn:tc_two_pt_func_tensor_bmf_fourier_basis_1_5}. By then summing over the $\sigma_{i}$ and $\tau_{i}$, the tensor product triples of \eqref{eqn:tenmod_two_point_func_pt7_2nd} are replaced by the $Q$-basis elements via the delta functions. As such we have
\begin{multline}
\label{eqn:tc_two_pt_func_tensor_bmf_fourier_basis_2}
\left\langle  \mathcal{O}^{R,S,T}_{\lambda_{1}, \rho_{1}}(\bar{\Lambda}, Z) \overline{\mathcal{O}^{R',S',T'}_{\lambda_{2}, \rho_{2}}(\bar{\Lambda}, Z)}  \right\rangle
= \frac{1}{(m!)^3} \sum_{ \alpha_{l}, \gamma_{1}, \gamma_{2}, \mu_{1},\mu_{2} \in S_{m}}  \mathcal{O}_{\alpha_{1}, \alpha_{2}, \alpha_{3}}(\bar{\Lambda},\Lambda)
\\
\times \bm{\delta}_{3} \left( (\alpha_{1}^{-1} \otimes \alpha_{2}^{-1} \otimes \alpha_{3}^{-1}) (\mu_{1}^{\otimes 3}) Q^{R,S,T}_{\text{un}; \lambda_{1}, \rho_{1}} (\gamma_{1}^{\otimes 3})(\gamma_{2}^{\otimes 3}) Q^{R',S',T'}_{\text{un}; \rho_{2}, \lambda_{2} } (\mu_{2}^{\otimes 3}) \right) \,.
\end{multline}
Using properties \eqref{eqn:tc_Qun_right_left_equiv} and \eqref{eqn:tc_Qun_multiplication} of the $Q_{\text{un}}$-basis,
\begin{align}
\MoveEqLeft[1] \left\langle \mathcal{O}^{R,S,T}_{\lambda_{1}, \rho_{1}}(\bar{\Lambda}, Z) \overline{\mathcal{O}^{R',S',T'}_{\lambda_{2}, \rho_{2}}(\bar{\Lambda}, Z)}  \right\rangle \nonumber 
\\
{}&= m!\sum_{\alpha_{l} \in S_{m}} \mathcal{O}_{\alpha_{1}, \alpha_{2}, \alpha_{3}}(\bar{\Lambda}, \Lambda) \bm{\delta}_{3}\left((\alpha_{1}^{-1} \otimes \alpha_{2}^{-1} \otimes \alpha_{3}^{-1})  Q^{R,S,T}_{\text{un}; \lambda_{1}, \rho_{1}} Q^{R',S',T'}_{\text{un}; \rho_{2}, \lambda_{2} } \right) \\
& =  m! \delta_{RR'} \delta_{SS'} \delta_{TT'} \delta_{\rho_{1} \rho_{2}}\sum_{\alpha_{l} \in S_{m}} \bm{\delta}_{3}\left((\alpha_{1}^{-1} \otimes \alpha_{2}^{-1} \otimes \alpha_{3}^{-1})  Q^{R,S,T}_{\text{un}; \lambda_{1}, \lambda_{2}} \right) \mathcal{O}_{\alpha_{1}, \alpha_{2}, \alpha_{3}}( \bar{\Lambda},\Lambda) 
\\
{}& =  m! \delta_{RR'} \delta_{SS'} \delta_{TT'} \delta_{\rho_{1} \rho_{2}} \mathcal{O}^{R,S,T}_{\text{un};\lambda_{1}, \lambda_{2}}(\bar{\Lambda},\Lambda)\,, &&
\end{align}
where the final equation comes from the operator definition in \eqref{eqn:tc_fourier_gio_form_for_tensors}. The correlator
\begin{equation}
\label{eqn:tc_boxed_fourier_corr_result}
\boxed{
\left\langle \mathcal{O}^{R,S,T}_{\lambda_{1}, \rho_{1}}(\bar{\Lambda}, Z) \overline{\mathcal{O}^{R',S',T'}_{\lambda_{2}, \rho_{2}}(\bar{\Lambda}, Z)}  \right\rangle = m! \delta_{RR'} \delta_{SS'} \delta_{TT'} \delta_{\rho_{1} \rho_{2}} \mathcal{O}^{R,S,T}_{\text{un};\lambda_{1}, \lambda_{2}}(\bar{\Lambda},\Lambda)
}
\end{equation}
is therefore akin to the matrix cases where the final result is also proportional to an operator of the classical witness fields in representation theoretic/Fourier basis.

\clearpage


\section{Summary  and Outlook } 

Representation theoretic orthogonal bases for the complex one-matrix, multi-matrix and tensor models have applications in the AdS/CFT correspondence and more general models of gauge-string duality. It has been understood that these orthogonal bases are related to permutation centralizer algebras and organise many aspects of the combinatorics of matrix and tensor model correlators.  In the context of $\text{CFT}_{4}$, these algebras have been shown to be related to enhanced symmetries in the free field limit \cite{Kimura:2008ac,Mattioli:2016eyp}.  However a direct physical interpretation, in terms of the observables of matrix or tensor models, of the structure constants of these algebras has so far been lacking. In this paper, we showed that the notion of matrix and tensor witness fields, defined in the introduction, allows such a direct interpretation. The Fourier transform from combinatorial bases of the PCAs to representation theoretic Wedderburn-Artin bases lead to generalisations of the orthogonality relations making contact with the super-integrability programme of Morozov and Mironov \cite{Mironov:2017och}. We outline a number of future research directions suggested by the results of this paper.

The problem of identifying the quantum state associated with the operator $\cO_R ( Z )$ in the half-BPS sector of 
$\cN =4$ SYM with $U(N)$ gauge group is related to interesting questions related to the black hole information paradox \cite{IILoss} and also connects with interesting structural properties of the centres of symmetric group algebras \cite{BPSCentres}. The Casimirs of $U(N)$, which can identify the Young diagrams, can be related to asymptotic multipole moments of the gravity fields generated by the LLM geometry \cite{Lin:2004nb} corresponding to the Young diagram $R$.  Using the results of this paper, we can measure the Young diagram labelling a matrix model observable by inspecting the outcome of the 2-point correlator as a function of a matrix coupling. It would be interesting to consider deformations of $\cN=4$ SYM involving the introduction of 4-dimensional background matrix fields. Presumably this would reduce or break the supersymmetry, but would provide an interesting higher dimensional quantum field theory arena to explore the implications of witness fields, with  potential implications for the gravity dual. 
The use of classical matrix fields alongside quantum matrix fields inside observables  has also been used in the context of  coherent state methods for matrix correlator  problems in  $\cN=4 $ SYM \cite{Berenstein:2022srd}.  In a distant corners limit the Young diagram bases of the 2-matrix sector considered in section 3 have been related to coherent state calculations, resulting in new integral formulae for symmetric group characters  \cite{Carlson:2022dot}. A better understanding of the link between coherent state methods and the results in the present paper is likely to be useful in applying  the present results to giant graviton physics.  For example  taking the matrix $A$ in section \ref{ss:back_fields_in_observables} to be another complex field say $Y$ among the $SU(3)$ triplet $\{ Z , Y , X \} $ of complex fields in SYM, the RHS of \eqref{BoxedResult1} contains Schur polynomial functions of $YY^{\dagger}$. Finding an interpretation of such Schur polynomials in the bulk space-time is an interesting problem. A non-trivial goal would be to justify such an interpretation using agreement between  independent calculations in bulk and boundary, e.g. agreements of $3$-point functions  along the lines of \cite{Bissi:2011dc,Caputa:2012yj,Lin:2012ey,Kristjansen:2015gpa,Jiang:2019xdz,Yang:2021kot,Chen:2019gsb,Holguin:2022zii}. 

In this paper we have considered modifications of Gaussian single and multi-matrix actions by adding matrix couplings to the quadratic terms. We have shown that two-point functions calculated as a function of the matrix couplings define structure constants of associative algebras. An interesting question is whether analogous correlators depending on matrix couplings in the quadratic terms lead to  associative algebra structure constants when the Gaussian actions  are perturbed with interaction terms. Commutator interaction terms are of interest in AdS/CFT (see  \cite{Cook:2007et,Filev:2013pza,Berenstein:2008eg})  and also in the context of emergence geometry in the context of the  IKKT IIB matrix model \cite{Battista:2022vvl,Ishibashi:1996xs}. Establishing associativity or characterising  the departure from associativity  in interacting matrix models are interesting projects.

The formula  \eqref{trace-basis-corr}  for the 2-point function in the one-complex matrix in the combinatorial basis has an interpretation as a  combinatorial  model of gauge-string duality \cite{Itzykson,deMelloKoch:2010hav,Brown:2010af} closely analogous to the duality between partition functions of two-dimensional Yang-Mills theory and combinatorial models of Euler characteristics of Hurwitz moduli space \cite{Gross:1993yt,Cordes:1994fc}. The combinatorial  string side is a sum over Belyi maps, which are branched covers of the sphere with three branch points.  The different world-sheet genera are summed with an $N$-dependent weight. With appropriate normalization of the operators, $1/N$ can be interpreted as the string coupling. In the presence of the matrix coupling $A$, powers of the string coupling have effectively been replaced by the invariant functions $\cO_{ p_3} ( B ) $. From \eqref{permbasCorr2} we recover $C_{ p_1 , p_2 }^{ p_3} $ by picking up 
the coefficient of $\cO_{ p_3} ( B )$. We thus have a direct matrix model interpretation of Belyi map counting with specified branching structure at the three branch points being given by the partitions $(p_1, p_2 , p_3 )$. 
This is a refinement of the Belyi map - matrix model connection known from 
\cite{Itzykson,deMelloKoch:2010hav,Brown:2010af}.

It is instructive to compare the emergence of the structure constants of a symmetry algebra in this paper  
from matrix and tensor witness fields, i.e. classical couplings in the action or classical constituents of composite classical/quantum observables,  with other ways of getting structure constants of algebras in quantum field theory. A common mechanism is to consider 3-point functions in topological or conformal field theories (see discussions, for example,  in \cite{deMelloKoch:2014aot,Bah:2022wot}). 

Algebraic techniques similar to those used in this paper have recently been applied to the case of matrix models or matrix quantum systems where the $U(N)$ invariance of matrix variables is replaced by an $S_N$ (symmetric group) invariance of the matrix variables. The hidden symmetry algebra is related to a  partition  algebra. More precisely, for the case of degree k invariants, we have an $S_k$ invariant subspace of a partition algebra $P_k (N)$, denoted $SP_k(N)$. This is developed in \cite{Barnes:2021tjp,Barnes:2022qli} . A natural future direction is to obtain the explicit structure constants of these algebras using witness fields as we have done here.

\vspace{1cm}
\centerline{\bf{Acknowledgments}}
We  would like to thank Joseph Ben Geloun for collaboration in the early stages of the project. We also thank George Barnes,  Robert de Mello Koch, Yang Lei, Adrian Padellaro, Gabriele Travaglini, David Vegh, Congkao Wen for insightful discussions during the course of the project. SR is supported by the STFC consolidated grant ST/P000754/1 ``String Theory, Gauge Theory and Duality''. LS is supported by an STFC quota studentship. 

\newpage


\appendix

\section{Deriving the two-point function with witness matrix field}
\label{app:background_matrix_field_correlator}
Here the basic two-point function of matrix variables $Z$ and $Z^{\dagger}$ is derived. The path integral with a coupling witness matrix field $A$ (that is taken to be invertible), is 
\begin{equation}
\label{eqn:amob_bmf_partition_func}
\Sigma[0] = \int [dZ] e^{ - \text{Tr}\left(Z A Z^{\dagger} \right)}
\end{equation}
where the $[0]$ in $\Sigma[0]$ represents a source-free action/integral. The action is 
\begin{align}
\begin{split}
\label{eqn:cd_action_split}
\text{Tr}\left(Z A Z^{\dagger} \right) 
&= \sum_{i,j,k} Z^{i}_{j} A^{j}_{k}  (Z^{\dagger})^{k}_{i} 
\\
&= \sum_{i,j,k,l,m,n} Z^{i}_{j} \left( \delta^{j}_{k} A^{k}_{l} \delta^{l}_{m} \delta^{n}_{i} \right) (Z^{\dagger})^{m}_{n}  
\\
&= \sum_{i,j,m,n} Z^{i}_{j} \left(\delta^{n}_{i} A^{j}_{m} \right) (Z^{\dagger})^{m}_{n}  \,.
\end{split}
\end{align}
Having separated the expression in terms of indices, to begin solving the integral, first define a $N^2$-dimensional vector with components equal to matrix $Z$'s elements as follows
\begin{align}
\label{cd_x_vector}
\vec{x} &= \begin{bmatrix}
           x^{1} \\
           x^{2} \\
           \vdots \\
           x^{N^2}
         \end{bmatrix} 
         = \begin{bmatrix}
           Z^{1}_{1} \\
           \vdots \\
           Z^{1}_{N} \\
           Z^{2}_{1} \\
           \vdots \\
           Z^{2}_{N} \\
           \vdots  \\
           Z^{N}_{1} \\
           \vdots \\
           Z^{N}_{N}
         \end{bmatrix}\,.
\end{align}
Its Hermitian conjugate is therefore
\begin{align}
\begin{split}
\label{eqn:cd_x_dagger}
\vec{x}^{\dagger} &= \begin{bmatrix}
										({x}^{1})^{*} & \cdots & \cdots  & \cdots 	&	({x}^{N^2})^{*} 
										\end{bmatrix} \\
									&= \begin{bmatrix}
   							 		 (Z^{1}_{1})^{*}  & \cdots & (Z^{1}_{N})^{*} & \cdots & \cdots & \cdots & (Z^{N}_{1})^{*} & \cdots & (Z^{N}_{N})^{*} 
 										\end{bmatrix} \\
 									&= \begin{bmatrix}
   							 		 (Z^{\dagger})^{1}_{1}  & \cdots & (Z^{\dagger})^{N}_{1} & \cdots & \cdots & \cdots & (Z^{\dagger})^{1}_{N} & \cdots & (Z^{\dagger})^{N}_{N}  
 										\end{bmatrix}\,.
\end{split}
\end{align}
In component form, each matrix element is then defined as 
\begin{equation}
\label{eqn:cd_x_comp_form}
x^{N(i-1) + j} := Z^{i}_{j}   \quad \text{and} \quad (x^{\dagger})_{N(n-1)+m}  := (Z^{*})^{n}_{m} =  (Z^{\dagger})^{m}_{n}\,.
\end{equation}
Similarly, for $\delta^{n}_{i} A^{j}_{m} $ in the last line of \eqref{eqn:cd_action_split}, introduce the $N^2 \times N^2$ dimensional matrix $R$, which is formed from this Kronecker product of $\mathbb{I}$ (the identity matrix) and $A^{T}$
\begin{equation}
\label{eqn:cd_delta_a-matrix_representation}
\delta^{n}_{i} A^{j}_{m} = \delta^{n}_{i} (A^{T})^{m}_{j}  = \left(\mathbb{I} \otimes A^{T} \right)^{n,m}_{i,j} =: R^{ N(n-1)+m}_{N(i-1) + j} \,.
\end{equation}
Note that $A^{T}$ was introduced so that $R$ and the vector variables of \eqref{eqn:cd_x_comp_form}, contract/multiply correctly. The action can thus be redefined using the above objects as
\begin{align}
\begin{split}
\label{eqn:cd_xdagger_b_x}
\vec{x}^{\dagger} R \vec{x} &= \sum_{i,j,n,m} (x^{\dagger})_{N(n-1)+m} \left(R^{ N(n-1)+m}_{N(i-1) + j} \right) x^{N(i-1) + j} \\
&=  \sum_{i,j,n,m} (Z^{\dagger})^{m}_{n} \left( \delta^{n}_{i} \left(A^{T}\right)^{m}_{j}  \right) Z^{i}_{j} \\
&=  \sum_{i,j,m} Z^{i}_{j}   A^{j}_{m} (Z^{\dagger})^{m}_{i} \\
&= \text{Tr}\left(Z A Z^{\dagger} \right)\,.
\end{split}
\end{align}

Substituting this into the integral and changing the variable of the measure to $\vec{x}$, we find it is of standard form and is readily computable
\begin{equation}
\label{eqn:cd_new_pf_with_vecs}
\Sigma[0] = \int [d\vec{x}] e^{- \vec{x}^{\dagger} R \vec{x}} = \pi^{N^2} \text{det}(R^{-1})\,.
\end{equation}
Equation \eqref{eqn:cd_new_pf_with_vecs} holds when $R$ is Hermitian and is positive definite (see Chapter 3.2 of \cite{Altland:2006si}). These conditions are adhered to throughout the discussion of matrix coupling fields in this paper. To derive the correlator, we introduce complex vector sources $(\vec{S}, \vec{S}^{\dagger})$ to the action and normalise by the sourceless partition function
\begin{equation}
\label{eqn:cd_sourced_pf}
 \Sigma[\vec{S}, \vec{S}^{\dagger}] = \frac{1}{\Sigma[0]} \int [d\vec{x}] e^{- \vec{x}^{\dagger} R \vec{x} + \vec{x}^{\dagger} \vec{S} + \vec{S}^{\dagger} \vec{x}}\,.
\end{equation}
where $\vec{x}^{\dagger}\vec{S} = (\vec{x}^{\dagger})_{N(i-1)+j} (\vec{S})^{N(i-1) +j}$ and $\vec{S}^{\dagger}\vec{x} = (\vec{S}^{\dagger})_{N(i-1)+j} (\vec{x})^{N(i-1) +j}$, with $\vec{S}$ an $N^2$-dimensional vector. One may write
\begin{align}
\begin{split}
\label{eqn:cd_complete_square}
(\vec{x} - R^{-1}\vec{S})^{\dagger} R(\vec{x} - R^{-1}\vec{S}) &= \vec{x}^{\dagger} R \vec{x} - \vec{x}^{\dagger} R R^{-1} \vec{S} - \vec{S}^{\dagger} (R^{-1})^{\dagger} R \vec{x} + \vec{S}^{\dagger} (R^{-1})^{\dagger} R R^{-1} \vec{S} \\
&= \vec{x}^{\dagger} R \vec{x} - \vec{x}^{\dagger} \vec{S} - \vec{S}^{\dagger} \vec{x} + \vec{S}^{\dagger} R^{-1} \vec{S}\,,
\end{split}
\end{align}
where in the second line $(R^{-1})^{\dagger} = R^{-1}$ was used to simplify. Therefore 
\begin{equation}
\label{eqn:cd_sub_for_comp_square}
\Sigma[\vec{S}, \vec{S}^{\dagger}] = \frac{1}{\Sigma[0]} \int [d\vec{x}] e^{-(\vec{x} - R^{-1}\vec{S})^{\dagger} R(\vec{x} - R^{-1}\vec{S}) + \vec{S}^{\dagger} R^{-1} \vec{S}} \,.
\end{equation}
Changing variables as $\vec{y} = \vec{x} - R^{-1}\vec{S}$ produces a trivial Jacobian and hence 
\begin{equation}
\label{eqn:cd_new_y_variable}
\Sigma[\vec{S}, \vec{S}^{\dagger}] = e^{\vec{S}^{\dagger} R^{-1} \vec{S}}\underbrace{\frac{1}{\Sigma[0]} \int [d\vec{y}] e^{-\vec{y}^{\dagger}R\vec{y}}}_{=1}  = e^{\vec{S}^{\dagger} R^{-1} \vec{S}}\,.
\end{equation}
Finally, to derive the correlator, we set
\begin{equation}
\label{eqn:cd_correlator_def}
\langle Z^{i}_{j} (Z^\dagger)^{k}_{l} \rangle = \langle x^{N(i-1) + j} (x^{\dagger})_{N(l-1)+k} \rangle = \langle  x^{\mu} ( x^{\dagger})_{\nu}\rangle
\end{equation}
where we have labelled $\mu = N(i-1) + j$, $\nu = N(l-1) + k$ for convenience. Taking derivatives of the partition function \eqref{eqn:cd_new_y_variable} with respect to the sources and at the end of the calculation, setting the sources to zero, we obtain 
\begin{align}
\begin{split}
\label{eqn:cd_derivs_of_pf}
\langle  x^{\mu} (x^{\dagger})_{\nu} \rangle 
&= \frac{\partial}{\partial S^{\nu}}\frac{\partial}{\partial (S^\dagger)_{\mu}}  \left( \Sigma[\vec{S}, \vec{S}^{\dagger}] \right) \biggr\rvert_{\vec{S} = \vec{S}^{\dagger} = 0}\\
& = \frac{\partial}{\partial S^{\nu}}\frac{\partial}{\partial (S^\dagger)_{\mu}} e^{\sum_{\sigma, \lambda} (S^{\dagger})_{\sigma} (R^{-1})^{\sigma}_{\lambda} S^{\lambda}} \biggr\rvert_{\vec{S} = \vec{S}^{\dagger} = 0} \\
&= \frac{\partial}{\partial S^{\nu}}\left[ \sum_{\sigma, \lambda} \delta^{\mu}_{\sigma}(R^{-1})^{\sigma}_{\lambda} S^{\lambda} \cdot \Sigma[\vec{S}, \vec{S}^{\dagger}] \right]  \Biggr\rvert_{\vec{S} = \vec{S}^{\dagger} = 0} \\
&= \left[ \sum_{\sigma, \lambda} \delta^{\mu}_{\sigma} (R^{-1})^{\sigma}_{\lambda} \delta^{\lambda}_{\nu} \cdot \Sigma[\vec{S}, \vec{S}^{\dagger}] + (\text{$\vec{S}$ and $\vec{S}^{\dagger}$ dependent terms}) \right]  \Biggr\rvert_{\vec{S} = \vec{S}^{\dagger} = 0} \\
&= (R^{-1})^{\mu}_{\nu}\,.
\end{split}
\end{align}
Therefore the basic two-point function with witness matrix field is
\begin{align}
\label{eqn:cd_final_corr_manip}
\langle Z^{i}_{j} (Z^\dagger)^{k}_{l} \rangle = \langle  x^{\mu} ( x^{\dagger})_{\nu}\rangle = (R^{-1})^{\mu}_{\nu} &= (R^{-1})^{N(i-1) + j}_{N(l-1)+k} 
\\
&= \left(\left(\mathbb{I} \otimes A^{T}\right)^{-1}\right)^{N(i-1) + j}_{N(l-1)+k} 
\\
&= \left((\mathbb{I})^{-1} \otimes (A^{T})^{-1}\right)^{N(i-1) + j}_{N(l-1)+k}
\\
&= \delta^{i}_{l} \left(\left(A^{T}\right)^{-1}\right)^{j}_{k} \label{eqn:cd_kron_prod_comps_del_a_trans_inv}
\\
&= \delta^{i}_{l} \left(A^{-1}\right)^{k}_{j}\,.
\end{align}
where equation \eqref{eqn:cd_kron_prod_comps_del_a_trans_inv} writes the $R^{-1}$ matrix components in terms of the Kronecker product components, as in \eqref{eqn:cd_delta_a-matrix_representation}. Note that the condition that $R$ is Hermitian, implies that in addition to invertibility, $A$ must also be Hermitian.
This can be seen by observing the block diagonal matrix $R$, formed from the Kronecker product between the delta function and $A^{T}$
\begin{equation}
R = (\mathbb{I} \otimes A^{T}) = \begin{pmatrix}
	\boxed{A^{T}} & 0 & \dots & 0 \\
	0 & \boxed{A^{T}} & \dots & 0 \\
	\vdots & \vdots & \ddots & \vdots \\
	0 & 0 & \dots & \boxed{A^{T}}
\end{pmatrix}
\end{equation}
Since $R = R^{\dagger}$, then by virtue that $R$ is block diagonal, where each block consists of the same matrix $A^{T}$, we have $A^{T} = (A^{T})^{\dagger}$ and consequently, $A = A^{\dagger}$.

Additionally, since Hermitian matrices are diagonalisable, it suffices that if $R$ is positive definite, i.e. its eigenvalues are all positive, and as $R$ is composed of $N$ copies of the eigenvalues of $A^{T}$ (noting that $A$ has the same eigenvalues as $A^{T}$), then $A$ must also be positive definite. Finally, integrals of the form \eqref{eqn:cd_new_pf_with_vecs}, can be extended to those with a complex symmetric matrix $R'$, providing further positivity conditions are applied to the real/Hermitian part of $R'$ (see Chapter 1.7 of \cite{zinn2004path}). 


\section{Interpreting correlator diagrams}
\label{app:corr_diagram_interpret}
This appendix describes how to interpret the correlator diagrams, focussing on the ``one-matrix correlator with coupling witness field" result of \S\ref{ss:background_matrix_field} as a main example first, then proceeds to describe a tensor GIO diagram, used to produce the correlator figure in \S\ref{ss:background_tensor_fields}.

Figure \ref{fig:two_perm_correlator_diagram_app} shows the one-matrix correlator diagrams where each box corresponds to an ``operator" which are connected to one another by ``branches". Note that there are also horizontal lines at the bottom and top of the diagrams, which imply that the ends of these branches should be identified. The operators can either be matrix variables such as $Z$, witness fields $B$, or permutation linear operators such as $\mathcal{L}_{\sigma}$ for $\sigma \in S_{n}$. As mentioned previously, the permutation operators act on basis vectors in the following way
\begin{equation}
\label{eqn:app_dia_perm_operator}
\mathcal{L}_{\sigma}\ket{e_{i_{1}} \otimes \dots \otimes e_{i_{n}}} = \ket{e_{i_{\sigma(1)}} \dots e_{i_{\sigma(n)}}} \,.
\end{equation}
where $e_{i_{1}}, \dots, e_{i_{n}}$ is a basis for $V^{\otimes n}$. 
They abide by the tensor composition property 
\begin{equation}
\label{eqn:app_dia_homo_prop}
\mathcal{L}_{\sigma}\mathcal{L}_{\tau} = \mathcal{L}_{\sigma\tau} \,,
\end{equation}
which is shown by calculating the tensor composition elements of these operators by inserting the identity
\begin{align}
\label{eqn:app_dia_identity_insert_ops}
\begin{split}
{}&\braket{e^{j_{1}} \otimes \dots \otimes e^{j_{n}} |\mathcal{L}_{\sigma} \mathcal{L}_{\tau} |e_{i_{1}} \otimes \dots \otimes e_{i_{n}} } 
\\
{}&= \braket{e^{j_{1}} \otimes \dots \otimes e^{j_{n}}|\mathcal{L}_{\sigma} |e_{k_{1}} \otimes \dots \otimes e_{k_{n}} } \braket{e^{k_{1}} \otimes \dots \otimes e^{k_{n}} |  \mathcal{L}_{\tau} |e_{i_{1}} \otimes \dots \otimes e_{i_{n}} }
\\
{}&=  \braket{e^{j_{1}} \otimes \dots \otimes e^{j_{n}}|e_{k_{\sigma(1)}} \otimes \dots \otimes e_{k_{\sigma(n)}} } \braket{e^{k_{1}} \otimes \dots \otimes e^{k_{n}}|e_{i_{\tau(1)}} \otimes \dots \otimes e_{i_{\tau(n)}} }
\\
{}&= \delta^{j_{1}}_{k_{\sigma(1)}} \dots \delta^{j_{n}}_{k_{\sigma(n)}} \delta^{k_{1}}_{i_{\tau(1)}} \dots \delta^{k_{n}}_{i_{\tau(n)}}
\\
{}&= \delta^{j_{\sigma^{-1}(1)}}_{k_{1}} \dots \delta^{j_{\sigma^{-1}(n)}}_{k_{n}} \delta^{k_{1}}_{i_{\tau(1)}} \dots \delta^{k_{n}}_{i_{\tau(n)}}
\\
{}&=\delta^{j_{1}}_{i_{\sigma \tau(1)}} \dots \delta^{j_{n}}_{i_{\sigma \tau(n)}} \,,
\end{split}
\end{align}
where the last two steps make use of Kronecker equivariance and the repeated indices are summed. By virtue of the operator action \eqref{eqn:app_dia_perm_operator}, it may also be written that
\begin{align}
\label{eqn:app_dia_perm_comp_ops}
\begin{split}
\braket{e^{j_{1}} \otimes \dots \otimes e^{j_{n}} |\mathcal{L}_{\sigma\tau} |e_{i_{1}} \otimes \dots \otimes e_{i_{n}} }  
&= \braket{e^{j_{1}} \otimes \dots \otimes e^{j_{n}}|e_{i_{\sigma \tau(1)}} \otimes \dots \otimes e_{i_{\sigma \tau(n)}} }
\\
&= \delta^{j_{1}}_{i_{\sigma \tau(1)}} \dots \delta^{j_{n}}_{i_{\sigma \tau(n)}} 
\end{split}
\end{align}
Since results \eqref{eqn:app_dia_identity_insert_ops} and \eqref{eqn:app_dia_perm_comp_ops} are equal, we conclude that $\cL_{\sigma}\cL_{\tau} = \cL_{\sigma \tau}$. The composition of these linear operators can therefore be thought of as first substituting $e_{i_{a}}$ for $e_{i_{\tau(a)}}$ from the action of $\cL_{\tau}$ ($\tau $ thus acting on the positions of basis vectors labelled by $a \in \{ 1, 2, , \cdots, n \}$ , then further substituting $e_{i_{\tau(a)}}$ for $e_{i_{\tau(\sigma(a))}} = e_{ i_{ \sigma \tau (a) } } $ from the action of $\cL_{\sigma}$. Having addressed the operator properties, to obtain the correlator result by reading Figure \ref{fig:two_perm_correlator_diagram_app}, start at the top $J$ index of the right hand side diagram. Follow this branch downward to encounter the $B^{\otimes n}$, $\mathcal{L}_{\gamma^{-1}}$ and $\mathcal{L}_{\sigma_{1}}$ operators. This connection of the three operators corresponds to multiplication/contraction of indices as
\begin{equation}
\sum_{K,L} B^{j_{1}}_{k_{1}} \dots B^{j_{n}}_{k_{n}} (\mathcal{L}_{\gamma^{-1}})^{k_{1}, \dots, k_{n}}_{l_{1}, \dots, l_{n}} (\mathcal{L}_{\sigma_{1}})^{l_{1}, \dots, l_{n}}_{i_{1}, \dots, i_{n}} = \sum_{K} B^{j_{1}}_{k_{1}} \dots B^{j_{n}}_{k_{n}} (\mathcal{L}_{\gamma^{-1}\sigma_{1}})^{k_{1}, \dots, k_{n}}_{i_{1}, \dots, i_{n}} \,.
\end{equation}
where $\sum_{K,L}$ indicates a sum over all $k,l$ indices. Note that the upper indices on the $B$ witness matrices are labelled with $j$ and the lower indices on operator $\mathcal{L}_{\sigma_{1}}$ are labelled with $i$, in keeping with the Figure \ref{fig:two_perm_correlator_diagram_app} labels. Continuing on, follow the bottom $I$ index under $\mathcal{L}_{\sigma_{1}}$ to the top of the diagram since we identify the horizontal lines, and progress down this branch, meeting $\mathcal{L}_{\gamma}$ then $\mathcal{L}_{\sigma_{2}^{-1}}$ and finally connecting back to the $J$ index at the top. Having formed a connected loop, all indices are contracted and we may represent the result as a trace over these operators. 
\begin{figure}[htb!]
\begin{center}
\includegraphics[width=13cm, height=7cm]{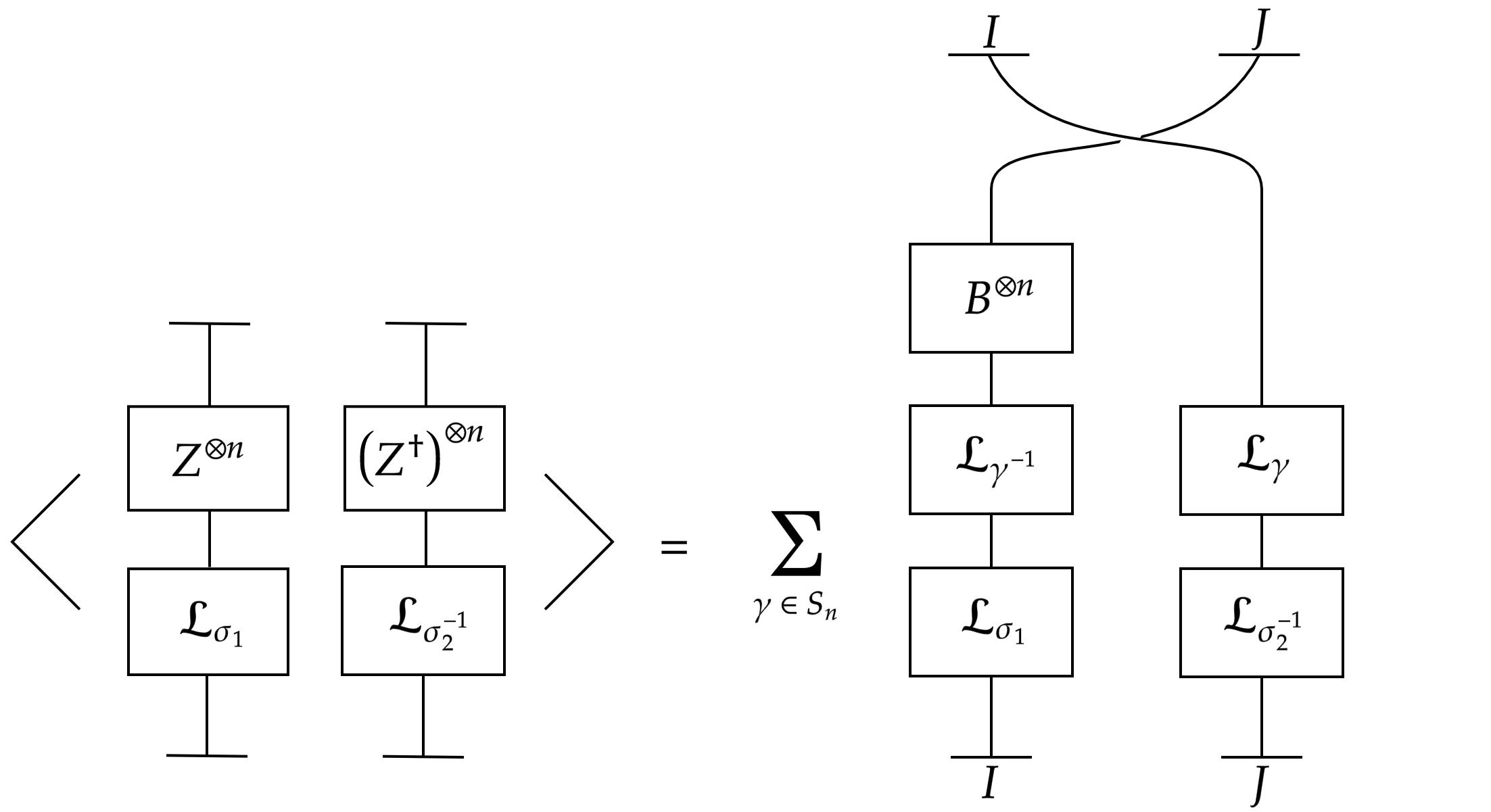}
\end{center}
\caption{The correlator diagrams where each box is an operator connected by branches. The right side diagram explicitly shows the indices, indicating that horizontal lines are to be thought of as traces connecting the bottom and top. The diagram is read by following branches downward and multiplying the operators encountered.}
\label{fig:two_perm_correlator_diagram_app}
\end{figure}
In total, including the sum over $\gamma$, the correlator of the one-matrix, coupling witness field model is therefore

\begin{align}
\label{eqn:app_dia_second_branch}
{}&\langle \mathcal{O}_{\sigma_{1}}(Z) \left(\mathcal{O}_{\sigma_{2}}(Z)\right)^{\dagger} \rangle \nonumber
\\
{}&= \sum_{\gamma \in S_{n}} \sum_{I,J,K,L,P} B^{j_{1}}_{k_{1}} \dots B^{j_{n}}_{k_{n}} (\mathcal{L}_{\gamma^{-1}})^{k_{1}, \dots, k_{n}}_{l_{1}, \dots, l_{n}} (\mathcal{L}_{\sigma_{1}})^{l_{1}, \dots, l_{n}}_{i_{1}, \dots, i_{n}} (\mathcal{L}_{\gamma})^{i_{1}, \dots, i_{n}}_{p_{1}, \dots, p_{n}}(\mathcal{L}_{\sigma_{2}^{-1}})^{p_{1}, \dots, p_{n}}_{j_{1}, \dots, j_{n}}
\\
&= \sum_{\gamma \in S_{n}}
\text{Tr}_{V_{N}^{\otimes n}}\left(B^{\otimes n} \mathcal{L}_{\gamma^{-1}} \mathcal{L}_{\sigma_{1}} \mathcal{L}_{\gamma} \mathcal{L}_{\sigma_{2}^{-1}} \right)
\\
&= \sum_{\gamma \in S_{n}} \text{Tr}_{V_{N}^{\otimes n}} \left(B^{\otimes n} \mathcal{L}_{\gamma^{-1} \sigma_{1} \gamma \sigma_{2}^{-1}} \right) \label{eqn:app_dia_used_homo_prop}
\\
&= \sum_{j_{1}, \dots j_{n}} \sum_{\gamma \in S_{n}} \bra{e^{j_{1}} \dots e^{j_{n}}} B^{\otimes n} \mathcal{L}_{\gamma^{-1} \sigma_{1} \gamma \sigma_{2}^{-1}} \ket{e_{j_{1}} \dots e_{j_{n}}}
\\
&= \sum_{j_{1}, \dots, j_{n}} \sum_{\gamma \in S_{n}} \bra{e^{j_{1}} \dots e^{j_{n}}} B^{\otimes n} \ket{e_{j_{\gamma^{-1} \sigma_{1} \gamma \sigma_{2}^{-1}(1)}} \dots e_{j_{\gamma^{-1} \sigma_{1} \gamma \sigma_{2}^{-1}(n)}}} \label{eqn:app_dia_used_operator_prop}
\\
&= \sum_{j_{1}, \dots, j_{n}} \sum_{\gamma \in S_{n}} B^{j_{1}}_{j_{\gamma^{-1} \sigma_{1} \gamma \sigma_{2}^{-1}(1)}} \dots B^{j_{n}}_{j_{\gamma^{-1} \sigma_{1} \gamma \sigma_{2}^{-1}(n)}} \,.
\end{align}
where \eqref{eqn:app_dia_homo_prop} was used in equation \eqref{eqn:app_dia_used_homo_prop}, and \eqref{eqn:app_dia_perm_operator} used in \eqref{eqn:app_dia_used_operator_prop}. This is the previously identified result of equation \eqref{eqn:bmf_cycle_struc_B_mat}. Note that the right hand side diagram introduces a swap in the indices. The origin of this swap stems from deriving the basic two-point function, as seen in Appendix \ref{app:background_matrix_field_correlator}, where indices $j,l$ on $\langle Z^{i}_{j} (Z^{\dagger})^{k}_{l} \rangle = \delta^{i}_{l} (A^{-1})^{k}_{j}$ switch positions in the outcome. 

The tensor correlator diagram of Figure \ref{fig:tensor_diagram_3} from \S\ref{ss:background_tensor_fields}, introduces tensor variables ($Z, \bar{Z}$) and witnesses ($\Lambda, \bar{\Lambda}$) illustrated by triangles. To help understand this, we explain the diagrammatic representation of a single tensor GIO, $\mathcal{O}_{\sigma_{1},\sigma_{2}, \sigma_{3}}(\bar{\Lambda}, Z)$, in Figure \ref{fig:app_tens_description_diagram}.
\begin{figure}[htb!]
\begin{center}
\includegraphics[width=9.5cm, height=5.5cm]{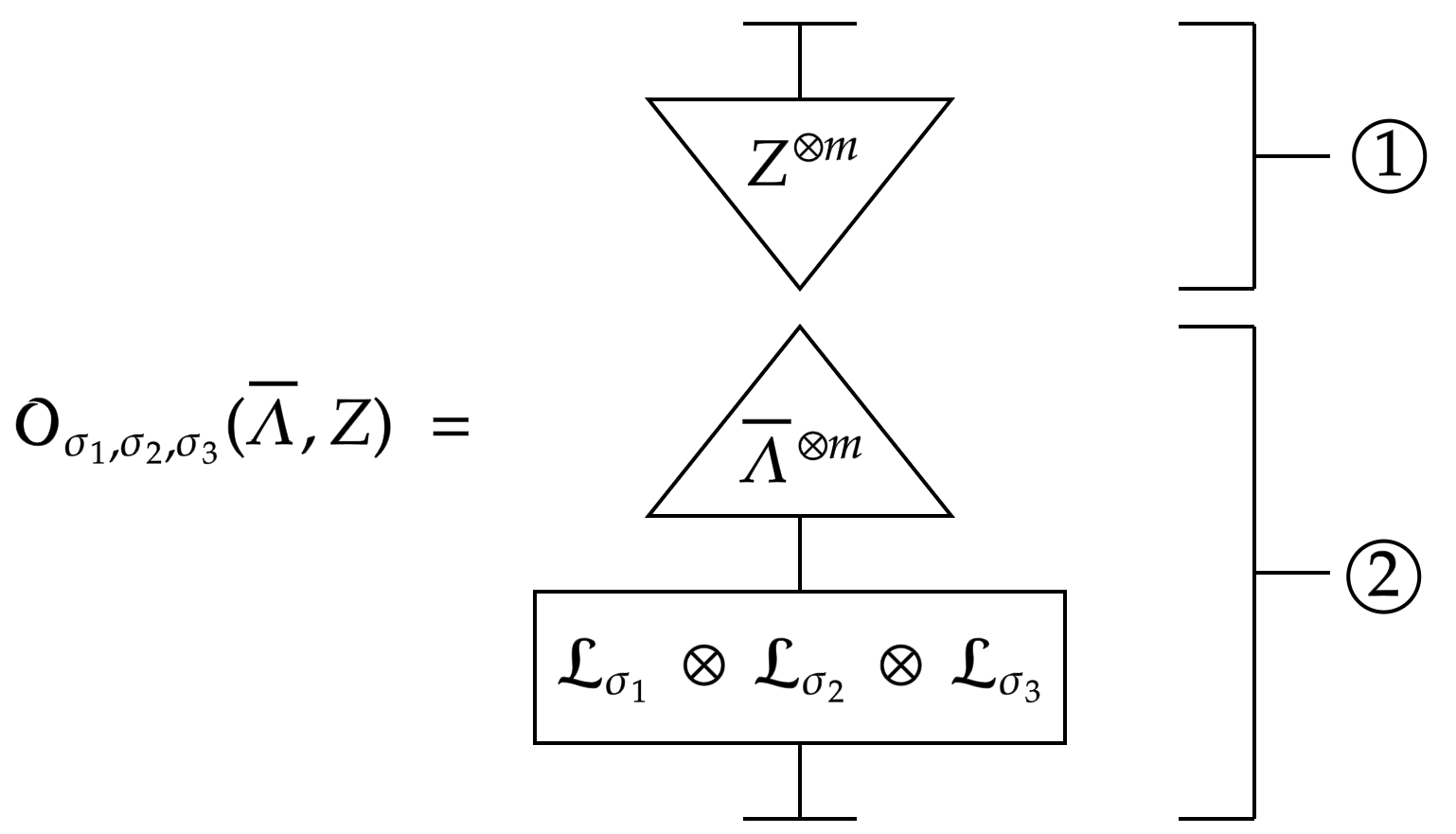}
\end{center}
\caption{A diagram representing a single gauge invariant tensor operator $\mathcal{O}_{\sigma_{1},\sigma_{2},\sigma_{3}}(\bar{\Lambda},Z)$, with quantum fields $Z$ and classical witnesses $\bar{\Lambda}$. It can be decomposed into two parts: \textcircled{\raisebox{-0.9pt}{1}} where $Z^{\otimes m}$ acts on a bra state and \textcircled{\raisebox{-0.9pt}{2}} where the linear permutation operators act first on a ket state, followed by $\bar{\Lambda}^{\otimes m}$.}
\label{fig:app_tens_description_diagram}
\end{figure}
Mathematically we calculate \textcircled{\raisebox{-0.9pt}{1}} and \textcircled{\raisebox{-0.9pt}{2}} as follows
\begin{align}
\text{\textcircled{\raisebox{-0.9pt}{1}}}
{}&=  \bra{(e^{i_{1}} \otimes e^{j_{1}} \otimes e^{k_{1}}) \otimes \dots \otimes (e^{i_{m}} \otimes e^{j_{m}} \otimes e^{k_{m}} )} Z^{\otimes m}
\\
{}&= Z^{i_{1} j_{1} k_{1}} \dots Z^{i_{m} j_{m} k_{m}}
\end{align}
\begin{align}
\text{\textcircled{\raisebox{-0.9pt}{2}}}
{}&= \bar{\Lambda}^{\otimes m} (\mathcal{L}_{\sigma_{1}} \otimes \mathcal{L}_{\sigma_{2}} \otimes \mathcal{L}_{\sigma_{3}})\ket{(e_{i_{1}} \otimes e_{j_{1}} \otimes e_{k_{1}}) \otimes \dots \otimes (e_{i_{m}} \otimes e_{j_{m}} \otimes e_{k_{m}} )}
\\
{}&= \bar{\Lambda}^{\otimes m}\ket{(e_{i_{\sigma_{1}(1)}} \otimes e_{j_{\sigma_{2}(1)}} \otimes e_{k_{\sigma_{3}(1)}}) \otimes \dots \otimes (e_{i_{\sigma_{1}(m)}} \otimes e_{j_{\sigma_{2}(m)}} \otimes e_{k_{\sigma_{3}(m)}} )}
\\
&{}=  \bar{\Lambda}_{i_{\sigma_{1}(1)} j_{\sigma_{2}(1)} k_{\sigma_{3}(1)}} \dots \bar{\Lambda}_{i_{\sigma_{1}(m)} j_{\sigma_{2}(m)} k_{\sigma_{3}(m)}}
\end{align}
Combining these two pieces, the final GIO evaluation is then 
\begin{align}
\mathcal{O}_{\sigma_{1},\sigma_{2},\sigma_{3}}(\bar{\Lambda},Z) {}&= \text{\textcircled{\raisebox{-0.9pt}{1}}} \times \text{\textcircled{\raisebox{-0.9pt}{2}}}
\\
{}&= \bar{\Lambda}_{i_{\sigma_{1}(1)} j_{\sigma_{2}(1)} k_{\sigma_{3}(1)}} \dots \bar{\Lambda}_{i_{\sigma_{1}(m)} j_{\sigma_{2}(m)} k_{\sigma_{3}(m)}} Z^{i_{1} j_{1} k_{1}} \dots Z^{i_{m} j_{m} k_{m}} \,,
\end{align}
which matches the definition in equation  \eqref{eqn:tenmod_observable_d_3}. The triangles are used to help emphasise that the tensor field operators are maps from the tensor space to the complex numbers $Z^{\otimes m}: V^{\otimes m} \to \mathbb{C}$. 

\section{Delta function sums}
\label{app:delta_function_sums}
Equation \eqref{eqn:tenmod_two_point_func_pt4} can be obtained from equation \eqref{eqn:tenmod_two_point_func_pt3} by
first exchanging the $\gamma \in S_{m}$. Begin with 
\begin{align}
\sum_{\gamma_{1}, \gamma_{2}, \alpha \in S_{m}} \delta(\gamma_{1} \gamma_{2} \alpha) \,.
\end{align}
Setting $\gamma = \gamma_{1}\gamma_{2}$ and noting that the sum over $\gamma_{1}$ runs over all permutations in $S_{m}$, and so too does the sum over $\gamma_{1} = \gamma \gamma_{2}^{-1}$, then we can  substitute $\sum_{\gamma_{1}}$ for $\sum_{\gamma}$
\begin{align}
\sum_{\gamma_{1}, \gamma_{2}, \alpha} \delta(\gamma_{1} \gamma_{2} \alpha) = \sum_{\gamma, \gamma_{2}, \alpha} \delta(\gamma \alpha) = m! \sum_{\gamma,\alpha} \delta(\gamma \alpha) \implies \sum_{\gamma, \alpha } \delta(\gamma \alpha) = \frac{1}{m!} \sum_{\gamma_{1}, \gamma_{2}, \alpha} \delta(\gamma_{1}\gamma_{2} \alpha) \,.
\end{align}
Therefore 
\begin{align}
\begin{split}
\label{eqn:app_delta_func1}
\langle \mathcal{O}_{\sigma_{1}, \sigma_{2}, \sigma_{3}}(\bar{\Lambda}, Z) \overline{\mathcal{O}_{\tau_{1}, \tau_{2}, \tau_{3}}(\bar{\Lambda}, Z)} \rangle  
&= \sum_{ \alpha_{l}, \gamma \in S_{m}} \mathcal{O}_{\alpha_{1}, \alpha_{2}, \alpha_{3}}(\bar{\Lambda},\Lambda)
\prod_{l=1}^{3} \delta(\alpha_{l}^{-1} \sigma_{l} \gamma \tau_{l}^{-1} ) \\
&= \frac{1}{m!}  \sum_{ \alpha_{l}, \gamma_{1}, \gamma_{2} \in S_{m}} \mathcal{O}_{\alpha_{1}, \alpha_{2}, \alpha_{3}}(\bar{\Lambda},\Lambda)
\prod_{l=1}^{3} \delta(\alpha_{l}^{-1} \sigma_{l} \gamma_{1} \gamma_{2} \tau_{l}^{-1} )  \,.
\end{split}
\end{align}
We then use the invariance of the operator $ \mathcal{O}_{\alpha_{1}, \alpha_{2}, \alpha_{3}}(\bar{\Lambda}, \Lambda)$ to insert permutations $\mu_{1}$ and $\mu_{2}$,
\begin{align}
\frac{1}{m!} &\sum_{ \alpha_{l}, \gamma_{1}, \gamma_{2} \in S_{m}} \mathcal{O}_{\alpha_{1}, \alpha_{2}, \alpha_{3}}(\bar{\Lambda},\Lambda) \prod_{l=1}^{3} \delta(\alpha_{l}^{-1} \sigma_{l} \gamma_{1} \gamma_{2} \tau_{l}^{-1}) \nonumber \\
&=\frac{1}{m!}  \sum_{ \alpha_{l}, \gamma_{1}, \gamma_{2} \in S_{m}} \mathcal{O}_{\mu_{1}\alpha_{1}\mu_{2}, \mu_{1}\alpha_{2}\mu_{2},\mu_{1} \alpha_{3}\mu_{2}}(\bar{\Lambda},\Lambda) \prod_{l=1}^{3} \delta(\alpha_{l}^{-1} \sigma_{l} \gamma_{1} \gamma_{2} \tau_{l}^{-1}) \label{eqn:appdel_op_invariance}\\
&= \frac{1}{m!} \sum_{ \gamma_{1}, \gamma_{2} \in S_{m}} \sum_{\mu_{1}^{-1}\tilde{\alpha}_{l}\mu_{2}^{-1}  \in S_{m}} \mathcal{O}_{\tilde{\alpha}_{1}, \tilde{\alpha}_{2},\tilde{\alpha}_{3}}(\bar{\Lambda},\Lambda) \prod_{l=1}^{3} \delta(\mu_{2}\tilde{\alpha}_{l}^{-1} \mu_{1} \sigma_{l} \gamma_{1} \gamma_{2} \tau_{l}^{-1}) \label{eqn:appdel_relabel}\\  
&= \frac{1}{m!} \sum_{ \gamma_{1}, \gamma_{2} \in S_{m}}  \sum_{\tilde{\alpha}_{l} \in S_{m}} \mathcal{O}_{\tilde{\alpha}_{1}, \tilde{\alpha}_{2},\tilde{\alpha}_{3}}(\bar{\Lambda},\Lambda) \prod_{l=1}^{3} \delta(\mu_{2}\tilde{\alpha}_{l}^{-1} \mu_{1} \sigma_{l} \gamma_{1} \gamma_{2} \tau_{l}^{-1}) \label{eqn:appdel_relabel_sum}
\\
&=\frac{1}{(m!)^3} \sum_{\alpha_{l}, \gamma_{1}, \gamma_{2} \mu_{1}, \mu_{2},  \in S_{m}} \mathcal{O}_{\alpha_{1}, \alpha_{2}, \alpha_{3}}(\bar{\Lambda},\Lambda) \prod_{l=1}^{3} \delta(\mu_{2} \alpha_{l}^{-1} \mu_{1} \sigma_{l} \gamma_{1} \gamma_{2} \tau_{l}^{-1}) \label{eqn:appdel_mu_sums}\,.
\end{align}

In \eqref{eqn:appdel_op_invariance} the equivalence property $\mathcal{O}_{\alpha_{1}, \alpha_{2}, \alpha_{3}}(\bar{\Lambda},\Lambda) = \mathcal{O}_{\mu_{1}\alpha_{1}\mu_{2}, \mu_{1}\alpha_{2}\mu_{2},\mu_{1} \alpha_{3}\mu_{2}}(\bar{\Lambda},\Lambda)$ is used. Equation \eqref{eqn:appdel_relabel} relabels $\mu_{1} \alpha_{l} \mu_{2} \to \tilde{\alpha}$, while \eqref{eqn:appdel_relabel_sum} relabels the sum since both cover all elements in $S_{m}$. Finally 
\eqref{eqn:appdel_mu_sums} introduces sums over $\mu_{1}$ and $\mu_{2}$ and includes the appropriate normalisation factors. 

\section{Wedderburn-Artin basis properties for general matrix models}
\label{app:redone_wa_q_basis}
Take the multi-matrix model tensor space decomposition 
\begin{equation}
\label{eqn:app_gen_mmm_spb_vector_decomp}
V_{R}^{(S_{\mathbf{m}})} = \bigoplus_{\mathbf{r}} V_{R_{1}}^{(S_{m_{1}})} \otimes V_{R_{2}}^{(S_{m_{2}})} \otimes \dots \otimes V_{R_{l}}^{(S_{m_{l}})} \otimes V_{\mathbf{r}}^{R}.
\end{equation}
Here, $R$ is the representation of $S_{\mathbf{m}}$ with ${\bf{m}} = m_{1} + m_{2} + \dots + m_{l}$ and $\mathbf{r}=(R_{1}, R_{2}, \dots, R_{l})$ which denotes the set of individual representations $R_{i}$, for $i = 1, \dots, l$. The restricted character is
\begin{equation}
\label{eqn:app_gen_mmm_rsc_in_branch_coef}
\chi^{R}_{\mathbf{r},\mu,\nu}(\sigma) = \sum_{\mathbf{t}} \sum_{i,j} D^{R}_{ij}(\sigma) B^{R; j}_{\mathbf{r},\mu; \mathbf{t}} B^{R; i}_{\mathbf{r},\nu;\mathbf{t}} 
\end{equation}
where $\mathbf{t} = (t_{1}, t_{2}, \dots, t_{l})$ are a set of labels used to denote the states of each individual representation space e.g. $V_{R_{1}}^{(S_{m_{1}})}$ has states labelled by $t_{1}$. $\mu,\nu$ are multiplicity indices. The Wedderburn-Artin basis (also referred to as $Q$-basis/Fourier basis) element
\begin{align}
Q^{R}_{\mathbf{r},\mu,\nu} = \frac{d_{R}}{\mathbf{m}!} \sum_{\sigma \in S_{\mathbf{m}}} \chi^{R}_{\mathbf{r},\mu,\nu} (\sigma) \sigma^{-1} \,,
\end{align}
multiplies like a matrix in the multiplicity indices and is subsequently proven.
\begin{lemma}
\label{lem:appendix_for_q_basis_multiplication}
\begin{equation}
Q^{R}_{\mathbf{r},\mu_{1},\nu_{1}} Q^{S}_{\mathbf{s},\mu_{2},\nu_{2}}  =  \delta^{RS}\delta_{\mathbf{r}\mathbf{s}}  \delta_{\nu_{1}\mu_{2}}  Q^{R}_{\mathbf{r},\mu_{1},\nu_{2}}
\end{equation}
\end{lemma}

\begin{proof}
\begin{align}
\label{eqn:app_gen_mmm_q_multiplication}
Q^{R}_{\mathbf{r},\mu_{1},\nu_{1}} Q^{S}_{\mathbf{s},\mu_{2},\nu_{2}}
&{}= \frac{d_{R}d_{S}}{(\mathbf{m}!)^2} \sum_{\sigma, \tau  \in S_{\mathbf{m}}} \chi^{R}_{\mathbf{r},\mu_{1},\nu_{1}} (\sigma) \chi^{S}_{\mathbf{s},\mu_{2},\nu_{2}} (\tau) \sigma^{-1} \tau^{-1}
\\
&{}= \frac{d_{R}d_{S}}{(\mathbf{m}!)^2} \sum_{\alpha \in S_{\mathbf{m}}} \left[\sum_{\tau \in S_{\mathbf{m}}}  \chi^{R}_{\mathbf{r},\mu_{1},\nu_{1}} (\tau^{-1} \alpha) \chi^{S}_{\mathbf{s},\mu_{2},\nu_{2}}  (\tau) \right] \alpha^{-1} \,, \label{eqn:app_gen_mmm_q_multiplication2}
\end{align}
where in the above, $\sigma \to  \tau^{-1} \alpha$ and the sum symbol is exchanged accordingly. The square bracket term can be evaluated using the definition of \eqref{eqn:app_gen_mmm_rsc_in_branch_coef}
\begin{align}
&{}\sum_{\tau \in S_{\mathbf{m}}}  \chi^{R}_{\mathbf{r},\mu_{1},\nu_{1}} (\tau^{-1} \alpha) \chi^{S}_{\mathbf{s},\mu_{2},\nu_{2}} (\tau) \nonumber 
\\
{}&= \sum_{\tau \in S_{\mathbf{m}}} \sum_{\mathbf{t},\mathbf{q}} \sum_{i,j,k,l} D^{R}_{ij}(\tau^{-1} \alpha) B^{R; j}_{\mathbf{r},\mu_{1}; \mathbf{t}} B^{R; i}_{\mathbf{r},\nu_{1};\mathbf{t}}  D^{S}_{kl}(\tau) B^{S; l}_{\mathbf{s},\mu_{2}; \mathbf{q}} B^{S; k}_{\mathbf{s},\nu_{2};\mathbf{q}}  
\\
{}&=  \sum_{\mathbf{t},\mathbf{q}} \sum_{i,j,k,l}  B^{R; j}_{\mathbf{r},\mu_{1}; \mathbf{t}} B^{R; i}_{\mathbf{r},\nu_{1};\mathbf{t}}   B^{S; l}_{\mathbf{s},\mu_{2}; \mathbf{q}} B^{S; k}_{\mathbf{s},\nu_{2};\mathbf{q}} \left[ \sum_{\tau \in S_{\mathbf{m}}}  D^{R}_{ij}(\tau^{-1} \alpha)D^{S}_{kl}(\tau) \right] \label{eqn:app_gen_mmm_spb_orthog_rel}
\end{align}
Using the following orthogonality relation
\begin{align}
\sum_{\tau \in S_{\mathbf{m}}}  D^{R}_{ij}(\tau^{-1} \alpha)D^{S}_{kl}(\tau) 
&{}=  \sum_{a}  D^{R}_{aj}(\alpha)  \sum_{\tau \in S_{\mathbf{m}}}  D^{S}_{kl}(\tau) D^{R}_{ia}(\tau^{-1})  \label{eqn:app_gen_mmm_spb_orthog_relation_ham_form}
\\
&{} = \sum_{a}  D^{R}_{aj}(\alpha)  \frac{\mathbf{m}!}{d_{R}}\delta^{RS} \delta_{ka}\delta_{li} 
\\
&{} = \frac{\mathbf{m}!}{d_{R}}\delta^{RS} D^{R}_{kj} (\alpha) \delta_{li} \,, \label{eqn:app_gen_mmm_spb_orthog_relation_step}
\end{align}
in equation \eqref{eqn:app_gen_mmm_spb_orthog_rel}, we achieve
\begin{align}
&{} \sum_{\tau \in S_{\mathbf{m}}}  \chi^{R}_{\mathbf{r},\mu_{1},\nu_{1}} (\tau^{-1} \alpha) \chi^{S}_{\mathbf{s},\mu_{2},\nu_{2}} (\tau)  \nonumber 
\\
&{}= \frac{\mathbf{m}!}{d_{R}}\delta^{RS}  \sum_{\mathbf{t},\mathbf{q}} \sum_{j,k}  \left[ \sum_{i,l} \delta_{li}  B^{R; i}_{\mathbf{r},\nu_{1};\mathbf{t}}  B^{S; l}_{\mathbf{s},\mu_{2}; \mathbf{q}} \right] D^{R}_{kj}(\alpha)  B^{R; j}_{\mathbf{r},\mu_{1}; \mathbf{t}} B^{S; k}_{\mathbf{s},\nu_{2};\mathbf{q}}  \label{eqn:app_gen_mmm_spb_bracn_contract}
\\
&{} = \frac{\mathbf{m}!}{d_{R}}\delta^{RS}   \sum_{j,k} \sum_{\mathbf{t},\mathbf{q}} \left[ \delta_{\mathbf{r}\mathbf{s}}\delta_{\mathbf{t}\mathbf{q}} \delta_{\nu_{1} \mu_{2}} \right] D^{R}_{kj}(\alpha)  B^{R; j}_{\mathbf{r},\mu_{1}; \mathbf{t}} B^{S; k}_{\mathbf{s},\nu_{2};\mathbf{q}} 
\\
&{} = \frac{\mathbf{m}!}{d_{R}}\delta^{RS}\delta_{\mathbf{r}\mathbf{s}}  \delta_{\nu_{1} \mu_{2}}   \sum_{j,k} \sum_{\mathbf{t}}  D^{R}_{kj}(\alpha)  B^{R; j}_{\mathbf{r},\mu_{1}; \mathbf{t}} B^{S; k}_{\mathbf{s},\nu_{2};\mathbf{t}} 
\\
&{}=   \frac{\mathbf{m}!}{d_{R}}\delta^{RS}\delta_{\mathbf{r}\mathbf{s}}  \delta_{\nu_{1} \mu_{2}} \chi^{R}_{\mathbf{r},\mu_{1},\nu_{2}}(\alpha) \,.
\end{align}
Therefore, substituting this result back into equation \eqref{eqn:app_gen_mmm_q_multiplication2}:
\begin{align}
Q^{R}_{\mathbf{r}, \mu_{1}, \nu_{1}} Q^{S}_{\mathbf{s}, \mu_{2}, \nu_{2}} 
&{}=  \frac{d_{R}d_{S}}{(\mathbf{m}!)^2} \sum_{\alpha \in S_{\mathbf{m}}} \left[\sum_{\tau \in S_{\mathbf{m}}}  \chi^{R}_{\mathbf{r},\mu_{1},\nu_{1}} (\tau^{-1} \alpha) \chi^{S}_{\mathbf{s},\mu_{2},\nu_{2}} (\tau) \right] \alpha^{-1}
\\
&{}=  \delta^{RS} \delta_{\mathbf{r}\mathbf{s}}  \delta_{\nu_{1} \mu_{2}} \left( \frac{d_{R}}{\mathbf{m!}} \sum_{\alpha \in S_{\mathbf{m}}} \chi^{R}_{\mathbf{r},\mu_{1},\nu_{2}}(\alpha) \alpha^{-1} \right)
\\
&{}=  \delta^{RS} \delta_{\mathbf{r}\mathbf{s}} \delta_{\nu_{1},\mu_{2}}  Q^{R}_{\mathbf{r}, \mu_{1},\nu_{2}}
\end{align}
\end{proof}
Note that specialising to $l=2$ in \eqref{eqn:app_gen_mmm_spb_vector_decomp} leads to the required rules for the two-matrix calculation in \S \ref{ss:two_matrix_schur_basis}. Namely, the decomposition becomes $V^{R}_{S_{m+n}} = V^{R_{1}}_{S_{m}} \otimes V^{R_{2}}_{S_{n}} \otimes V^{R}_{R_{1},R_{2}}$ where we label $m_{1} = m, m_{2} = n$ such that $\mathbf{m} = m+n$. This specialisation also implies that notationally, $\delta_{\mathbf{r}\mathbf{s}} \equiv \delta_{R_{1}S_{1}}\delta_{R_{2}S_{2}}$. Appendix B.2 of \cite{Pasukonis:2013ts} provides additional details and proofs on restricted quiver characters. Additionally, the Hermitian conjugate of the $Q$-basis element switches its multiplicity indices, such that  $\left(Q^{R}_{\mathbf{r}, \mu_{1},\nu_{1}}\right)^{\dagger} = Q^{R}_{\mathbf{r},\nu_{1},\mu_{1}}$.
\begin{proof}
\begin{align}
\left(Q^{R}_{\mathbf{r}, \mu_{1},\nu_{1}}\right)^{\dagger} 
&= 
\frac{d_{R}}{\mathbf{m}!} \sum_{\sigma \in S_{\mathbf{m}}} \left(\chi^{R}_{\mathbf{r}, \mu_{1},\nu_{1}} (\sigma)\right)^{\dagger} \left(\sigma^{-1} \right)^{\dagger}
\\
&= \frac{d_{R}}{\mathbf{m}!} \sum_{\sigma \in S_{\mathbf{m}}} \chi^{R}_{\mathbf{r}, \mu_{1},\nu_{1}} (\sigma) \sigma
\\
&= \frac{d_{R}}{\mathbf{m}!} \sum_{\sigma \in S_{\mathbf{m}}} \sum_{\mathbf{t}} \sum_{i,j} D^{R}_{ij}(\sigma) B^{R; j}_{\mathbf{r},\mu_{1}; \mathbf{t}} B^{R; i}_{\mathbf{r},\nu_{1};\mathbf{t}}  \sigma 
\\
&= \frac{d_{R}}{\mathbf{m}!} \sum_{\sigma \in S_{\mathbf{m}}} \sum_{\mathbf{t}} \sum_{i,j} D^{R}_{ji}(\sigma^{-1}) B^{R; i}_{\mathbf{r},\nu_{1};\mathbf{t}} B^{R; j}_{\mathbf{r},\mu_{1}; \mathbf{t}}  \sigma 
\\
&= \frac{d_{R}}{\mathbf{m}!} \sum_{\sigma \in S_{\mathbf{m}}} \chi^{R}_{\mathbf{r}, \nu_{1},\mu_{1}} (\sigma^{-1}) \sigma
\\
&= \frac{d_{R}}{\mathbf{m}!} \sum_{\tilde{\sigma} \in S_{\mathbf{m}}} \chi^{R}_{\mathbf{r}, \nu_{1},\mu_{1}} (\tilde{\sigma}) \tilde{\sigma}^{-1}
\\
&= Q^{R}_{\mathbf{r},\nu_{1},\mu_{1}} \,.
\end{align}
\end{proof}

\printbibliography
\end{document}